\documentclass[11pt]{article}

\usepackage{amsfonts}
\usepackage{alltt}
\usepackage{amssymb,verbatim,ifthen}
\usepackage{dsfont}
\usepackage[all]{xy}
\usepackage{color}
\usepackage{graphicx}
\usepackage{amsmath}
\usepackage{amsthm}
\usepackage{mathabx}
\usepackage{pdfsync}
\usepackage{listliketab}
\usepackage{natbib}
\usepackage{thm-restate}
\usepackage{subcaption} 
\usepackage{hyperref}
\usepackage{textcomp}
\usepackage{tikz}
\usepackage[normalem]{ulem}
\usepackage{pst-node}

\usepackage{multirow,array}
\usepackage{pgfplots}
\pgfplotsset{compat=1.7}
\usetikzlibrary{intersections}
\usetikzlibrary{calc}
\usetikzlibrary{patterns}

\usepackage{enumitem}

\usepackage{comment}

\usepackage[textwidth=30mm]{todonotes}

 \def\split{true}
 \ifthenelse{\equal{\split}{true}} {
 \textwidth = 440pt \oddsidemargin = 2pt \evensidemargin = 2pt
 }

\newtheorem{theorem}{Theorem}
\newtheorem*{theorem*}{Theorem}

\newtheorem{lemma}{Lemma}

\newtheorem{definition}{Definition}

\newtheorem{proposition}{Proposition}

\newtheorem{example}{Example}
\newtheorem{corollary}{Corollary}

\newtheorem{remark}{Remark}

\theoremstyle{empty}

\renewcommand\thmcontinues[1]{Continued}

\def\1{{1\hskip-2.5pt{\rm l}}} 
\def\RR{\mathbb{R}} 
\def\E{\mathbb{E}} 

\makeatother

\begin{document}

\title{Weighted Average-convexity and Cooperative Games}

\author{Alexandre Skoda\thanks{Universit\'e de Paris I, Centre d'Economie de la Sorbonne,
106-112 Bd de l'H\^opital, 75013 Paris, France, \textsf{askoda@univ-paris1.fr}},
Xavier Venel\thanks{Dipartimento di Economia e Finanza, LUISS, Viale Romania 32, 00197 Rome, Italy,  
\textsf{xvenel@luiss.it}, https://orcid.org/0000-0003-1150-9139.
Xavier Venel is part of the Italian MIUR PRIN 2017 Project ALGADIMAR “Algorithms, Games, and Digital
Markets}
}

\maketitle

\thispagestyle{empty}

\begin{abstract}
\noindent
We generalize the notion of convexity and average-convexity to the notion of weighted average-convexity.
We show several results on the relation between weighted average-convexity and cooperative games.
First,
we prove that if a game is  weighted average-convex,
then the corresponding weighted Shapley value is in the core.
Second,
we exhibit necessary conditions for a communication TU-game to preserve the weighted average-convexity.
Finally,
we provide a complete characterization when the underlying graph is a priority decreasing tree.
\end{abstract}


\bigskip




\section{Introduction}

We consider a situation where a set of players $N = \lbrace 1, \ldots, n \rbrace$ can cooperate to generate some profit.
It is usually assumed that any coalition $S \subseteq N$ of players is feasible and 
that the profit generated by the players in $S$ can be freely distributed among them.
This kind of situation corresponds to a cooperative game with transferable utility.
One of the key question in this framework is to allocate to each player his share
of the profit generated by the grand coalition~$N$.
Any share of the total worth of $N$ between the $n$ players is called a solution of the game.

Lots of solution concepts have been introduced and they can be divided in two different families:
single-valued solutions and set-valued solutions.
The best-known singled-valued solution in cooperative game theory is the Shapley value introduced and characterized in \cite{Shapley53}.
The Shapley value allocates to each player its average marginal contribution and is always defined.
A typical set valued solution is the core (\cite{gillies1959solutions}).
The core is the set of allocations that are efficient and coalitionally rational,
two particularly desirable properties in the context of cooperative games.
An allocation is efficient if it fully distributes the worth of the grand coalition $N$
among the $n$ players.
Coalitional rationality requires that for every coalition $S$ the total payoff obtained by $S$ is at least equal to the profit 
that it can generate by itself. Hence no coalition can have an incentive to leave the grand coalition~$N$.
Unfortunately, the core of a game may be empty but some properties of the game can guarantee the non-emptiness of the core
and more specifically that the Shapley value lies in the core.

\cite{Shapley71} introduced the class of convex games.
In a convex game, the marginal contribution of a player is an increasing function with respect to the coalition size.
\cite{Shapley71}  showed that if a game is convex,
then the Shapley value belongs to the core and therefore the core is non-empty.
In fact, the convexity assumption is very strong and can be relaxed
to ensure that the Shapley value is a core element.
\cite{Inarra-Usategui-1993} introduced a weaker notion than convexity,
called average-convexity and proved that it is a sufficient condition for the Shapley value to be in the core.
In the average convexity,
one compares averages of marginal contributions instead of comparing directly marginal contributions.
More precisely,
for any coalition
the sum of the marginal contributions of the players belonging to this coalition
has to be less or equal to the sum of the marginal contributions of the same players in a larger coalition.
Independently,
\cite{Sprumont1990} established a result similar to the one of \cite{Inarra-Usategui-1993}
but in a different framework as he investigated population monotonic allocation schemes (PMAS).
In particular,
average convexity is also a sufficient condition for the existence of a PMAS.

By definition, the Shapley value is symmetric between the players.
But many situations like for example players representing groups of individuals
or players having different bargaining powers are naturally asymmetric.
\cite{Shapley1953a} extended the notion of Shapley value to weighted Shapley value by considering asymmetric positive weights on the players.
Even more general,
\cite{KalaiAndSamet1987} introduced the notion of weight system introducing different priorities on the players
and provided an axiomatic characterization of the weighted Shapley values.
Further results and characterizations of the weighted Shapley values were proven in \cite{Hart-Mas-Colell1989}.
\cite{Monderer-Samet1992} proved that if the game is convex,
then the weighted Shapley values are in the core (for any weight system).
Moreover,
they showed that for any game one can obtain all the elements of the core as weighted Shapley values by varying the weight system.

The first aim of this article is to establish given a weight system a weaker sufficient condition than convexity
to ensure that the weighted Shapley value belongs to the core.
To do so, we introduce the notion of weighted average-convexity generalizing the average-convexity of \cite{Inarra-Usategui-1993}.
We prove then that it is indeed a sufficient condition.

In many situations, the cooperation of players may be restricted
by some communication or social structures.
Then, the worths of coalitions have to be
modified to take these restrictions into account, leading to the introduction of restricted games.
\cite{Myerson77} introduced the class of graph restricted games, also known as communication games.
In his framework,
one considers an initial game and a communication graph. 
Players have to be connected in the underlying graph to be able to cooperate.
Formally, the value of a coalition in the communication game is the sum of the values in the initial game of connected subcoalitions.
The Myerson value of the game and the communication graph is the Shapley value of the communication game.
\cite{SlikkervandenNouweland2000} extended the Myerson value to weight systems.


A central question is the inheritance of properties from the original game (without the graph structure) to the associated communication game.
For a general overview of inheritance of properties for communication games we refer to \cite{SlikkerVandenNouweland2001}.
Our focus will be on convexity and its variants.
\cite{NouwelandandBorm1991} proved that there is inheritance of convexity if and only if the graph is cycle-complete.
\cite{Slikker1998} showed that there is inheritance of average convexity if and only if connected components of the graph are complete graphs or stars.

We investigate the question of inheritance of weighted average-convexity for communication games.
We exhibit necessary conditions on the underlying graph to preserve weighted average-convexity from any  game to the associated communication game.

When there is a unique priority, we prove that the conditions given in Slikkers's characterization
are also necessary using more general counter-examples taking into account the different weights of the players.
In this case we get the same characterization as \cite{Slikker1998} for inheritance of weighted average convexity.

When there are more than one priority,
we provide necessary conditions.
The first condition is for each subgraph associated with a priority level to be composed of stars or complete graphs.
The second is how two connected components lying in different priority levels are connected.
In the third condition, we analyze how several components can interact with each other.
Finally, we focus on the case where the graph is a tree and the priorities are decreasing along the tree.
We provide a full characterization of the class of priority decreasing trees preserving weighted-average convexity.
Informally, such a graph is a succession of stars along a path and with decreasing priorities from one star to another.

By combining these last results with our first one on the weighted Shapley value,
we obtain sufficient conditions on the game and on the graph 
to guarantee that the weighted Myerson value belongs to the core.


The outline of the article is the following.
Section~\ref{SectionDefinitions}
provides classical definitions in cooperative game theory, including several equivalent definitions of the Shapley value
and sufficient conditions for it to be in the core.
In Section~\ref{SectionWeightedShapleyAndweightedConvexity},
we recall the definition of the weighted Shapley value 
and introduce the notion of weighted average-convexity.
We prove that weighted average-convexity is a sufficient condition to ensure that
the weighted Shapley value belongs to the core
(Theorem~\ref{thIfweightedAverageConvexThenWeightedShapleyValueInCore}).
In the next sections we investigate the class of graphs ensuring inheritance of weighted average-convexity for communication games.
Section~\ref{SectionSingletonCase} provides
a complete characterization (Theorem~\ref{characterization_singleton}) of this class in the case of weight systems with a unique priority.
In Section~\ref{SectionWeightedConvexityAndCommunicationGameCaseWithSeveralPriorities},
we establish necessary conditions when weight systems have several priorities
(Propositions~\ref{Levelk-StarOrCompleteSubgraphs}, \ref{PlayerConnectedToAHigherLevel}, and \ref{propC1andC2LinkedToCOfHigherLevelThenSingletons}).
Finally,
we provide a complete characterization (Theorem~\ref{characterization_hierarchy}) for the class of priority decreasing trees.

\section{TU-games and Classical Shapley Value}
\label{SectionDefinitions}

\subsection{TU-Games}


A \textit{$TU$-game} is defined by a set of players $N$ and a characteristic function $v$ from $2^N$ to $\RR$ that satisfies $v(\emptyset)=0$.
An \textit{allocation} is a vector $x\in \RR^N$ that represents the respective payoff of each player.
It is said to be \textit{efficient} if
\[
\sum_{i\in N} x_i =v(N).
\]
and \textit{individually rational} if every player receives at least the value that he could guarantee alone
\[
\forall i\in N,\ x_i \geq v\left(\{i\}\right).
\]
This last property can be extended to coalitions.
An allocation is \textit{coalitionally rational} if for every coalition 
the sum of payoffs of its members is at least the value of the coalition
\[
\forall S \subseteq N,\ \sum_{i \in S} x_i \geq v\left(S\right).
\]
The \textit{core} is the set of efficient and coalitionally rational allocations. Formally,
\[
\mathcal{C}(v)=\left\{ x \in \RR^N,\ \sum_{i \in N} x_i= v(N),\ \sum_{i\in S} x_i \geq v(S), \forall S\subset N\right\}.
\]
A player $i$ in $N$ is a \textit{null player} if $v(S \cup \lbrace i \rbrace) = v(S)$ for every coalition $S \subseteq N \setminus \lbrace i \rbrace$.\\
A game $(N,v)$ is \textit{superadditive} if, for all $A,B \in 2^{N}$ such that $A \cap B = \emptyset$, $v(A \cup B) \geq v(A) + v(B)$.\\
For any given subset $\emptyset \not= S \subseteq N$,
the unanimity game $(N, u_{S})$ is defined by
\[
u_{S}(A) =
\left \lbrace
\begin{array}{cl}
1 &  \textrm{ if } A \supseteq S,\\
0 & \textrm{ otherwise.}
\end{array}
\right.
\]

Let us consider a game $(N,v)$. For arbitrary subsets $A$ and $B$ of $N$,
we define the value
\[
\Delta v(A,B):=v(A\cup B)+v(A\cap B)-v(A)-v(B).
\]
A game $(N,v)$ is \textit{convex} if its characteristic function $v$ is supermodular,
\emph{i.e.},
$\Delta v(A,B)\geq 0$ for all $A,B \in 2^{N}$.

\cite{Shapley53} proved that
every cooperative game $(N, v)$ can be written as a unique linear combination of unanimity games,
$v = \sum_{S \subseteq N} \lambda_{S}(v)u_{S}$,
where $\lambda_{\emptyset}(v) = 0$,
and for any $S \not= \emptyset$ the coefficients $\lambda_S(v) \in \mathbb{R}$ are
given by $\lambda_S(v) = \sum_{T \subseteq S} (-1)^{s-t} v(T)$.
We refer to these coefficients $\lambda_S(v)$ as the \textit{unanimity coefficients}\footnote{The $\lambda_S$'s are also called Harsanyi dividends in cooperative game theory \citep{Harsanyi1963}.} of $v$.

\subsection{Shapley value}

The Shapley value of $(N,v)$ is a solution concept. In terms of the unanimity coefficients the Shapley value is given by
\[
\Phi_{i}(v) = \sum_{S \subseteq N :\, i \in S} \frac{1}{s} \lambda_S(v),
\]
for all $i \in N$. It is equivalent to say that the Shapley value associated to a game $(N,v)$ is a linear function and that the allocation associated to the unanimity game of set $S$ is 
\[
x_i=
\begin{cases}
\frac{1}{s}, &\text{ if } i\in S,\\
0 &\text{ otherwise}.
\end{cases}
\]
Hence, the agents in $S$ share the utility of $1$ in equal shares.\\

Another definition is based on the following story.
Consider that all the players are in the room. At every stage, one player is chosen uniformly among the remaining players and leaves the room. When leaving, he obtains as payoff its marginal contribution between when he was present and when he was not anymore. The Shapley value is its expected payoff. Formally, let $\mathcal{L}$ be the set of ordered sequences, then define $\mathcal{U}$ to be the uniform distribution over $\mathcal{L}$. Then 
\[
\Phi_{i}(v) = \E_{\mathcal{U}} \left(v(\{L \geq i\})-v(\{L > i\}) \right),
\]
where $\{L \geq i\}$ (resp. $\{L >i\}$) is the set of agents (resp. strictly) after $i$ in the list $L$, i.e. $\{j\in N, L(j)\geq L(i) \}$ (resp. $\{j\in N, L(j)> L(i) \}$).\\

Let us justify this alternative definition. First, the right-hand side is indeed a linear function of $v$. Second, let us consider the unanimity game $(N,u_S)$. Given an order $L$, we define the pivot of $S$ the first occurence of an element of $S$. 
Formally, it is the element $i$ of $S$ such that $L(i)$ is minimal.  By definition of $u_S$, the marginal contribution of an agent is equal to $1$ if and only if this agent is a pivot. Hence
\[
\E_{\mathcal{U}} \left(u_S(\{L \geq i\})-u_S(\{L > i\}) \right)=\mathbb{P}_{\mathcal{U}} \left( i \text{ is pivot of } S\right)=\frac{1}{s}
\]
By symmetry of $\mathcal{U}$, the probability for $i$ to be the pivot of $S$ is $\frac{1}{s}$ giving the result.\\

\cite{Shapley53} showed that the Shapley value of a convex game belongs to the core. Looking for a weaker condition insuring that the Shapley value of a game lies in the core, \cite{Inarra-Usategui-1993} relaxed the convexity assumption by introducing the notion of average convexity.

\begin{definition}
The game $(N,v)$ is \emph{average convex} if for every $S \subset T \subseteq N$,
\[
\sum_{i\in S} \left(v(S)-v(S\setminus\{i\}) \right)\leq \sum_{i\in S} \left(v(T)-v(T \setminus \{i\})\right).
\]
\end{definition}

They obtained the following result.

\begin{proposition}[\cite{Inarra-Usategui-1993,Sprumont1990}\protect\footnote{\cite{Sprumont1990} established the same result.
Average convex games are called quasiconvex in Sprumont's work.}]
\label{PropIfAverageConvexThenShapleyValueInCore}
If the game is average convex then the Shapley value is in the core.
\end{proposition}

\section{Weighted Shapley Value and weighted convexity}
\label{SectionWeightedShapleyAndweightedConvexity}

\subsection{Weighted Shapley Value}

The notion of Shapley value was then extended in \cite{Shapley1953a} and in \cite{KalaiAndSamet1987} to weighted Shapley value. In the first one, the weight of each player is strictly positive whereas in the second one, a player may have a null weight. It is then necessary to define a weight system to describe how the players with a zero weight would share their contributions when only them are present.

Formally, a weight system is a pair $(\omega,\Sigma)$ where $\omega \in \mathbb{R}^N_{++}$
and $\Sigma=(N_1, ..., N_m)$ is an ordered partition of $N$.
It is called simple if $\Sigma$ is equal to the singleton $N$. $m$ will be the size of the partition.
Given an element $i$ in $N$,
we define its priority denoted by $p(i)$ as the unique integer such that $i \in N_{p(i)}$.
Given a set $S$, we define the priority of $S$, denoted by $p(S)$, as the largest $k \in \{1,\cdots,m\}$ such that $N_k\cap S \neq \emptyset$.
Notice that this coincides with the previous definition if $S$ is a singleton.
Moreover, priority is a nondecreasing mapping.
We also define $\overline{S}$ as the elements in $S$ with the highest priority, i.e., $\overline{S}=\{ i \in S, p(i)=p(S)\}$.

\begin{definition}
The weighted Shapley value with weight system $(\omega,\Sigma)$ is the unique function from the set of $TU$-games to allocation such that
\begin{itemize}
\item it is linear,
\item the allocation of the unanimity game on the set $S$ is defined as follows: 
for all $i\in N$,
\[
x_i=\begin{cases}
\frac{\omega_i}{\sum_{i\in \overline{S}} \omega_i}, &\text{ if } i \in \overline{S},\\
0 &\text{ otherwise}.
\end{cases}
\]
\end{itemize}
\end{definition}

\noindent In the previous definition, it is interesting to distinguish three sets of agents: $\overline{S}$, $S-\overline{S}$ and $N-S$. In the unanimity game with set $S$,
\begin{itemize}
\item agents in $N-S$ are not contributing to obtain a positive payoff, hence they obtain $0$,
\item agents in $S-\overline{S}$ are contributing to obtain a positive payoff but they have low priority, hence they obtain $0$,
\item agents in $\overline{S}$ are contributing to obtain a positive payoff and have the highest priority inside $S$, hence they share the total value of $1$.
\end{itemize}

\noindent It is convenient for the following results to define the relative weight of the agents in the unanimity game on $S$.
Given a set $S$, one defines for all $i\in N$
\[
\overline{\omega}^S_i=\begin{cases}
\omega_i \text{ if } i \in \overline{S},\\ 
0  \text{ if } i \in S \setminus \overline{S},\\ 
0  \text{ if } i\notin S.
\end{cases}
\]
We also define $\overline{\omega}^S$ as the sum for all $i\in S$ of $\overline{\omega}^S_i$. Notice that it is equal to the sum of the weights of elements in $\overline{S}$. Moreover, the unanimity games on $S$ and $\overline{S}$ lead to the same allocation vector.\\

Using the decomposition of a game into unanimity games, we obtain that the $(\omega,\Sigma)$-weighted Shapley value $\Phi^{\omega}$ of a game $(N,v)$ is defined for all $i\in N$ by
\[
\Phi^{\omega}_{i}(v) = \sum_{S \subseteq N :\, i \in S } \frac{\overline{\omega}^S_i}{\overline{\omega}^S} \lambda_S(v)=
\sum_{S \subseteq N :\, i \in \overline{S} } \frac{\omega_i}{\overline{\omega}^S} \lambda_S(v).
\]
The interesting case is if $i$ is in $S$ and not in $\overline{S}$. It means that the priority of $i$ is too low hence, he does not obtain a contribution.\\

Similarly to the classical Shapley value, a third interpretation can be given to the weighted Shapley value.  Recall that $\mathcal{L}$ is the set of ordered sequences on $N$. Let $\mathbb{P}$ be a probability over $\mathcal{L}$, one can define an allocation as follows: let
\[
\Psi^{\mathbb{P}}_i(v)=\E_{\mathbb{P}} \Big( v\big(\{L \geq i \}\big) - v\big(\{L> i \}\big) \Big).
\] 
When $\mathbb{P}$ is the uniform distribution, one obtains the classical Shapley value.\\

Let us now define a probability distribution $\mathbb{P}_{\omega,\Sigma}$ generating the $(\omega,\Sigma)$-weighted Shapley value. We define it by induction:
\begin{itemize}
\item Let $\overline{N}$ be the agents with highest priority in $N$.
\item Then $L(1)$  is equal to $i$ with probability $\overline{\omega}^N_i/\overline{\omega}^N$. In particular, it is strictly positive if and only if $i$ has highest priority or equivalently $i\in \overline{N}$, 
\item Given $L(1)$ fixed, $(L(2),....,L(N))$ are defined following the distribution associated to the projection of $(\omega,\Sigma)$ to $N \setminus\{ L(1) \}$.
\end{itemize}
In particular by construction,  if $i,j\in N$ such that $p(i)>p(j)$ then the probability under $\mathbb{P}$ of an order $L$ where $j$ appears before $i$ in the order is equal to $0$.

\begin{proposition}\label{equivalence}
Given $N$ and $\mathbb{P}_{\omega,\Sigma}$,
then $\Psi^{\mathbb{P}_{\omega,\Sigma}}(v)$ is equal to the $(\omega,\Sigma)$-weighted Shapley value $\Phi^{\omega}$.
\end{proposition}

\begin{proof}
Proof in Appendix~\ref{appendProofOfPropositionEquivalence}.
\end{proof}

\subsection{Weighted-convexity}

In a similar way than the notion of Shapley value has been extended to weight systems, one can extend the notion of average convexity.

\begin{definition}
Let $(\omega,\Sigma)$ be a  weight system, we say that the game is $(\omega,\Sigma)$-convex if 
for every $S\subset T \subseteq N$,
\begin{equation}
\label{eqDefinition(omega,Sigma)-convexity}
\left(\sum_{i\in S} \overline{\omega_i}^T\left(v(T)-v(T \setminus \{i\})-v(S)+v(S\setminus\{i\}) \right) \right)\geq 0
\end{equation}
\end{definition}

In particular, notice that if $p(S)<p(T)$, then for every $i\in S$, $\overline{\omega}_i^T=0$ and the corresponding inequality is satisfied by any game.
If $\Sigma = \lbrace N \rbrace$ and if all weights are equal,
then $(\omega, \Sigma)$-convexity corresponds to average-convexity.
Let us also note that if a game is convex
then it is $(\omega, \Sigma)$-convex for any weight system $(\omega, \Sigma)$.

\begin{remark}
\label{remarkExtensionWithNullPlayers}
If the game $(N,v)$ contains null players then we can restrict condition~(\ref{eqDefinition(omega,Sigma)-convexity})
to coalitions $S$ and $T$ excluding null players.
We first argue that one can exclude null players from $S$ and then from $T$. If $i \in S$ is a null player
then $v(T)-v(T \setminus \{i\})-v(S)+v(S\setminus\{i\}) = 0$, hence its contribution is zero. Thus we can assume $S$ without null player. Now, let $T'$ be the set of null players in $T$. By the previous sentence, we know that $S \subset T \setminus T'$.
There are two cases:
\begin{itemize}
\item if $p(T) > p(T \setminus T')$, then $\overline{w_i}^{T} = 0$ for all $i$ in $S$. Element of $S$ have a too low priority and the inequality is trivially satisfied and can be forgotten.
\item if $p(T)=p(T\setminus T')$, then one can replace $T$ by $T\setminus T'$.
\end{itemize}
\end{remark}

\subsection{Results}

One obtains the following theorem which extends the result of
\cite{Inarra-Usategui-1993}.

\begin{theorem}
\label{thIfweightedAverageConvexThenWeightedShapleyValueInCore}
Let $\omega \in \mathbb{R}_{++}^N$ and $\Sigma$ be a partition.
If the game is $(\omega,\Sigma)$-convex then its $(\omega,\Sigma)$-weighted Shapley value is in the core.
\end{theorem}

The first part of the proof is inspired from \citep{Sprumont1990}.
Let $(N,v)$ be a given game.
For any non-empty coalition $T \subseteq N$, we denote by $v^T$ the subgame of $v$ induced by $T$,
\emph{i.e.}, $v^{T}(S) = v(S)$ for any $S \subseteq T$.
Let $\Phi^\omega(v^T)$ be the $(\omega, \Sigma)$-weighted Shapley value of the subgame $v^T$.
Define the vector $\Psi^{\omega} = (\Psi^{\omega}_{iT})_{i \in T,T\in \mathcal{P}(N)}$ recursively by
\[
\Psi^{\omega}_{iT}=
\frac{\overline{\omega}^T_i}{\overline{\omega}^T} (v(T)-v(T \setminus \{i\}))+
\sum_{j \in T \setminus \lbrace i \rbrace} \frac{\overline{\omega}^T_j}{\overline{\omega}^T} \Psi^{\omega}_{iT\setminus \{j\}},
\]
for all $i \in T$, $T \in \mathcal{P}(N)$,
and setting $\Psi^{\omega}_{i\emptyset} = 0$ for all $i \in N$.


\begin{restatable}{proposition}{VectorAndShapleyValue}
\label{propPsiOmegaiTiInT=PhiOmegavT}
For all $T\in \mathcal{P}(N)$,
$(\Psi^{\omega}_{iT})_{i\in T}=\Phi^{\omega}(v^T)$.
\end{restatable}

\begin{proof}
Proof in Appendix~\ref{appendProofOfPropositionRec}.
\end{proof}

As a consequence, one obtains the following result.
\begin{proposition}
Let us denote by $\Phi^\omega_{T}$ the $(\omega, \Sigma)$-weighted Shapley value of $v^T$.
We have
\[
\Phi^{\omega}_{iT}=
\frac{\overline{\omega}^T_i}{\overline{\omega}^T} (v(T)-v(T \setminus \{i\}))+\sum_{j \in T \setminus \lbrace i \rbrace} \frac{\overline{\omega}^T_j}{\overline{\omega}^T} \Phi^{\omega}_{iT\setminus \{j\}},
\]
for all $i \in T$.
\end{proposition}

Let us now prove the main result.

\begin{proof}
[Proof of Theorem~\ref{thIfweightedAverageConvexThenWeightedShapleyValueInCore}]
We do a proof by induction.
If  $n=1$, the result is true.
Let $n>1$ and assume that for all $(\omega, \Sigma)$-convex games of size $n-1$, the result is true.
Let us consider a $(\omega, \Sigma)$-convex game $v$ of size $n$.
We check the condition for $\Phi^{\omega}$ to be in the core.
Let $T$ be a subset of $N$.
By Proposition~\ref{propPsiOmegaiTiInT=PhiOmegavT} and by definition of $\Psi^{\omega}$, we get
\begin{align}
\label{eqSumiInTPhiOmegai(v)=SumiInTPsiOmegaiN}
\sum_{i\in T} \Phi^{\omega}_i(v) & = \sum_{i\in T} \Psi^{\omega}_{iN}, \nonumber \\
& = \sum_{i\in T}  \left( \frac{\overline{\omega}^N_i}{\overline{\omega}^N} (v(N)-v(N \setminus\{i\})) + \sum_{k\in N,\, k\neq i}  \frac{\overline{\omega}^N_k}{\overline{\omega}^N}\Psi^{\omega}_{iN\setminus{k}}\right), \nonumber \\
& = \frac{1}{\overline{\omega}^N} \left( \sum_{i\in T} \overline{\omega}^N_i (v(N)-v(N \setminus \{i\}))+ \sum_{k \in N} \overline{\omega}^N_k\left( \sum_{i\neq k,\, i\in T} \Psi^{\omega}_{iN \setminus{k}}\right)\right).
\end{align}
By Proposition~\ref{propPsiOmegaiTiInT=PhiOmegavT},
we have $\Phi^{\omega}(v^{N \setminus k}) = (\Psi^{\omega}_{iN\setminus k})_{i\in N \setminus k}$
and by the induction assumption,
$\Phi^{\omega}(v^{N \setminus k})$ belongs to the core of $v^{N \setminus k}$.
We distinguish two cases for the sum $\sum_{i\neq k,\, i\in T}\Psi^{\omega}_{i N \setminus{k}}$:
\begin{itemize}
\item
If $k\in T$, then one sum for all $i\in T\setminus{k}$ and by induction assumption
\begin{equation}
\label{eqSumiNeqkiInTPsiOmegaiNSetminuskGeqv(TSetminusk)}
\sum_{i\neq k,\, i\in T} \Psi^{\omega}_{i N \setminus{k}} \geq v(T\setminus{k}).
\end{equation}
Let us note that (\ref{eqSumiNeqkiInTPsiOmegaiNSetminuskGeqv(TSetminusk)}) is satisfied at equality for $T = N$.
\item
If $k\notin T$, then one sum for all $i\in T$ and by induction assumption
\begin{equation}
\label{eqSumi<=kiInTPsiOmegaiNSetminusk<=v(T)}
\sum_{i\neq k,\, i\in T} \Psi^{\omega}_{i N \setminus{k}} \geq v(T).
\end{equation}
\end{itemize}
If $T = N$, then (\ref{eqSumiInTPhiOmegai(v)=SumiInTPsiOmegaiN}) and (\ref{eqSumiNeqkiInTPsiOmegaiNSetminuskGeqv(TSetminusk)})
imply $\sum_{i\in N} \Phi^{\omega}_i(v) = v(N).$
Let now $T$ be a strict subset of $N$.
By (\ref{eqSumiInTPhiOmegai(v)=SumiInTPsiOmegaiN}),(\ref{eqSumiNeqkiInTPsiOmegaiNSetminuskGeqv(TSetminusk)}),
and (\ref{eqSumi<=kiInTPsiOmegaiNSetminusk<=v(T)}), one obtains
\[
\sum_{i\in T} \Phi^{\omega}_i(v) \geq \frac{1}{\overline{\omega}^N} \left( \sum_{i\in T} \overline{\omega}^N_i (v(N)-v(N \setminus \{i\}))
+ \sum_{k\in T} \overline{\omega}^N_k v(T\setminus{k})+ \sum_{k \in N,\, k\notin T} \overline{\omega}^N_k v(T)\right).
\]
Finally,
the assumption of $(\omega,\Sigma)$-convexity implies
\[
\sum_{i\in T} \Phi^{\omega}_i(v) \geq \frac{1}{\overline{\omega}^N} \left( \sum_{i\in T} \overline{\omega}^N_i (v(T)-v(T \setminus \{i\}))
+ \sum_{k\in T} \overline{\omega}_k^N v(T\setminus{k})+ \sum_{k \in N,\, k\notin T} \overline{\omega}^N_k v(T)\right).
\]
We now see that the terms $v(T\setminus\{i\})$ are compensated by  $v(T\setminus\{k\})$,
whereas the terms $v(T)$ are accumulated with a total weight of $\overline{w}^N$, hence
\begin{align*}
\sum_{i\in T} \Phi^{\omega}_i(v) \geq v(T).
\end{align*}
\end{proof}

\subsection{Link between \texorpdfstring{$(\omega,\Sigma)$}{(w,E)}-convexity and superadditivity}

It is well known that average-convexity implies superadditivity (\cite{Inarra-Usategui-1993}).
We prove that $(\omega, \Sigma)$-convexity also implies superadditivity when $\Sigma = \lbrace N \rbrace$. 
We first establish a more general result with possibly different priorities
which will be also useful in Section~\ref{SectionWeightedConvexityAndCommunicationGameCaseWithSeveralPriorities}.

\begin{proposition}
\label{weak_superadditivity}
Let $(N,v)$ be an $(\omega,\Sigma)$-convex game.
Let us consider $S,T,U \subseteq N$ such that $S \cap T = \emptyset$, $U \subseteq T$ and for all $s\in S$, $p(s)>p(U)$ and $p(s)\geq p(T)$ then
\begin{equation}
\label{eqSuperadditivity_weak}
v(S \cup T) - v(T) \geq v(S \cup U)-v(U).
\end{equation}
\end{proposition}

\begin{proof}
The proof is by induction on $\vert S \vert$.
If $\vert S \vert = 1$, then we set $S = \lbrace i \rbrace$
and $(\omega,\Sigma)$-convexity applied to $S'=\lbrace i \rbrace \cup U \subseteq T' = \lbrace i \rbrace \cup T$ implies
\[
\overline{\omega}^{T'}_i (v(\lbrace i \rbrace \cup T) - v(T)) \geq \overline{\omega}^{T'}_i (v(\lbrace i \rbrace \cup U) - v(U)).
\]
By assumption on the priorities,
we have $p(i)= p(T')$ hence $\overline{\omega}^{T'}_i >0$ and one can divide the inequality by $\overline{\omega}^{T'}_i >0$.
We obtain that (\ref{eqSuperadditivity_weak}) is satisfied for $|S|=1$.

Let $k$ be a given integer with $k \geq 1$. 
Let us assume (\ref{eqSuperadditivity_weak}) satisfied for any triple $S, T,U \subseteq N$ such that $S \cap T = \emptyset$,
$U \subseteq T$,
and for all $s\in S$, $p(s)>p(U)$, $p(s)\geq p(T)$, and $|S|=k$.
Let us consider $S, T,U \subseteq N$ such that $S \cap T = \emptyset$,
$U \subseteq T$,
and for all $s\in S$, $p(s)>p(U)$, $p(s) \geq p(T)$, and $|S|=k+1$.
Let us set w.l.o.g. $S = \lbrace 1, 2 , \ldots, k+1 \rbrace$.
$(\omega,\Sigma)$-convexity applied to $S\cup U \subseteq S \cup T$ implies
\begin{equation}
\label{eqSumi=1k+1omegaivTUS-vTS-i>=Sumi=1k+1omegaivS-vS-i}
\sum_{i =1}^{k+1} \overline{\omega}^{T\cup S}_i \left(v(S \cup T) - v((S \setminus \lbrace i \rbrace) \cup T)\right) \geq
\sum_{i =1}^{k+1} \overline{\omega}^{T\cup S}_i \left(v(S \cup U) - v((S \setminus \lbrace i \rbrace)\cup U )\right).
\end{equation}
We do not need to consider other element in $S\cup U$ since all elements of $U$ have too low priorities.
By induction hypothesis,
we have 
\[
v((S \setminus \{i\}) \cup T) - v(T) \geq v((S \setminus \{i\}) \cup U) - v(U),
\]
for all $i \in \lbrace 1, 2, \ldots, k+1 \rbrace$.
Multiplying each inequality by $\overline{\omega}^{T\cup S}_i$ and adding them to (\ref{eqSumi=1k+1omegaivTUS-vTS-i>=Sumi=1k+1omegaivS-vS-i}),
we get
\begin{equation}
\sum_{i =1}^{k+1} \overline{\omega}^{T\cup S}_i (v(S \cup T)-v(T)) \geq
\sum_{i =1}^{k+1} \overline{\omega}^{T\cup S}_i (v(S \cup U) -v(U)).
\end{equation}
This implies (\ref{eqSuperadditivity_weak})
and proves the result at the next step of the induction.
By the principle of induction, the result is true for every set $S$.
\end{proof}

We obtain as a particular result that the value function $v$ is superadditive if there exists only one partition.

\begin{corollary}
Let $(N,v)$ be an $(\omega,\Sigma)$-convex game such that $\Sigma=\{N\}$, then $v$ is superadditive.
\end{corollary}

\begin{proof}
Let us consider $S, T\subseteq N$ such that $S \cap T = \emptyset$.
Taking $U=\emptyset$, the conditions of Proposition~\ref{weak_superadditivity} are satisfied for any pair of sets. We obtain
\begin{equation}
\label{eqSuperadditivity}
v(S \cup T) - v(T) \geq v(S )-v(\emptyset).
\end{equation}
The game is superadditive.
\end{proof}

\section{Weighted-convexity and Communication game : The singleton case}
\label{SectionSingletonCase}

We now consider TU-games with a communication graph $G=(N,E)$.
Let $(N,v)$ be a TU-game.
We define its restriction to the communication graph $G=(N,E)$ as the TU-game with characteristic function
\[
v^G(S)=\sum_{A \in C_G(S)} v(A),
\]
where $C_G(S)$ is the set of all connected components of the induced graph $(S,E^S)$
where only nodes in $S$ and edges with both end-nodes in $S$ remained.

\begin{definition}
Let $(\omega,\Sigma)$ be a weight system. We say that a graph $G=(N,E)$ \emph{preserves the} $(\omega,\Sigma)$-convexity
if for every game $(N,v)$ that is $(\omega,\Sigma)$-convex then the communication game $(N,v^G)$ is also $(\omega,\Sigma)$-convex.
\end{definition}

We want to investigate what type of graph preserves the $(\omega,\Sigma)$-convexity as a function of $\omega$ and $\Sigma$. 
\cite{Slikker1998} shows that a connected communication graph preserves average convexity if and only if it is a complete graph or a star. If the graph is not connected, the condition becomes for each connected component to be a complete graph or a star. In our formalism, it corresponds to the case where $\Sigma=\{N\}$ and $\omega$ is a constant vector.  In order to prove this result, he shows that a graph $G$ preserving average convexity has to satisfy the following conditions:
\begin{enumerate}
\item
\label{enumGHasToBeCycleComplete}
 $G$ has to be cycle-complete and,
\item
\label{enumGCannotContainA4-pathOrA3-pan}
$G$ cannot contain the following subgraphs: a $4$-path and a $3$-pan.
\end{enumerate}
The formal definitions of these notions are the following.
\begin{definition}
We say that a cycle $C = \lbrace v_1, e_1, v_2, e_2, \ldots,v_m, e_m, v_1 \rbrace $ in a graph $G = (N,E)$ is complete (resp. non-complete)
if the subset $\lbrace v_1, v_2, \ldots v_m \rbrace \subseteq N$ of nodes of $C$ induces a complete (resp. non-complete) subgraph in $G$.
A graph $G=(N,E)$ is cycle-complete if any cycle
$C$ in $G$ is complete.
\end{definition}

\begin{definition}
A graph $G=(N,E)$ admits a \emph{$4$-path} subgraph if there exist $i,j,k,l \in N$ such that
\begin{itemize}
\item
$\{i,j\},\{j,k\},\{k,l\}\in E$
\item
any other link between $\{i,j,k,l\}$ is absent:
\end{itemize}
A graph $G=(N,E)$ admits a \emph{$3$-pan} subgraph if there exists $i,j,k,l \in N$ such that
\begin{itemize}
\item
$\{i,j\},\{j,k\},\{k,i\}\in E$ and $\{l,i\}\in E$
\item
any other link between $\{i,j,k,l\}$ is absent:
\end{itemize}
Notice that in both cases, $j$ and $k$ have degree $2$ in the subgraph induced by $\{i,j,k,l\}$.
\end{definition}

A $4$-path and a $3$-pan are represented in Figures~\ref{4-path} and \ref{3-pan} respectively.
\def\ech{.5}
\begin{figure}[!ht]
\begin{subfigure}[t]{0.45\textwidth}
\centering{
\begin{tikzpicture}[scale=\ech,shorten >=1pt,auto,node distance=3cm,thick,main
 node/.style={circle,draw,font=\Large\bfseries}]
\node [draw,text width=0.5cm,text centered,circle,scale=\ech] (A) at (-6,0) {$i$};
\node [draw,text width=0.5cm,text centered,circle,scale=\ech] (B) at (-4,0) {$j$};
\node [draw,text width=0.5cm,text centered,circle,scale=\ech] (C) at (-2,0) {$k$};
\node [draw,text width=0.5cm,text centered,circle,scale=\ech] (D) at (0,0) {$l$};
\draw[-,>=latex,color=red,scale=\ech] (A) to node[midway,color=red,scale=\ech,above] {}(B);
\draw[-,>=latex,color=red,scale=\ech] (B) to node[midway,color=red,scale=\ech,above] {}(C);
\draw[-,>=latex,color=red,scale=\ech] (C) to node[midway,color=red,scale=\ech,above] {}(D);
\end{tikzpicture}
}
\caption{$4$-path.}
\label{4-path}
\end{subfigure}
\hfill
\begin{subfigure}[c]{0.45\textwidth}
\centering{
\begin{tikzpicture}[scale=\ech,shorten >=1pt,auto,node distance=3cm,thick,main
 node/.style={circle,draw,font=\Large\bfseries}]
\node [draw,text width=0.5cm,text centered,circle,scale=\ech] (A) at (-4,-2) {$k$};
\node [draw,text width=0.5cm,text centered,circle,scale=\ech] (B) at (-4,2) {$j$};
\node [draw,text width=0.5cm,text centered,circle,scale=\ech] (C) at (-2,0) {$i$};
\node [draw,text width=0.5cm,text centered,circle,scale=\ech] (D) at (0,0) {$l$};
\draw[-,>=latex,color=red,scale=\ech] (A) to node[midway,color=red,scale=\ech,above] {}(B);
\draw[-,>=latex,color=red,scale=\ech] (A) to node[midway,color=red,scale=\ech,above] {}(C);
\draw[-,>=latex,color=red,scale=\ech] (B) to node[midway,color=red,scale=\ech,above] {}(C);
\draw[-,>=latex,color=red,scale=\ech] (C) to node[midway,color=red,scale=\ech,above] {}(D);
\end{tikzpicture}
}
\caption{$3$-pan.}
\label{3-pan}
\end{subfigure}
\caption{}
\end{figure}

Let us start with the simple case where $\Sigma$ is a singleton, \emph{i.e.}, $\Sigma = \lbrace N \rbrace$.
Then all players have the same priority and we will prove that we  obtain the same characterization as \cite{Slikker1998} for inheritance of average convexity. Let us note that, as all players have the same priority, we have $\overline{\omega_i}^T = \omega_i$ for any player $i \in T \subseteq N$ throughout this section. Hence, the difference with \cite{Slikker1998} is the non homogeneity of the weights.

\begin{theorem}
\label{characterization_singleton}
Let $G=(N,E)$ be a connected graph and let $(\omega,\Sigma)$ be a weight system with $\Sigma = \lbrace N \rbrace$.
Then the following properties are equivalent.
\begin{enumerate}
\item
\label{GPreserves(omega,Sigma)-convexity}
$G$ preserves $(\omega,\Sigma)$-convexity.
\item
\label{InheritanceOf(omega,Sigma)-convexityCycle-completeNo4PathOr3Pan}
\begin{enumerate}
\item
\label{InheritanceOf(omega,Sigma)-convexityCycle-complete}
$G$ is cycle-complete.
\item
\label{InheritanceOf(omega,Sigma)-convexityNo4PathOr3Pan}
There is no restricted subgraph that is a $4$-path or a $3$-pan.
\end{enumerate}
\item
\label{(omega,Sigma)-ConvexityIfAndOnlyIfCompleteGraphOrStar}
$G$ is a complete graph or a star.
\end{enumerate}
\end{theorem}

The equivalence between Conditions~\ref{InheritanceOf(omega,Sigma)-convexityCycle-completeNo4PathOr3Pan}
and~\ref{(omega,Sigma)-ConvexityIfAndOnlyIfCompleteGraphOrStar} in Theorem~\ref{characterization_singleton}
was already established in~\cite{Slikker1998}.

\begin{remark}
In fact, it is enough to establish the necessity of Condition~\ref{InheritanceOf(omega,Sigma)-convexityCycle-complete}
(resp. Condition~\ref{InheritanceOf(omega,Sigma)-convexityNo4PathOr3Pan})
in Theorem~\ref{characterization_singleton} for a graph that is exactly a non-complete cycle
(resp. a $4$-path or a $3$-pan).
For a graph containing one of these graphs as subgraph,
one can consider a game where only agents involved in the specific subgraph
contribute to the utility and all remaining agents are null players.
By Remark~\ref{remarkExtensionWithNullPlayers},
such a game is $(\omega,\Sigma)$-convex if its restriction to players in the non-complete cycle
(resp. $4$-path or $3$-pan)
is $(\omega,\Sigma)$-convex.
\end{remark}

Let us establish three intermediate results before proving Theorem~\ref{characterization_singleton}: a counter-example if the graph is not cycle-complete, a counter-example if the graph is a $3$-pan (or a $4$-path) and a proposition for stars.

\begin{example}[Weighted Non-complete cycle]
\label{ExampleNon-CompleteCycle}
\rm
Let $C = \lbrace 1, e_1, 2, e_2, \ldots,m, e_m, 1 \rbrace $ be a non-complete cycle in $G$
and let $l^*$ be a node of $C$. Let $j$ and $k$ be the neighbors of $l^*$ in~$C$ and let us assume $\lbrace j, k \rbrace \notin E$.
We consider the convex game $(N,v)$ defined by $v(S) = \vert S \vert -1$ for all $S \subseteq N$, $S \not= \emptyset$.
As $(N,v)$ is convex, it is also $(\omega, \Sigma)$-convex for any weight system $(\omega, \Sigma)$.
Let us consider $S = \lbrace j, l^*, k \rbrace$ and $T = V(C)$
as represented in Figure~\ref{figNon-CompleteCycle-a}.
\begin{figure}[!ht]
\centering
\begin{pspicture}(-.5,-.3)(1,1.8)
\tiny
\begin{psmatrix}[mnode=circle,colsep=0.4,rowsep=0.1]
			& {}	& {$j$}\\
{$1$}	&				&				&{$l^*$}\\
 & {$m$} 	& {$k$}
\psset{arrows=-, shortput=nab,labelsep={0.08}}
\tiny
\ncline{2,1}{3,2}
\ncline[linecolor=gray]{2,1}{3,3}
\ncline[linecolor=gray]{2,1}{1,3}
\ncline[linecolor=gray]{3,2}{2,4}
\ncline[linecolor=gray]{3,3}{1,2}
\ncline{2,1}{1,2}
\ncline{3,2}{3,3}
\ncline{1,2}{1,3}
\ncline[linecolor=red]{1,3}{2,4}
\ncline[linecolor=gray]{2,4}{1,2}
\ncline[linecolor=red]{2,4}{3,3}
\end{psmatrix}
\end{pspicture}
\caption{Non-complete cycle $C$, $\lbrace j,k \rbrace \notin E$.}
\label{figNon-CompleteCycle-a}
\end{figure}
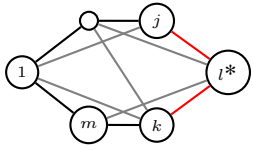
We get
\[
\sum_{i \in S} \omega_i (v^G(S) - v^G(S \setminus \lbrace i \rbrace)) = 
\omega_{j} + 2 \omega_{l^*} + \omega_{k}
>
\omega_{j} + \omega_{l^*} + \omega_{k} = 
\sum_{i \in S} \omega_i (v^G(T) - v^G(T \setminus \lbrace i \rbrace)).
\]
This contradicts $(\omega, \Sigma)$-convexity of the communication game $(N,v^G)$.
\end{example}

As a Corollary of this example, one recovers the necessity for the graph to be cycle-complete.
Indeed, if the graph is not cycle-complete, one can consider the smallest cycle that is not complete.
By minimality, there exists three nodes $l^*$, $j$ and $k$  such that $j$ and $k$ are neighbors of $l^*$ in~$C$ and $\lbrace j, k \rbrace \notin E$.
We choose the previous formulation since it will be closer to the general result with priorities.\\

The second counter-example shows that a graph preserving $(\omega, \Sigma)$-convexity cannot contain a $3$-pan subgraph. Our example is different and much more general than the one given by~\cite{Slikker1998} for average convexity since it needs to take into account the different weights. 
We present the counter-example here and highlights only that the communication game is not weighted average convex.
This counter-example will be also used later in a more general version when agents have different priorities.
Checking all the inequalities to establish $(\omega,\Sigma)$-convexity of the game

will be done in the Appendix
in the general case with priorities.

Moreover, we show thereafter that the game considered in this example also gives a contradiction to inheritance of  $(\omega, \Sigma)$-convexity if the graph contains a $4$-path.

\begin{example}[$3$-pan]
\label{Example3-pan}
\rm
Let us assume that there are only $4$ nodes denoted by $\{1,2,3,4\}$.
The graph $G$ is a $3$-pan such that $4$ has degree $2$ as represented in Figure~\ref{3-panWithw_4>Epsilon}.
\begin{figure}[!ht]
\begin{center}
\begin{tikzpicture}[scale=\ech,shorten >=1pt,auto,node distance=4cm,thick,main
 node/.style={circle,draw,font=\Large\bfseries}]
\node [draw,text width=0.5cm,text centered,circle,scale=\ech] (A) at (-4,-2) {$4$};
\node [draw,text width=0.5cm,text centered,circle,scale=\ech] (B) at (-4,2) {$1$};
\node [draw,text width=0.5cm,text centered,circle,scale=\ech] (C) at (-2,0) {$2$};
\node [draw,text width=0.5cm,text centered,circle,scale=\ech] (D) at (0,0) {$3$};
\draw[-,>=latex,color=red,scale=\ech] (A) to node[midway,color=red,scale=\ech,above] {}(B);
\draw[-,>=latex,color=red,scale=\ech] (A) to node[midway,color=red,scale=\ech,above] {}(C);
\draw[-,>=latex,color=red,scale=\ech] (B) to node[midway,color=red,scale=\ech,above] {}(C);
\draw[-,>=latex,color=red,scale=\ech] (C) to node[midway,color=red,scale=\ech,above] {}(D);
\end{tikzpicture}
\caption{}
\label{3-panWithw_4>Epsilon}
\end{center}
\end{figure}
We assume that all nodes have the same priority. Let us prove that there is no conservation of $(\omega,\Sigma)$-convexity.

Define
\begin{align*}
X & = \max \left(1+\frac{\omega_2}{\omega_4},1+\frac{\omega_3}{\omega_4}\right),\\
Y & =1+\frac{\omega_1}{\omega_4},\\
Z &=X+Y+1+\frac{\omega_1}{\omega_2+\omega_3+\omega_4}X, \\
\Theta & =Z+X-1.
\end{align*}
We consider the following game $(N,v)$ and the communication game $(N,v^G)$ induced by $(N,v)$ and the communication structure given by $G$. 
\[
\begin{array}{ccc}
v(S)=\begin{cases}
X \text{ if } S =\{1,4\} \text{ or } \{1,2,4\},\\
Y \text{ if } S=\{3,4\},\\
X+Y-1 \text{ if } S=\{1,3,4\},\\
Z \text { if } S=\{2,3,4\},\\
\Theta \text{ if } S=N,\\
0 \text{ otherwise}.
\end{cases}
&
&
v^G(S)=\begin{cases}
X \text{ if } S=\{1,4\} \text{ or } \{1,2,4\} ,\\
0 \text{ if } S=\{3,4\},\\
X \text{ if } S = \{1,3,4\},\\
Z \text { if } S=\{2,3,4\},\\
\Theta \text{ if } S=N\\
0 \text{ otherwise}.
\end{cases}
\end{array}
\]
We represent in Figure~\ref{LatticeFirstCounterExample-main}
the lattices of coalitions of $\lbrace 1, 2, 3, 4 \rbrace$ ordered by inclusion for $(N,v)$ and $(N,v^G)$
with the corresponding values.
\def\ech{.6}
\begin{figure}[!ht]
\begin{subfigure}[c]{0.45\textwidth}
\centering
\begin{center}
\begin{tikzpicture}[scale=\ech,shorten >=1pt,auto,node distance=3cm,thick,main
 node/.style={circle,draw,font=\Large\bfseries}]
\node [draw,text width=0.5cm,text centered,circle,scale=\ech,label=right:{\scriptsize $0$}] (V) at (0,-2) {$\emptyset$};
\node [draw,text width=0.5cm,text centered,circle,scale=\ech,label=right:{\scriptsize $0$}] (1) at (-3,0) {$1$};
\node [draw,text width=0.5cm,text centered,circle,scale=\ech,label=right:{\scriptsize $0$}] (2) at (-1,0) {$2$};
\node [draw,text width=0.5cm,text centered,circle,scale=\ech,label=right:{\scriptsize $0$}] (3) at (1,0) {$3$};
\node [draw,text width=0.5cm,text centered,circle,scale=\ech,label=right:{\scriptsize $0$}] (4) at (3,0) {$4$};
\node [draw,text width=0.5cm,text centered,circle,scale=\ech,label=right:{\scriptsize $0$}] (12) at (-5,2) {$12$};
\node [draw,text width=0.5cm,text centered,circle,scale=\ech,label=right:{\scriptsize $0$}] (13) at (-3,2) {$13$};
\node [draw,text width=0.5cm,text centered,circle,scale=\ech,label=right:{\scriptsize $X$}] (14) at (-1,2) {$14$};
\node [draw,text width=0.5cm,text centered,circle,scale=\ech,label=right:{\scriptsize $0$}] (23) at (1,2) {$23$};
\node [draw,text width=0.5cm,text centered,circle,scale=\ech,label=right:{\scriptsize $0$}] (24) at (3,2) {$24$};
\node [draw,text width=0.5cm,text centered,circle,scale=\ech,label=right:{\scriptsize $Y$}] (34) at (5,2) {$34$};
\node [draw,text width=0.5cm,text centered,circle,scale=\ech,label=right:{\scriptsize $0$}] (123) at (-3,4) {$123$};
\node [draw,text width=0.5cm,text centered,circle,scale=\ech,label=right:{\scriptsize $X$}] (124) at (-1,4) {$124$};
\node [draw,text width=0.5cm,text centered,circle,scale=\ech,label={[label distance=-.3cm]2:\tiny $\begin{array}{l}X+\\Y-1\end{array}$}] (134) at (1,4) {$134$};
\node [draw,text width=0.5cm,text centered,circle,scale=\ech,label=right:{\scriptsize $Z$}] (234) at (3,4) {$234$};
\node [draw,text width=0.5cm,text centered,circle,scale=\ech,label=right:{\scriptsize $\Theta$}] (1234) at (0,6) {$N$};
\draw[-,>=latex,color=red,scale=\ech] (V) to node[midway,color=red,scale=\ech,above] {}(1);
\draw[-,>=latex,color=red,scale=\ech] (V) to node[midway,color=red,scale=\ech,above] {}(2);
\draw[-,>=latex,color=red,scale=\ech] (V) to node[midway,color=red,scale=\ech,above] {}(3);
\draw[-,>=latex,color=red,scale=\ech] (V) to node[midway,color=red,scale=\ech,above] {}(4);
\draw[-,>=latex,color=red,scale=\ech] (1) to node[midway,color=red,scale=\ech,above] {}(12);
\draw[-,>=latex,color=red,scale=\ech] (1) to node[midway,color=red,scale=\ech,above] {}(13);
\draw[-,>=latex,color=red,scale=\ech] (1) to node[midway,color=red,scale=\ech,above] {}(14);
\draw[-,>=latex,color=red,scale=\ech] (2) to node[midway,color=red,scale=\ech,above] {}(12);
\draw[-,>=latex,color=red,scale=\ech] (2) to node[midway,color=red,scale=\ech,above] {}(23);
\draw[-,>=latex,color=red,scale=\ech] (2) to node[midway,color=red,scale=\ech,above] {}(24);
\draw[-,>=latex,color=red,scale=\ech] (3) to node[midway,color=red,scale=\ech,above] {}(13);
\draw[-,>=latex,color=red,scale=\ech] (3) to node[midway,color=red,scale=\ech,above] {}(23);
\draw[-,>=latex,color=red,scale=\ech] (3) to node[midway,color=red,scale=\ech,above] {}(34);
\draw[-,>=latex,color=red,scale=\ech] (4) to node[midway,color=red,scale=\ech,above] {}(14);
\draw[-,>=latex,color=red,scale=\ech] (4) to node[midway,color=red,scale=\ech,above] {}(24);
\draw[-,>=latex,color=red,scale=\ech] (4) to node[midway,color=red,scale=\ech,above] {}(34);
\draw[-,>=latex,color=red,scale=\ech] (12) to node[midway,color=red,scale=\ech,above] {}(123);
\draw[-,>=latex,color=red,scale=\ech] (12) to node[midway,color=red,scale=\ech,above] {}(124);
\draw[-,>=latex,color=red,scale=\ech] (13) to node[midway,color=red,scale=\ech,above] {}(123);
\draw[-,>=latex,color=red,scale=\ech] (13) to node[midway,color=red,scale=\ech,above] {}(134);
\draw[-,>=latex,color=red,scale=\ech] (14) to node[midway,color=red,scale=\ech,above] {}(124);
\draw[-,>=latex,color=red,scale=\ech] (14) to node[midway,color=red,scale=\ech,above] {}(134);
\draw[-,>=latex,color=red,scale=\ech] (23) to node[midway,color=red,scale=\ech,above] {}(123);
\draw[-,>=latex,color=red,scale=\ech] (23) to node[midway,color=red,scale=\ech,above] {}(234);
\draw[-,>=latex,color=red,scale=\ech] (24) to node[midway,color=red,scale=\ech,above] {}(124);
\draw[-,>=latex,color=red,scale=\ech] (24) to node[midway,color=red,scale=\ech,above] {}(234);
\draw[-,>=latex,color=red,scale=\ech] (34) to node[midway,color=red,scale=\ech,above] {}(134);
\draw[-,>=latex,color=red,scale=\ech] (34) to node[midway,color=red,scale=\ech,above] {}(234);
\draw[-,>=latex,color=red,scale=\ech] (1234) to node[midway,color=red,scale=\ech,above] {}(123);
\draw[-,>=latex,color=red,scale=\ech] (1234) to node[midway,color=red,scale=\ech,above] {}(124);
\draw[-,>=latex,color=red,scale=\ech] (1234) to node[midway,color=red,scale=\ech,above] {}(134);
\draw[-,>=latex,color=red,scale=\ech] (1234) to node[midway,color=red,scale=\ech,above] {}(234);
\end{tikzpicture}
\end{center}
\caption{$(N,v)$.}
\label{LatticeFirstCounterExample-main-1}
\end{subfigure}
\hfill
\begin{subfigure}[c]{0.45\textwidth}
\centering
\begin{center}
\begin{tikzpicture}[scale=\ech,shorten >=1pt,auto,node distance=3cm,thick,main
 node/.style={circle,draw,font=\Large\bfseries}]
\node [draw,text width=0.5cm,text centered,circle,scale=\ech,label=right:{\scriptsize $0$}] (V) at (0,-2) {$\emptyset$};
\node [draw,text width=0.5cm,text centered,circle,scale=\ech,label=right:{\scriptsize $0$}] (1) at (-3,0) {$1$};
\node [draw,text width=0.5cm,text centered,circle,scale=\ech,label=right:{\scriptsize $0$}] (2) at (-1,0) {$2$};
\node [draw,text width=0.5cm,text centered,circle,scale=\ech,label=right:{\scriptsize $0$}] (3) at (1,0) {$3$};
\node [draw,text width=0.5cm,text centered,circle,scale=\ech,label=right:{\scriptsize $0$}] (4) at (3,0) {$4$};
\node [draw,text width=0.5cm,text centered,circle,scale=\ech,label=right:{\scriptsize $0$}] (12) at (-5,2) {$12$};
\node [draw,text width=0.5cm,text centered,circle,scale=\ech,label=right:{\scriptsize $0$}] (13) at (-3,2) {$13$};
\node [draw,text width=0.5cm,text centered,circle,scale=\ech,label=right:{\scriptsize $X$}] (14) at (-1,2) {$14$};
\node [draw,text width=0.5cm,text centered,circle,scale=\ech,label=right:{\scriptsize $0$}] (23) at (1,2) {$23$};
\node [draw,text width=0.5cm,text centered,circle,scale=\ech,label=right:{\scriptsize $0$}] (24) at (3,2) {$24$};
\node [draw,text width=0.5cm,text centered,circle,scale=\ech,label=right:{\scriptsize $0$}] (34) at (5,2) {$34$};
\node [draw,text width=0.5cm,text centered,circle,scale=\ech,label=right:{\scriptsize $0$}] (123) at (-3,4) {$123$};
\node [draw,text width=0.5cm,text centered,circle,scale=\ech,label=right:{\scriptsize $X$}] (124) at (-1,4) {$124$};
\node [draw,text width=0.5cm,text centered,circle,scale=\ech,label=right:{\scriptsize $X$}] (134) at (1,4) {$134$};
\node [draw,text width=0.5cm,text centered,circle,scale=\ech,label=right:{\scriptsize $Z$}] (234) at (3,4) {$234$};
\node [draw,text width=0.5cm,text centered,circle,scale=\ech,label=right:{\scriptsize $\Theta$}] (1234) at (0,6) {$N$};
\draw[-,>=latex,color=red,scale=\ech] (V) to node[midway,color=red,scale=\ech,above] {}(1);
\draw[-,>=latex,color=red,scale=\ech] (V) to node[midway,color=red,scale=\ech,above] {}(2);
\draw[-,>=latex,color=red,scale=\ech] (V) to node[midway,color=red,scale=\ech,above] {}(3);
\draw[-,>=latex,color=red,scale=\ech] (V) to node[midway,color=red,scale=\ech,above] {}(4);
\draw[-,>=latex,color=red,scale=\ech] (1) to node[midway,color=red,scale=\ech,above] {}(12);
\draw[-,>=latex,color=red,scale=\ech] (1) to node[midway,color=red,scale=\ech,above] {}(13);
\draw[-,>=latex,color=red,scale=\ech] (1) to node[midway,color=red,scale=\ech,above] {}(14);
\draw[-,>=latex,color=red,scale=\ech] (2) to node[midway,color=red,scale=\ech,above] {}(12);
\draw[-,>=latex,color=red,scale=\ech] (2) to node[midway,color=red,scale=\ech,above] {}(23);
\draw[-,>=latex,color=red,scale=\ech] (2) to node[midway,color=red,scale=\ech,above] {}(24);
\draw[-,>=latex,color=red,scale=\ech] (3) to node[midway,color=red,scale=\ech,above] {}(13);
\draw[-,>=latex,color=red,scale=\ech] (3) to node[midway,color=red,scale=\ech,above] {}(23);
\draw[-,>=latex,color=red,scale=\ech] (3) to node[midway,color=red,scale=\ech,above] {}(34);
\draw[-,>=latex,color=red,scale=\ech] (4) to node[midway,color=red,scale=\ech,above] {}(14);
\draw[-,>=latex,color=red,scale=\ech] (4) to node[midway,color=red,scale=\ech,above] {}(24);
\draw[-,>=latex,color=red,scale=\ech] (4) to node[midway,color=red,scale=\ech,above] {}(34);
\draw[-,>=latex,color=red,scale=\ech] (12) to node[midway,color=red,scale=\ech,above] {}(123);
\draw[-,>=latex,color=red,scale=\ech] (12) to node[midway,color=red,scale=\ech,above] {}(124);
\draw[-,>=latex,color=red,scale=\ech] (13) to node[midway,color=red,scale=\ech,above] {}(123);
\draw[-,>=latex,color=red,scale=\ech] (13) to node[midway,color=red,scale=\ech,above] {}(134);
\draw[-,>=latex,color=red,scale=\ech] (14) to node[midway,color=red,scale=\ech,above] {}(124);
\draw[-,>=latex,color=red,scale=\ech] (14) to node[midway,color=red,scale=\ech,above] {}(134);
\draw[-,>=latex,color=red,scale=\ech] (23) to node[midway,color=red,scale=\ech,above] {}(123);
\draw[-,>=latex,color=red,scale=\ech] (23) to node[midway,color=red,scale=\ech,above] {}(234);
\draw[-,>=latex,color=red,scale=\ech] (24) to node[midway,color=red,scale=\ech,above] {}(124);
\draw[-,>=latex,color=red,scale=\ech] (24) to node[midway,color=red,scale=\ech,above] {}(234);
\draw[-,>=latex,color=red,scale=\ech] (34) to node[midway,color=red,scale=\ech,above] {}(134);
\draw[-,>=latex,color=red,scale=\ech] (34) to node[midway,color=red,scale=\ech,above] {}(234);
\draw[-,>=latex,color=red,scale=\ech] (1234) to node[midway,color=red,scale=\ech,above] {}(123);
\draw[-,>=latex,color=red,scale=\ech] (1234) to node[midway,color=red,scale=\ech,above] {}(124);
\draw[-,>=latex,color=red,scale=\ech] (1234) to node[midway,color=red,scale=\ech,above] {}(134);
\draw[-,>=latex,color=red,scale=\ech] (1234) to node[midway,color=red,scale=\ech,above] {}(234);
\end{tikzpicture}
\end{center}
\caption{$(N,v^G )$.}
\label{LatticeFirstCounterExample-main-2}
\end{subfigure}
\caption{}
\label{LatticeFirstCounterExample-main}
\end{figure}
Let us state a few facts.

\noindent
\underline{The game $(N,v)$ is not convex but is $(\omega,\Sigma)$-convex}

A detailed proof for the general case with priorities is provided in Appendix~\ref{Append-Counter-Example-3-pan-Non-Valid-For-4-path}.

\noindent
\underline{The communication game $(N,v^G)$ on the $3$-pan is not $(\omega,\Sigma)$-convex:}

Let us insist on the the coalitions whose value changed compared to the original game $(N,v)$
\begin{itemize}
\item the value of $\{3,4\}$ changes from $Y$ to $0$ hence loosing $Y$,
\item the value of $\{1,3,4\}$ changes from $X+Y-1$ to $X$ hence loosing  $Y-1$. 
\end{itemize}
The key element is to notice that the value of $\{1,3,4\}$ is loosing less than the value of $\{3,4\}$
which allows for a contradiction to $(\omega,\Sigma)$-convexity.
Let us prove that $(N,v^G)$ is not $(\omega,\Sigma)$-convex.
We consider $S=\{2,3,4\}$ and $T=N$:
\begin{eqnarray}
\sum_{i\in S} \omega_i (v^G(S)-v^G(S \setminus{i}))
& = & (\omega_2 + \omega_3 + \omega_4) Z,
\end{eqnarray}
whereas
\begin{eqnarray}
\sum_{i\in S} \omega_i (v^G(T)-v^G(T \setminus{i}))
& = & \omega_2 (\Theta-X) + \omega_3 (\Theta-X) + \omega_4 (\Theta-0),\nonumber\\
& = & (\omega_2 + \omega_3 + \omega_4) Z - \omega_2 - \omega_3 + \omega_4 (X-1).
\end{eqnarray}
By definition of $X$,
we have either $- \omega_2 + \omega_4 (X-1) = 0$ or $- \omega_3 + \omega_4 (X-1) = 0$.
As $\omega_2>0$ and $\omega_3>0$,
we get in both cases
\[
\sum_{i\in S} \omega_i (v^G(T)-v^G(T \setminus{i})) < \sum_{i\in S} \omega_i (v^G(S)-v^G(S \setminus{i})),
\]
a contradiction to $(\omega, \Sigma)$-convexity.
\end{example}

\begin{remark}
\label{RemarkExample-3-panAlsoValidFor-4-path}
Although Example~\ref{Example3-pan} is designed for a $3$-pan, it is also valid for a $4$-path.
Let us consider the $4$-path
corresponding to $\lbrace 1,4,2,3 \rbrace$
as represented in Figure~\ref{4-path1423}.
\begin{figure}[!ht]
\centering{
\begin{tikzpicture}[scale=\ech,shorten >=1pt,auto,node distance=3cm,thick,main
 node/.style={circle,draw,font=\Large\bfseries}]
\node [draw,text width=0.5cm,text centered,circle,scale=\ech] (A) at (-6,0) {$1$};
\node [draw,text width=0.5cm,text centered,circle,scale=\ech] (B) at (-4,0) {$4$};
\node [draw,text width=0.5cm,text centered,circle,scale=\ech] (C) at (-2,0) {$2$};
\node [draw,text width=0.5cm,text centered,circle,scale=\ech] (D) at (0,0) {$3$};
\draw[-,>=latex,color=red,scale=\ech] (A) to node[midway,color=red,scale=\ech,above] {}(B);
\draw[-,>=latex,color=red,scale=\ech] (B) to node[midway,color=red,scale=\ech,above] {}(C);
\draw[-,>=latex,color=red,scale=\ech] (C) to node[midway,color=red,scale=\ech,above] {}(D);
\end{tikzpicture}
}
\caption{$4$-path $\lbrace 1,4,2,3 \rbrace$.}
\label{4-path1423}
\end{figure}
In this case, the communication game is equal to the previous game $v^G$ associated with the $3$-pan.
Indeed, starting from a game $(N,v)$, the only differences between the communication games associated with the specified $3$-pan and $4$-path
are potentially the values of $\{1,2,3\}$ and $\lbrace 1,2 \rbrace$.
In the $3$-pan, it is equal to the value $v(\{1,2,3\}$ (resp. $v(\{1,2\}$)
whereas in the $4$-path it is equal to $v(\{1\})+v(\{2,3\})$ (resp. $v(\{1\})+v(\{2\})$).
By construction, these values are all equal to $0$ in Example~\ref{Example3-pan},
hence it is also a counter-example to inheritance
of $(\omega, \Sigma)$-convexity for a $4$-path.
\end{remark}

\begin{remark}
In order to obtain a counter-example for another $3$-pan or $4$-path, one exchange the role of the players.
\end{remark}

Finally, \cite{Slikker1998} established that there is always inheritance of average-convexity if the underlying graph is a star.
It can be easily seen that this result is more generally valid for $(\omega,\Sigma)$-convexity
if $\Sigma = \lbrace N \rbrace$.

\begin{proposition}
\label{PropStarGraphPreserves(omega,Sigma)-convexity}
A star graph preserves $(\omega,\Sigma)$-convexity for any weight system $(\omega,\Sigma)$ with $\Sigma = \lbrace N \rbrace$.
\end{proposition}

We omit the proof of Proposition~\ref{PropStarGraphPreserves(omega,Sigma)-convexity}
as the result can be straightforwardly obtained from the proof given in~\cite{Slikker1998}\footnote{Theorem 3.2 in~\cite{Slikker1998}.}
replacing average convexity by $(\omega,\Sigma)$-convexity.

We can now prove Theorem~\ref{characterization_singleton}.

\begin{proof}[Proof of Theorem~\ref{characterization_singleton}]
\textbf{
\ref{(omega,Sigma)-ConvexityIfAndOnlyIfCompleteGraphOrStar} $\Rightarrow$ \ref{GPreserves(omega,Sigma)-convexity}
}
Let us consider a graph $G=(N,E)$ corresponding to a star or a complete graph and let $(N,v)$ be an $(\omega,\Sigma)$-convex game. 
If $G$ is complete then $v^G = v$ and $(N, v^G)$ is $(\omega,\Sigma)$-convex.
If $G$ is a star, then $(N, v^G)$ is also $(\omega,\Sigma)$-convex by~Proposition~\ref{PropStarGraphPreserves(omega,Sigma)-convexity}.

\textbf{
\ref{GPreserves(omega,Sigma)-convexity} $\Rightarrow$ \ref{InheritanceOf(omega,Sigma)-convexityCycle-completeNo4PathOr3Pan}
}
By the previous examples,
a graph preserving $(\omega,\Sigma)$-convexity has to be cycle-complete and it cannot contain any $4$-path or $3$-pan.

\textbf{
\ref{InheritanceOf(omega,Sigma)-convexityCycle-completeNo4PathOr3Pan} $\Rightarrow$ \ref{(omega,Sigma)-ConvexityIfAndOnlyIfCompleteGraphOrStar}
}
Following \cite{Slikker1998}\footnote{Lemmas~3.1 and 3.2 in \cite{Slikker1998}.},
we know that the only connected cycle-complete graphs that do not contain any $4$-path or $3$-pan
are the star and the complete graphs.
\end{proof}

As the characterization given in Theorem~\ref{characterization_singleton}
is the same established by~\cite{Slikker1998} for inheritance of average convexity, we get
that the class of graphs preserving $(\omega,\Sigma)$-convexity is equal to
the class of graphs preserving average convexity when $\Sigma = \lbrace N \rbrace$ .

\begin{corollary}
Let $G=(N,E)$ be a connected graph.
The following properties are equivalent:
\begin{enumerate}
\item
$G$ preserves average convexity.
\item
$G$ preserves $(\omega, \Sigma)$-convexity for any weight system with $\Sigma = \lbrace N \rbrace$.
\end{enumerate}
\end{corollary}

\section{Weighted convexity and Communication game: The case with several priorities}
\label{SectionWeightedConvexityAndCommunicationGameCaseWithSeveralPriorities}

The situation with several priorities is much more complex.
We will see that the previous conditions are also necessary in the presence of priorities but only for specific subgraphs of~$G$.
In particular, whereas in \cite{Slikker1998} all nodes have the same behavior, it is not the case anymore and we will need to precise the role of each node in each configuration.

We will start by presenting examples similar to Examples \ref{ExampleNon-CompleteCycle} and \ref{Example3-pan}
where priorities are now taken into account.
Then, we deduce from these examples some necessary conditions on the structure of the graphs.
Finally,
we provide a complete characterization in the special case of tree graphs satisfying
the additional condition that the priorities are inducing a hierarchy in the trees.

\subsection{Examples with priorities}

\subsubsection{Cycles and Chords}

The first example focuses on cycles.
We show that the neighbors of a node with maximal priority in any cycle have to be linked.

\begin{example}[Weighted Non-complete cycle with priorities]
\label{ExampleNon-CompleteCycle-priorities}
\rm
Let $C = \lbrace 1, e_1, 2, e_2, \ldots,m, e_m, 1 \rbrace $ be a non-complete cycle in $G$
and let $l^*$ be a node of $C$ satisfying
\[
p(l^*) = p(V(C)).
\]
Let $j$ and $k$ be the neighbors of $l^*$ in~$C$
and let us assume $\lbrace j, k \rbrace \notin E$.
We consider the convex game $(N,v)$ defined by $v(S) = \vert S \vert -1$ for all $S \subseteq N$, $S \not= \emptyset$.
As $(N,v)$ is convex, it is also $(\omega, \Sigma)$-convex for any weight system $(\omega, \Sigma)$.
Let us consider $S = \lbrace j, l^*, k \rbrace$ and $T = V(C)$
as represented in Figure~\ref{figNon-CompleteCycle-a-priorities}.
\begin{figure}[!ht]
\centering
\begin{pspicture}(-.5,-.3)(1,1.8)
\tiny
\begin{psmatrix}[mnode=circle,colsep=0.4,rowsep=0.1]
			& {}	& {$j$}\\
{$1$}	&				&				&{$l^*$}\\
 & {$m$} 	& {$k$}
\psset{arrows=-, shortput=nab,labelsep={0.08}}
\tiny
\ncline{2,1}{3,2}
\ncline[linecolor=gray]{2,1}{3,3}
\ncline[linecolor=gray]{2,1}{1,3}
\ncline[linecolor=gray]{3,2}{2,4}
\ncline[linecolor=gray]{3,3}{1,2}
\ncline{2,1}{1,2}
\ncline{3,2}{3,3}
\ncline{1,2}{1,3}
\ncline[linecolor=red]{1,3}{2,4}
\ncline[linecolor=gray]{2,4}{1,2}
\ncline[linecolor=red]{2,4}{3,3}
\end{psmatrix}
\end{pspicture}
\caption{Non-complete cycle $C$, $\lbrace j,k \rbrace \notin E$, $p(l^*)=p(V(C))$.}
\label{figNon-CompleteCycle-a-priorities}
\end{figure}
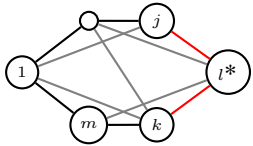
We get
\[
\sum_{i \in S} \omega_i^T (v^G(S) - v^G(S \setminus \lbrace i \rbrace)) = 
\omega_{j}^T + 2 \omega_{l^*}^T + \omega_{k}^T
>
\omega_{j}^T + \omega_{l^*}^T + \omega_{k}^T = 
\sum_{i \in S} \omega_i^T (v^G(T) - v^G(T \setminus \lbrace i \rbrace)).
\]
This contradicts $(\omega, \Sigma)$-convexity of the communication game $(N,v^G)$
as $p(l^*) = p(T)$.
\end{example}

\subsubsection{\texorpdfstring{$3$}{3}-pan}

The aim of this section is to exhibit a counter-example for the $3$-pan as depicted
in Figure~\ref{3-panWithw_4>Epsilon} under the following condition on the priorities:
\begin{equation}
\label{eqp(2)>=Max(p(1),p(4))>=p(3)}
p(2) \geq \max(p(1),p(4)) \geq p(3).
\end{equation}
The counter-example is the following.

\begin{example}[$3$-pan]
\label{Example3-pan-NonValidFor4-path}

\rm
Let us assume that there are only $4$ nodes denoted by $\{1,2,3,4\}$.
The graph is a $3$-pan such that $4$ has degree $2$ as represented in Figure~\ref{3-panWithw_4>Epsilon-bis}.
\def\ech{.5}
\begin{figure}[!ht]
\begin{center}
\begin{tikzpicture}[scale=\ech,shorten >=1pt,auto,node distance=4cm,thick,main
 node/.style={circle,draw,font=\Large\bfseries}]
\node [draw,text width=0.5cm,text centered,circle,scale=\ech] (A) at (-4,-2) {$4$};
\node [draw,text width=0.5cm,text centered,circle,scale=\ech] (B) at (-4,2) {$1$};
\node [draw,text width=0.5cm,text centered,circle,scale=\ech] (C) at (-2,0) {$2$};
\node [draw,text width=0.5cm,text centered,circle,scale=\ech] (D) at (0,0) {$3$};
\draw[-,>=latex,color=red,scale=\ech] (A) to node[midway,color=red,scale=\ech,above] {}(B);
\draw[-,>=latex,color=red,scale=\ech] (A) to node[midway,color=red,scale=\ech,above] {}(C);
\draw[-,>=latex,color=red,scale=\ech] (B) to node[midway,color=red,scale=\ech,above] {}(C);
\draw[-,>=latex,color=red,scale=\ech] (C) to node[midway,color=red,scale=\ech,above] {}(D);
\end{tikzpicture}
\caption{}
\label{3-panWithw_4>Epsilon-bis}
\end{center}
\end{figure}
We assume 
\[
p(2) \geq p(4),\ p(4)\geq p(1) \text{ and } p(4)\geq p(3).
\]
Let us prove that there is no conservation of $(\omega,\Sigma)$-convexity.
Let us define 
\begin{align*}
X & = \max \left(1+\frac{\omega_2}{\omega_4}, 1+\frac{\omega_3}{\omega_4} \right),\\
Y & = 1+\frac{\omega_1}{\omega_4},\\
Z & = X+2Y + (1-\alpha_p) \left\lbrack \frac{\omega_1}{\omega_2 + \omega_3 + \omega_4}X - \frac{\omega_1}{\omega_4}\right\rbrack, \\
\Theta &= Z+X-1,
\end{align*}
where
\[
\alpha_p =
\left\lbrace
\begin{array}{ll}
1 & \textrm{ if } p(2) = p(4) > p(3),\\
0 & \textrm{ otherwise}.
\end{array}
\right.
\]
We consider the following game $(N,v)$ and the communication game $(N,v^G)$ induced by $(N,v)$ and the communication structure given by $G$. 
\[
\begin{array}{ccc}
v(S)=
\begin{cases}
\alpha_p (X-1) \text{ if } S=\{1,2,3\},\\
X \text{ if } S=\{1,4\},\\
\alpha_p Y \text{ if } S = \{2,4\},\\
Y \text{ if } S = \{3,4\},\\
X+ \alpha_p(Y-1) \text{ if } S = \{1,2,4\},\\
X+Y-1 \text{ if } S = \{1,3,4\},\\
Z \text { if } S=\{2,3,4\},\\
\Theta \text{ if } S=N,\\
0 \text{ otherwise}.
\end{cases}
& &
v^G(S)=
\begin{cases}
\alpha_p (X-1) \text{ if } S=\{1,2,3\},\\
X \text{ if } S=\{1,4\},\\
\alpha_p Y \text{ if } S = \{2,4\},\\
0 \text{ if } S = \{3,4\},\\
X+ \alpha_p(Y-1) \text{ if } S = \{1,2,4\},\\
X \text{ if } S = \{1,3,4\},\\
Z \text { if } S=\{2,3,4\},\\
\Theta \text{ if } S=N,\\
0 \text{ otherwise}.
\end{cases}
\end{array}
\]
Let us note that this game corresponds exactly to the game given in Example~\ref{Example3-pan} if $\alpha_p = 0$.
We represent in Figure~\ref{LatticeFirstCounterExample}
the lattices of coalitions of $\lbrace 1, 2, 3, 4 \rbrace$ ordered by inclusion for $(N,v)$ and $(N,v^G)$
with the corresponding values.
Let us establish a few facts.
\def\ech{.6}
\begin{figure}[!ht]
\begin{subfigure}[c]{0.45\textwidth}
\begin{center}
\begin{tikzpicture}[scale=\ech,shorten >=1pt,auto,node distance=3cm,thick,main
 node/.style={circle,draw,font=\Large\bfseries}]
\node [draw,text width=0.5cm,text centered,circle,scale=\ech,label=right:{\scriptsize $0$}] (V) at (0,-2) {$\emptyset$};
\node [draw,text width=0.5cm,text centered,circle,scale=\ech,label=right:{\scriptsize $0$}] (1) at (-3,0) {$1$};
\node [draw,text width=0.5cm,text centered,circle,scale=\ech,label=right:{\scriptsize $0$}] (2) at (-1,0) {$2$};
\node [draw,text width=0.5cm,text centered,circle,scale=\ech,label=right:{\scriptsize $0$}] (3) at (1,0) {$3$};
\node [draw,text width=0.5cm,text centered,circle,scale=\ech,label=right:{\scriptsize $0$}] (4) at (3,0) {$4$};
\node [draw,text width=0.5cm,text centered,circle,scale=\ech,label=right:{\scriptsize $0$}] (12) at (-5,2) {$12$};
\node [draw,text width=0.5cm,text centered,circle,scale=\ech,label=right:{\scriptsize $0$}] (13) at (-3,2) {$13$};
\node [draw,text width=0.5cm,text centered,circle,scale=\ech,label=right:{\scriptsize $X$}] (14) at (-1,2) {$14$};
\node [draw,text width=0.5cm,text centered,circle,scale=\ech,label=right:{\scriptsize $0$}] (23) at (1,2) {$23$};
\node [draw,text width=0.5cm,text centered,circle,scale=\ech,label={[label distance=-.1cm]2:\scriptsize $\alpha_p Y$}](24) at (3,2) {$24$};
\node [draw,text width=0.5cm,text centered,circle,scale=\ech,label=right:{\scriptsize $Y$}] (34) at (5,2) {$34$};
\node [draw,text width=0.5cm,text centered,circle,scale=\ech,label=left:{\scriptsize $\alpha_p(X-1)$}] (123) at (-3,4) {$123$};
\node [draw,text width=0.5cm,text centered,circle,scale=\ech,label={[label distance=-.3cm]2:\tiny $\begin{array}{l}\\\\X+\\\alpha_p(Y-1)\end{array}$}] (124) at (-1,4) {$124$};
\node [draw,text width=0.5cm,text centered,circle,scale=\ech,label={[label distance=-.3cm]2:\tiny $\begin{array}{l}X+\\Y-1\end{array}$}] (134) at (1,4) {$134$};
\node [draw,text width=0.5cm,text centered,circle,scale=\ech,label=right:{\scriptsize $Z$}] (234) at (3,4) {$234$};
\node [draw,text width=0.5cm,text centered,circle,scale=\ech,label=right:{\scriptsize $\Theta$}] (1234) at (0,6) {$N$};
\draw[-,>=latex,color=red,scale=\ech] (V) to node[midway,color=red,scale=\ech,above] {}(1);
\draw[-,>=latex,color=red,scale=\ech] (V) to node[midway,color=red,scale=\ech,above] {}(2);
\draw[-,>=latex,color=red,scale=\ech] (V) to node[midway,color=red,scale=\ech,above] {}(3);
\draw[-,>=latex,color=red,scale=\ech] (V) to node[midway,color=red,scale=\ech,above] {}(4);
\draw[-,>=latex,color=red,scale=\ech] (1) to node[midway,color=red,scale=\ech,above] {}(12);
\draw[-,>=latex,color=red,scale=\ech] (1) to node[midway,color=red,scale=\ech,above] {}(13);
\draw[-,>=latex,color=red,scale=\ech] (1) to node[midway,color=red,scale=\ech,above] {}(14);
\draw[-,>=latex,color=red,scale=\ech] (2) to node[midway,color=red,scale=\ech,above] {}(12);
\draw[-,>=latex,color=red,scale=\ech] (2) to node[midway,color=red,scale=\ech,above] {}(23);
\draw[-,>=latex,color=red,scale=\ech] (2) to node[midway,color=red,scale=\ech,above] {}(24);
\draw[-,>=latex,color=red,scale=\ech] (3) to node[midway,color=red,scale=\ech,above] {}(13);
\draw[-,>=latex,color=red,scale=\ech] (3) to node[midway,color=red,scale=\ech,above] {}(23);
\draw[-,>=latex,color=red,scale=\ech] (3) to node[midway,color=red,scale=\ech,above] {}(34);
\draw[-,>=latex,color=red,scale=\ech] (4) to node[midway,color=red,scale=\ech,above] {}(14);
\draw[-,>=latex,color=red,scale=\ech] (4) to node[midway,color=red,scale=\ech,above] {}(24);
\draw[-,>=latex,color=red,scale=\ech] (4) to node[midway,color=red,scale=\ech,above] {}(34);
\draw[-,>=latex,color=red,scale=\ech] (12) to node[midway,color=red,scale=\ech,above] {}(123);
\draw[-,>=latex,color=red,scale=\ech] (12) to node[midway,color=red,scale=\ech,above] {}(124);
\draw[-,>=latex,color=red,scale=\ech] (13) to node[midway,color=red,scale=\ech,above] {}(123);
\draw[-,>=latex,color=red,scale=\ech] (13) to node[midway,color=red,scale=\ech,above] {}(134);
\draw[-,>=latex,color=red,scale=\ech] (14) to node[midway,color=red,scale=\ech,above] {}(124);
\draw[-,>=latex,color=red,scale=\ech] (14) to node[midway,color=red,scale=\ech,above] {}(134);
\draw[-,>=latex,color=red,scale=\ech] (23) to node[midway,color=red,scale=\ech,above] {}(123);
\draw[-,>=latex,color=red,scale=\ech] (23) to node[midway,color=red,scale=\ech,above] {}(234);
\draw[-,>=latex,color=red,scale=\ech] (24) to node[midway,color=red,scale=\ech,above] {}(124);
\draw[-,>=latex,color=red,scale=\ech] (24) to node[midway,color=red,scale=\ech,above] {}(234);
\draw[-,>=latex,color=red,scale=\ech] (34) to node[midway,color=red,scale=\ech,above] {}(134);
\draw[-,>=latex,color=red,scale=\ech] (34) to node[midway,color=red,scale=\ech,above] {}(234);
\draw[-,>=latex,color=red,scale=\ech] (1234) to node[midway,color=red,scale=\ech,above] {}(123);
\draw[-,>=latex,color=red,scale=\ech] (1234) to node[midway,color=red,scale=\ech,above] {}(124);
\draw[-,>=latex,color=red,scale=\ech] (1234) to node[midway,color=red,scale=\ech,above] {}(134);
\draw[-,>=latex,color=red,scale=\ech] (1234) to node[midway,color=red,scale=\ech,above] {}(234);
\end{tikzpicture}
\end{center}
\caption{$(N,v)$.}
\label{LatticeFirstCounterExample-1}
\end{subfigure}
\hfill
\begin{subfigure}[c]{0.45\textwidth}
\begin{center}
\begin{tikzpicture}[scale=\ech,shorten >=1pt,auto,node distance=3cm,thick,main
 node/.style={circle,draw,font=\Large\bfseries}]
\node [draw,text width=0.5cm,text centered,circle,scale=\ech,label=right:{\scriptsize $0$}] (V) at (0,-2) {$\emptyset$};
\node [draw,text width=0.5cm,text centered,circle,scale=\ech,label=right:{\scriptsize $0$}] (1) at (-3,0) {$1$};
\node [draw,text width=0.5cm,text centered,circle,scale=\ech,label=right:{\scriptsize $0$}] (2) at (-1,0) {$2$};
\node [draw,text width=0.5cm,text centered,circle,scale=\ech,label=right:{\scriptsize $0$}] (3) at (1,0) {$3$};
\node [draw,text width=0.5cm,text centered,circle,scale=\ech,label=right:{\scriptsize $0$}] (4) at (3,0) {$4$};
\node [draw,text width=0.5cm,text centered,circle,scale=\ech,label=right:{\scriptsize $0$}] (12) at (-5,2) {$12$};
\node [draw,text width=0.5cm,text centered,circle,scale=\ech,label=right:{\scriptsize $0$}] (13) at (-3,2) {$13$};
\node [draw,text width=0.5cm,text centered,circle,scale=\ech,label=right:{\scriptsize $X$}] (14) at (-1,2) {$14$};
\node [draw,text width=0.5cm,text centered,circle,scale=\ech,label=right:{\scriptsize $0$}] (23) at (1,2) {$23$};
\node [draw,text width=0.5cm,text centered,circle,scale=\ech,label={[label distance=-.1cm]2:\scriptsize $\alpha_p Y$}] (24) at (3,2) {$24$};
\node [draw,text width=0.5cm,text centered,circle,scale=\ech,label=right:{\scriptsize $0$}] (34) at (5,2) {$34$};
\node [draw,text width=0.5cm,text centered,circle,scale=\ech,label=left:{\scriptsize $\alpha_p(X-1)$}] (123) at (-3,4) {$123$};
\node [draw,text width=0.5cm,text centered,circle,scale=\ech,label={[label distance=-.3cm]2:\tiny $\begin{array}{l}\\\\X+\\\alpha_p(Y-1)\end{array}$}] (124) at (-1,4) {$124$};
\node [draw,text width=0.5cm,text centered,circle,scale=\ech,label=right:{\scriptsize $X$}] (134) at (1,4) {$134$};
\node [draw,text width=0.5cm,text centered,circle,scale=\ech,label=right:{\scriptsize $Z$}] (234) at (3,4) {$234$};
\node [draw,text width=0.5cm,text centered,circle,scale=\ech,label=right:{\scriptsize $\Theta$}] (1234) at (0,6) {$N$};
\draw[-,>=latex,color=red,scale=\ech] (V) to node[midway,color=red,scale=\ech,above] {}(1);
\draw[-,>=latex,color=red,scale=\ech] (V) to node[midway,color=red,scale=\ech,above] {}(2);
\draw[-,>=latex,color=red,scale=\ech] (V) to node[midway,color=red,scale=\ech,above] {}(3);
\draw[-,>=latex,color=red,scale=\ech] (V) to node[midway,color=red,scale=\ech,above] {}(4);
\draw[-,>=latex,color=red,scale=\ech] (1) to node[midway,color=red,scale=\ech,above] {}(12);
\draw[-,>=latex,color=red,scale=\ech] (1) to node[midway,color=red,scale=\ech,above] {}(13);
\draw[-,>=latex,color=red,scale=\ech] (1) to node[midway,color=red,scale=\ech,above] {}(14);
\draw[-,>=latex,color=red,scale=\ech] (2) to node[midway,color=red,scale=\ech,above] {}(12);
\draw[-,>=latex,color=red,scale=\ech] (2) to node[midway,color=red,scale=\ech,above] {}(23);
\draw[-,>=latex,color=red,scale=\ech] (2) to node[midway,color=red,scale=\ech,above] {}(24);
\draw[-,>=latex,color=red,scale=\ech] (3) to node[midway,color=red,scale=\ech,above] {}(13);
\draw[-,>=latex,color=red,scale=\ech] (3) to node[midway,color=red,scale=\ech,above] {}(23);
\draw[-,>=latex,color=red,scale=\ech] (3) to node[midway,color=red,scale=\ech,above] {}(34);
\draw[-,>=latex,color=red,scale=\ech] (4) to node[midway,color=red,scale=\ech,above] {}(14);
\draw[-,>=latex,color=red,scale=\ech] (4) to node[midway,color=red,scale=\ech,above] {}(24);
\draw[-,>=latex,color=red,scale=\ech] (4) to node[midway,color=red,scale=\ech,above] {}(34);
\draw[-,>=latex,color=red,scale=\ech] (12) to node[midway,color=red,scale=\ech,above] {}(123);
\draw[-,>=latex,color=red,scale=\ech] (12) to node[midway,color=red,scale=\ech,above] {}(124);
\draw[-,>=latex,color=red,scale=\ech] (13) to node[midway,color=red,scale=\ech,above] {}(123);
\draw[-,>=latex,color=red,scale=\ech] (13) to node[midway,color=red,scale=\ech,above] {}(134);
\draw[-,>=latex,color=red,scale=\ech] (14) to node[midway,color=red,scale=\ech,above] {}(124);
\draw[-,>=latex,color=red,scale=\ech] (14) to node[midway,color=red,scale=\ech,above] {}(134);
\draw[-,>=latex,color=red,scale=\ech] (23) to node[midway,color=red,scale=\ech,above] {}(123);
\draw[-,>=latex,color=red,scale=\ech] (23) to node[midway,color=red,scale=\ech,above] {}(234);
\draw[-,>=latex,color=red,scale=\ech] (24) to node[midway,color=red,scale=\ech,above] {}(124);
\draw[-,>=latex,color=red,scale=\ech] (24) to node[midway,color=red,scale=\ech,above] {}(234);
\draw[-,>=latex,color=red,scale=\ech] (34) to node[midway,color=red,scale=\ech,above] {}(134);
\draw[-,>=latex,color=red,scale=\ech] (34) to node[midway,color=red,scale=\ech,above] {}(234);
\draw[-,>=latex,color=red,scale=\ech] (1234) to node[midway,color=red,scale=\ech,above] {}(123);
\draw[-,>=latex,color=red,scale=\ech] (1234) to node[midway,color=red,scale=\ech,above] {}(124);
\draw[-,>=latex,color=red,scale=\ech] (1234) to node[midway,color=red,scale=\ech,above] {}(134);
\draw[-,>=latex,color=red,scale=\ech] (1234) to node[midway,color=red,scale=\ech,above] {}(234);
\end{tikzpicture}
\end{center}
\caption{$(N,v^G)$.}
\label{LatticeFirstCounterExample-2}
\end{subfigure}
\caption{}
\label{LatticeFirstCounterExample}
\end{figure}

\noindent
\underline{The game $(N,v)$ is not convex but is $(\omega,\Sigma)$-convex}

A detailed proof for the general case with priorities is provided in Appendix~\ref{Append-Counter-Example-3-pan-Non-Valid-For-4-path}.

\noindent
\underline{The communication game $(N,v^G)$ on the $3$-pan is not $(\omega,\Sigma)$-convex:}

Let us insist on the coalitions whose value changed compared to the original game $(N,v)$
\begin{itemize}
\item the value of $\{3,4\}$ changes from $Y$ to $0$ hence loosing $Y$,
\item the value of $\{1,3,4\}$ changes from $X+Y-1$ to $X$ hence loosing  $Y-1$. 
\end{itemize}
The key element is to notice that the value of $\{1,3,4\}$ is loosing less than the value of $\{3,4\}$
which allows for a contradiction to $(\omega,\Sigma)$-convexity.
Let us consider $S=\{2,3,4\}$ and $T=N$:
\begin{eqnarray}
\label{eqSumiInSOmegaiT(vG(S)-vG(S-i))}
\sum_{i\in S} \omega_i^T (v^G(S)-v^G(S \setminus{i})) & = & \omega_2^T (Z-0) + \omega_3^T (Z-\alpha_pY) + \omega_4^T (Z-0),
\end{eqnarray}
whereas
\begin{eqnarray}
\label{eqSumiInSOmegaiT(vG(T)-vG(T-i))}
\sum_{i\in S} \omega_i^T (v^G(T)-v^G(T \setminus{i}))
& = & \omega_2^T (\Theta-X) + \omega_3^T (\Theta-X-\alpha_p(Y-1)) + \omega_4^T (\Theta-\alpha_p(X-1)) \nonumber \\
& = & \omega_2^T (Z-1) + \omega_3^T (Z-\alpha_p Y - (1- \alpha_p))\nonumber \\
& & + \omega_4^T (Z + (1 -\alpha_p)(X-1)).
\end{eqnarray}
If $\alpha_p = 1$,
we get a contradiction as $2$ has maximal priority.
If $\alpha_p = 0$, then (\ref{eqSumiInSOmegaiT(vG(S)-vG(S-i))}) and (\ref{eqSumiInSOmegaiT(vG(T)-vG(T-i))}) are equivalent to
\begin{eqnarray}
\label{eqSumiInSOmegaiT(vG(S)-vG(S-i))-Alpha_p}
\sum_{i\in S} \omega_i^T (v^G(S)-v^G(S \setminus{i})) & = & (\omega_2^T + \omega_3^T + \omega_4^T) Z,
\end{eqnarray}
and
\begin{eqnarray}
\label{eqSumiInSOmegaiT(vG(T)-vG(T-i))-Alpha_p}
\sum_{i\in S} \omega_i^T (v^G(T)-v^G(T \setminus{i}))
& = & (\omega_2^T + \omega_3^T + \omega_4^T) Z - \omega_2^T - \omega_3^T + \omega_4^T (X-1).
\end{eqnarray}
Let us first assume $p(4) < p(N)$.
Then we have $\omega_4^T = 0$ and get a contradiction as $p(2) = p(N)$.
Let us now assume $p(4) = p(N)$.
By assumption and as $\alpha_p = 0$,
we also have $p(3) = p(N)$.
Then, by definition of $X$, we get
either
$- \omega_2^T + \omega_4^T (X-1)= - \omega_2 + \omega_4 (X-1) = 0$
or 
$- \omega_3^T + \omega_4^T (X-1)= - \omega_3 + \omega_4 (X-1) = 0$.
In both cases we still obtain a contradiction.
\end{example}

By symmetry
of $1$ and $4$,
interverting $1$ and $4$ in Example~\ref{Example3-pan-NonValidFor4-path},
we also get a contradiction if $\max(p(1),p(4)) = p(1)$.
Thus we have a contradiction on the $3$-pan
for any weight system satisfying (\ref{eqp(2)>=Max(p(1),p(4))>=p(3)}).

\begin{remark}
\label{remarkExtensionofExample5-panTo4-path}
Example~\ref{Example3-pan-NonValidFor4-path} can be extended to the $4$-path
if $\alpha_p = 0$ but not if $\alpha_p = 1$.
Let us consider the $4$-path
corresponding to $\lbrace 1,4,2,3 \rbrace$
as represented in Figure~\ref{4-path1423-2}.
Starting from a game $(N,v)$,
the only differences between the communication games associated with the specified $3$-pan and $4$-path
are the values of $\{1,2,3\}$ and $\lbrace 1,2 \rbrace$.
In the $4$-path,
the value of coalition $\{1,2,3\}$
(resp. $\lbrace 1,2 \rbrace$)
is equal to $v(\{1\})+v(\{2,3\}) = 0$
(resp. $v(\{1\})+v(\{2\}) = 0$).
(\ref{eqSumiInSOmegaiT(vG(S)-vG(S-i))}) is not modified and (\ref{eqSumiInSOmegaiT(vG(T)-vG(T-i))}) becomes
\begin{eqnarray}
\label{eqSumiInSOmegaiT(vG(T)-vG(T-i))-Alpha_p-2}
\sum_{i\in S} \omega_i^T (v^G(T)-v^G(T \setminus{i}))
& = & \omega_2^T (\Theta-X) + \omega_3^T (\Theta-X-\alpha_p(Y-1)) + \omega_4^T (\Theta-0) \nonumber \\
& = & \omega_2^T (Z-1) + \omega_3^T (Z-\alpha_p Y - (1- \alpha_p)) \nonumber\\
& & + \omega_4^T (Z+X-1).
\end{eqnarray}
If $\alpha_p = 0$,
then (\ref{eqSumiInSOmegaiT(vG(S)-vG(S-i))}) and (\ref{eqSumiInSOmegaiT(vG(T)-vG(T-i))-Alpha_p-2})
are equivalent to (\ref{eqSumiInSOmegaiT(vG(S)-vG(S-i))-Alpha_p}) and (\ref{eqSumiInSOmegaiT(vG(T)-vG(T-i))-Alpha_p})
and we get the same contradiction as before.
If $\alpha_p = 1$,
then we would need $\omega_2 > \omega_4(X-1)$ to get a contradiction.
This is not possible as $X \geq 1 + \frac{\omega_2}{\omega_4}$ by definition.
\end{remark}

\subsubsection{\texorpdfstring{$4$}{4}-path}

The aim of this section is to exhibit a counter-example for the $4$-path as represented
in Figure~\ref{4-path1423-2}
\begin{figure}[!ht]
\centering{
\begin{tikzpicture}[scale=\ech,shorten >=1pt,auto,node distance=3cm,thick,main
 node/.style={circle,draw,font=\Large\bfseries}]
\node [draw,text width=0.5cm,text centered,circle,scale=\ech] (A) at (-6,0) {$1$};
\node [draw,text width=0.5cm,text centered,circle,scale=\ech] (B) at (-4,0) {$4$};
\node [draw,text width=0.5cm,text centered,circle,scale=\ech] (C) at (-2,0) {$2$};
\node [draw,text width=0.5cm,text centered,circle,scale=\ech] (D) at (0,0) {$3$};
\draw[-,>=latex,color=red,scale=\ech] (A) to node[midway,color=red,scale=\ech,above] {}(B);
\draw[-,>=latex,color=red,scale=\ech] (B) to node[midway,color=red,scale=\ech,above] {}(C);
\draw[-,>=latex,color=red,scale=\ech] (C) to node[midway,color=red,scale=\ech,above] {}(D);
\end{tikzpicture}
}
\caption{$4$-path $\lbrace 1,4,2,3 \rbrace$.}
\label{4-path1423-2}
\end{figure}
under the following condition on the priorities:
\begin{equation}
\label{eqMin(p(2),p(4))>=Max(p(1),p(3))}
\min(p(2),p(4)) \geq \max(p(1),p(3)).
\end{equation}
Depending on the precise order of priorities, several counter-examples are necessary.
\begin{itemize}
\item
Assume that 
\[
p(2)> p(4) \geq \max(p(1),p(3)).
\]
By Remark~\ref{remarkExtensionofExample5-panTo4-path} and as $\alpha_p = 0$,
Example~\ref{Example3-pan-NonValidFor4-path} provides a counter-example.

\item Assume that 
\[
p(2) =p(4) =\max(p(1),p(3)).
\]
We have two cases that are similar up to permutation of the roles.
If $p(3) \geq p(1)$,
then Example~\ref{Example3-pan-NonValidFor4-path} still provides a counterexample
by Remark~\ref{remarkExtensionofExample5-panTo4-path} and as $\alpha_p = 0$.
If $p(1) > p(3)$,
one can intervert at the same time the role of $2$ and $4$ and the role of $1$ and $3$ to obtain a counterexample.
\item Finally, assume that
\[
p(2)= p(4) > \max(p(1),p(3)).
\]
One needs a new counter-example. Example~\ref{Example4-path-strict} will be our counter-example.
\end{itemize}
By symmetry of $2$ and $4$, one obtains a counterexample for any weight system satisfying~(\ref{eqMin(p(2),p(4))>=Max(p(1),p(3))}).

\begin{example}
\label{Example4-path-strict}
\rm
We consider the following game $(N,v)$ and the communication game $(N,v^G)$ induced by $(N,v)$ and the communication structure given by $G$. 
\[
v(S)=\begin{cases}
1 \text{ if } S \in
\{ \{1,4\}, \{3,4\}, \{1,2,4\},\{1,3,4\}, \{2,3,4\}, N  \},\\
0 \text{ otherwise}.
\end{cases}
\]
\[
v^G(S)=\begin{cases}
1 \text{ if } S \in \{ \{1,4\}, \{1,2,4\}, \{1,3,4\}, \{2,3,4\}, N  \},\\
0 \text{ otherwise}.
\end{cases}
\]
We represent in Figure~\ref{LatticeSecondCounterExample-2} the lattices of coalitions of $\lbrace 1, 2, 3, 4 \rbrace$ ordered by inclusion
for $(N,v)$ and $(N,v^G)$ with the corresponding values.
\begin{figure}[!ht]
\begin{subfigure}[c]{0.45\textwidth}
\begin{center}
\begin{tikzpicture}[scale=\ech,shorten >=1pt,auto,node distance=3cm,thick,main
 node/.style={circle,draw,font=\Large\bfseries}]
\node [draw,text width=0.5cm,text centered,circle,scale=\ech,label=right:{\scriptsize $0$}] (V) at (0,-2) {$\emptyset$};
\node [draw,text width=0.5cm,text centered,circle,scale=\ech,label=right:{\scriptsize $0$}] (1) at (-3,0) {$1$};
\node [draw,text width=0.5cm,text centered,circle,scale=\ech,label=right:{\scriptsize $0$}] (2) at (-1,0) {$2$};
\node [draw,text width=0.5cm,text centered,circle,scale=\ech,label=right:{\scriptsize $0$}] (3) at (1,0) {$3$};
\node [draw,text width=0.5cm,text centered,circle,scale=\ech,label=right:{\scriptsize $0$}] (4) at (3,0) {$4$};
\node [draw,text width=0.5cm,text centered,circle,scale=\ech,label=right:{\scriptsize $0$}] (12) at (-5,2) {$12$};
\node [draw,text width=0.5cm,text centered,circle,scale=\ech,label=right:{\scriptsize $0$}] (13) at (-3,2) {$13$};
\node [draw,text width=0.5cm,text centered,circle,scale=\ech,label=right:{\scriptsize $1$}] (14) at (-1,2) {$14$};
\node [draw,text width=0.5cm,text centered,circle,scale=\ech,label=right:{\scriptsize $0$}] (23) at (1,2) {$23$};
\node [draw,text width=0.5cm,text centered,circle,scale=\ech,label=right:{\scriptsize $0$}] (24) at (3,2) {$24$};
\node [draw,text width=0.5cm,text centered,circle,scale=\ech,label=right:{\scriptsize $1$}] (34) at (5,2) {$34$};
\node [draw,text width=0.5cm,text centered,circle,scale=\ech,label=right:{\scriptsize $0$}] (123) at (-3,4) {$123$};
\node [draw,text width=0.5cm,text centered,circle,scale=\ech,label=right:{\scriptsize $1$}] (124) at (-1,4) {$124$};
\node [draw,text width=0.5cm,text centered,circle,scale=\ech,label=right:{\scriptsize $1$}] (134) at (1,4) {$134$};
\node [draw,text width=0.5cm,text centered,circle,scale=\ech,label=right:{\scriptsize $1$}] (234) at (3,4) {$234$};
\node [draw,text width=0.5cm,text centered,circle,scale=\ech,label=right:{\scriptsize $1$}] (1234) at (0,6) {$N$};
\draw[-,>=latex,color=red,scale=\ech] (V) to node[midway,color=red,scale=\ech,above] {}(1);
\draw[-,>=latex,color=red,scale=\ech] (V) to node[midway,color=red,scale=\ech,above] {}(2);
\draw[-,>=latex,color=red,scale=\ech] (V) to node[midway,color=red,scale=\ech,above] {}(3);
\draw[-,>=latex,color=red,scale=\ech] (V) to node[midway,color=red,scale=\ech,above] {}(4);
\draw[-,>=latex,color=red,scale=\ech] (1) to node[midway,color=red,scale=\ech,above] {}(12);
\draw[-,>=latex,color=red,scale=\ech] (1) to node[midway,color=red,scale=\ech,above] {}(13);
\draw[-,>=latex,color=red,scale=\ech] (1) to node[midway,color=red,scale=\ech,above] {}(14);
\draw[-,>=latex,color=red,scale=\ech] (2) to node[midway,color=red,scale=\ech,above] {}(12);
\draw[-,>=latex,color=red,scale=\ech] (2) to node[midway,color=red,scale=\ech,above] {}(23);
\draw[-,>=latex,color=red,scale=\ech] (2) to node[midway,color=red,scale=\ech,above] {}(24);
\draw[-,>=latex,color=red,scale=\ech] (3) to node[midway,color=red,scale=\ech,above] {}(13);
\draw[-,>=latex,color=red,scale=\ech] (3) to node[midway,color=red,scale=\ech,above] {}(23);
\draw[-,>=latex,color=red,scale=\ech] (3) to node[midway,color=red,scale=\ech,above] {}(34);
\draw[-,>=latex,color=red,scale=\ech] (4) to node[midway,color=red,scale=\ech,above] {}(14);
\draw[-,>=latex,color=red,scale=\ech] (4) to node[midway,color=red,scale=\ech,above] {}(24);
\draw[-,>=latex,color=red,scale=\ech] (4) to node[midway,color=red,scale=\ech,above] {}(34);
\draw[-,>=latex,color=red,scale=\ech] (12) to node[midway,color=red,scale=\ech,above] {}(123);
\draw[-,>=latex,color=red,scale=\ech] (12) to node[midway,color=red,scale=\ech,above] {}(124);
\draw[-,>=latex,color=red,scale=\ech] (13) to node[midway,color=red,scale=\ech,above] {}(123);
\draw[-,>=latex,color=red,scale=\ech] (13) to node[midway,color=red,scale=\ech,above] {}(134);
\draw[-,>=latex,color=red,scale=\ech] (14) to node[midway,color=red,scale=\ech,above] {}(124);
\draw[-,>=latex,color=red,scale=\ech] (14) to node[midway,color=red,scale=\ech,above] {}(134);
\draw[-,>=latex,color=red,scale=\ech] (23) to node[midway,color=red,scale=\ech,above] {}(123);
\draw[-,>=latex,color=red,scale=\ech] (23) to node[midway,color=red,scale=\ech,above] {}(234);
\draw[-,>=latex,color=red,scale=\ech] (24) to node[midway,color=red,scale=\ech,above] {}(124);
\draw[-,>=latex,color=red,scale=\ech] (24) to node[midway,color=red,scale=\ech,above] {}(234);
\draw[-,>=latex,color=red,scale=\ech] (34) to node[midway,color=red,scale=\ech,above] {}(134);
\draw[-,>=latex,color=red,scale=\ech] (34) to node[midway,color=red,scale=\ech,above] {}(234);
\draw[-,>=latex,color=red,scale=\ech] (1234) to node[midway,color=red,scale=\ech,above] {}(123);
\draw[-,>=latex,color=red,scale=\ech] (1234) to node[midway,color=red,scale=\ech,above] {}(124);
\draw[-,>=latex,color=red,scale=\ech] (1234) to node[midway,color=red,scale=\ech,above] {}(134);
\draw[-,>=latex,color=red,scale=\ech] (1234) to node[midway,color=red,scale=\ech,above] {}(234);
\end{tikzpicture}
\end{center}
\caption{$(N,v)$.}
\label{LatticeSecondCounterExample-2-1}
\end{subfigure}
\hfill
\begin{subfigure}[c]{0.45\textwidth}
\begin{center}
\begin{tikzpicture}[scale=\ech,shorten >=1pt,auto,node distance=3cm,thick,main
 node/.style={circle,draw,font=\Large\bfseries}]
\node [draw,text width=0.5cm,text centered,circle,scale=\ech,label=right:{\scriptsize $0$}] (V) at (0,-2) {$\emptyset$};
\node [draw,text width=0.5cm,text centered,circle,scale=\ech,label=right:{\scriptsize $0$}] (1) at (-3,0) {$1$};
\node [draw,text width=0.5cm,text centered,circle,scale=\ech,label=right:{\scriptsize $0$}] (2) at (-1,0) {$2$};
\node [draw,text width=0.5cm,text centered,circle,scale=\ech,label=right:{\scriptsize $0$}] (3) at (1,0) {$3$};
\node [draw,text width=0.5cm,text centered,circle,scale=\ech,label=right:{\scriptsize $0$}] (4) at (3,0) {$4$};
\node [draw,text width=0.5cm,text centered,circle,scale=\ech,label=right:{\scriptsize $0$}] (12) at (-5,2) {$12$};
\node [draw,text width=0.5cm,text centered,circle,scale=\ech,label=right:{\scriptsize $0$}] (13) at (-3,2) {$13$};
\node [draw,text width=0.5cm,text centered,circle,scale=\ech,label=right:{\scriptsize $1$}] (14) at (-1,2) {$14$};
\node [draw,text width=0.5cm,text centered,circle,scale=\ech,label=right:{\scriptsize $0$}] (23) at (1,2) {$23$};
\node [draw,text width=0.5cm,text centered,circle,scale=\ech,label=right:{\scriptsize $0$}] (24) at (3,2) {$24$};
\node [draw,text width=0.5cm,text centered,circle,scale=\ech,label=right:{\scriptsize $0$}] (34) at (5,2) {$34$};
\node [draw,text width=0.5cm,text centered,circle,scale=\ech,label=right:{\scriptsize $0$}] (123) at (-3,4) {$123$};
\node [draw,text width=0.5cm,text centered,circle,scale=\ech,label=right:{\scriptsize $1$}] (124) at (-1,4) {$124$};
\node [draw,text width=0.5cm,text centered,circle,scale=\ech,label=right:{\scriptsize $1$}] (134) at (1,4) {$134$};
\node [draw,text width=0.5cm,text centered,circle,scale=\ech,label=right:{\scriptsize $1$}] (234) at (3,4) {$234$};
\node [draw,text width=0.5cm,text centered,circle,scale=\ech,label=right:{\scriptsize $1$}] (1234) at (0,6) {$N$};
\draw[-,>=latex,color=red,scale=\ech] (V) to node[midway,color=red,scale=\ech,above] {}(1);
\draw[-,>=latex,color=red,scale=\ech] (V) to node[midway,color=red,scale=\ech,above] {}(2);
\draw[-,>=latex,color=red,scale=\ech] (V) to node[midway,color=red,scale=\ech,above] {}(3);
\draw[-,>=latex,color=red,scale=\ech] (V) to node[midway,color=red,scale=\ech,above] {}(4);
\draw[-,>=latex,color=red,scale=\ech] (1) to node[midway,color=red,scale=\ech,above] {}(12);
\draw[-,>=latex,color=red,scale=\ech] (1) to node[midway,color=red,scale=\ech,above] {}(13);
\draw[-,>=latex,color=red,scale=\ech] (1) to node[midway,color=red,scale=\ech,above] {}(14);
\draw[-,>=latex,color=red,scale=\ech] (2) to node[midway,color=red,scale=\ech,above] {}(12);
\draw[-,>=latex,color=red,scale=\ech] (2) to node[midway,color=red,scale=\ech,above] {}(23);
\draw[-,>=latex,color=red,scale=\ech] (2) to node[midway,color=red,scale=\ech,above] {}(24);
\draw[-,>=latex,color=red,scale=\ech] (3) to node[midway,color=red,scale=\ech,above] {}(13);
\draw[-,>=latex,color=red,scale=\ech] (3) to node[midway,color=red,scale=\ech,above] {}(23);
\draw[-,>=latex,color=red,scale=\ech] (3) to node[midway,color=red,scale=\ech,above] {}(34);
\draw[-,>=latex,color=red,scale=\ech] (4) to node[midway,color=red,scale=\ech,above] {}(14);
\draw[-,>=latex,color=red,scale=\ech] (4) to node[midway,color=red,scale=\ech,above] {}(24);
\draw[-,>=latex,color=red,scale=\ech] (4) to node[midway,color=red,scale=\ech,above] {}(34);
\draw[-,>=latex,color=red,scale=\ech] (12) to node[midway,color=red,scale=\ech,above] {}(123);
\draw[-,>=latex,color=red,scale=\ech] (12) to node[midway,color=red,scale=\ech,above] {}(124);
\draw[-,>=latex,color=red,scale=\ech] (13) to node[midway,color=red,scale=\ech,above] {}(123);
\draw[-,>=latex,color=red,scale=\ech] (13) to node[midway,color=red,scale=\ech,above] {}(134);
\draw[-,>=latex,color=red,scale=\ech] (14) to node[midway,color=red,scale=\ech,above] {}(124);
\draw[-,>=latex,color=red,scale=\ech] (14) to node[midway,color=red,scale=\ech,above] {}(134);
\draw[-,>=latex,color=red,scale=\ech] (23) to node[midway,color=red,scale=\ech,above] {}(123);
\draw[-,>=latex,color=red,scale=\ech] (23) to node[midway,color=red,scale=\ech,above] {}(234);
\draw[-,>=latex,color=red,scale=\ech] (24) to node[midway,color=red,scale=\ech,above] {}(124);
\draw[-,>=latex,color=red,scale=\ech] (24) to node[midway,color=red,scale=\ech,above] {}(234);
\draw[-,>=latex,color=red,scale=\ech] (34) to node[midway,color=red,scale=\ech,above] {}(134);
\draw[-,>=latex,color=red,scale=\ech] (34) to node[midway,color=red,scale=\ech,above] {}(234);
\draw[-,>=latex,color=red,scale=\ech] (1234) to node[midway,color=red,scale=\ech,above] {}(123);
\draw[-,>=latex,color=red,scale=\ech] (1234) to node[midway,color=red,scale=\ech,above] {}(124);
\draw[-,>=latex,color=red,scale=\ech] (1234) to node[midway,color=red,scale=\ech,above] {}(134);
\draw[-,>=latex,color=red,scale=\ech] (1234) to node[midway,color=red,scale=\ech,above] {}(234);
\end{tikzpicture}
\end{center}
\caption{$(N,v^G)$.}
\label{LatticeSecondCounterExample-2-2}
\end{subfigure}
\caption{}
\label{LatticeSecondCounterExample-2}
\end{figure}

We establish the following facts under the assumption that $p(2)=p(4)>\max(p(1),p(3))$.

\noindent
\underline{The game $(N,v)$ is not convex but  $(\omega,\Sigma)$-convex}:
See the proof in Appendix \ref{proof-Example4-path-strict}.

\noindent
\underline{The communication game $(N,v^G)$ is not $(\omega,\Sigma)$-convex:}

Let us insist on the coalitions whose value changed compared to the original game $(N,v)$
\begin{itemize}
\item the value of $\{3,4\}$ changes from $1$ to $0$,
\item the value of $\{1,2,3\}$
(resp. $\{1,3,4\}$)
becomes equal to the value of $\{2,3\}$
(resp. $\{1,4\}$)
and therefore stays equal to $0$
(resp. $1$).
\end{itemize}

Let us prove that $(N,v^G)$ is not $(\omega,\Sigma)$-convex.
We consider $S=\{2,3,4\}$ and $T=N$:
\begin{equation}
\label{eqSumiInSOmegaiT(vG(S)-vG(S-i))-2}
\sum_{i\in S} \omega_i^T (v^G(S)-v^G(S \setminus{i})) = \omega_2^T + \omega_3^T + \omega_4^T, \nonumber
\end{equation}
whereas
\begin{equation}
\label{eqSumiInSOmegaiT(vG(T)-vG(T-i))-2}
\sum_{i\in S} \omega_i^T (v^G(T)-v^G(T \setminus{i})) =  \omega_2^T (1 - 1) + \omega_3^T (1 - 1) + \omega_4^T (1 - 0)
= \omega_4^T.\nonumber
\end{equation}
As $p(2)=p(N)>p(3)$, we have $\omega_3^T = 0$ but $\omega_2^T \not= 0$ 
implying a contradiction to $(\omega, \Sigma)$-convexity of $(N,v^G)$.
\end{example}

\subsection{Necessary conditions}

We now establish several necessary conditions based on the previous counter-examples and the result for the singleton partition.
Example \ref{ExampleNon-CompleteCycle-priorities} implies directly the following result.

\begin{lemma}
\label{Lemma-Weighted-Non-Complete-Cycle}
Let $(\omega,\Sigma)$ be a weight system and let $G=(N,E)$ be a graph preserving  $(\omega,\Sigma)$-convexity.
Let $C$ be a cycle in $G$ and let $j$ be a node of $C$.
Let $i$ and $k$ be neighbors of $j$ in~$C$.
If $p(j)=p(V(C))$
then $\lbrace i,k \rbrace \in E$.
\end{lemma}

\begin{remark}
Let us note that if all nodes in $C$ have the same priority
then Lemma~\ref{Lemma-Weighted-Non-Complete-Cycle} implies
cycle-completeness of $C$.
\end{remark}

Based on the previous section on singleton partition,
a necessary condition is for each level of priority to be composed of disjoint complete subgraphs or stars.

\begin{proposition}
\label{Levelk-StarOrCompleteSubgraphs}
If a graph $(N,E)$ preserves the $(\omega,\Sigma)$-convexity, given a priority $k$, the set of players of priority $k$ corresponds to a collection of disconnected star/complete subgraphs.
\end{proposition}

The proof relies on showing that there exists no non-complete cycle, no $3$-pan, and no $4$-path but not directly in the original graph but in a family of graphs constructed by taking into account the different priorities.

\begin{proof}
Fix a priority $k$ and $N_k$ the set of players in layer $k$. If the connected components of $G_k$ are not complete graphs or stars, then $G_k$ contains a subgraph $\tilde{G}_k$ corresponding to a non-complete cycle or a $4$-path or a $3$-pan.
We can consider one of the games defined in Examples~\ref{ExampleNon-CompleteCycle} or~\ref{Example3-pan}
and extend it to $N$ by taking all remaining players as null players.
By Remark~\ref{remarkExtensionWithNullPlayers}, the resulting game is $(\omega,\Sigma)$-convex.
As all nodes of $\tilde{G}_k$ belong to the same layer,
we still get a contradiction to $(\omega,\Sigma)$-convexity of the communication game
as in the case of a single partition.
\end{proof}

\begin{remark}
Let us insist that when restricting to the set of players in layer $k$, one can loose the connectivity.
Hence \emph{a priori} the set of players of priority $k$ may be composed of several stars and/or complete subgraphs.
\end{remark}

Moreover the connection between two such components in different layers depends on the type of the component of larger priority.

\begin{proposition}
\label{PlayerConnectedToAHigherLevel}
Let $(\omega,\Sigma)$ be a weight system and $(N,E)$ be a graph preserving the $(\omega,\Sigma)$-convexity.
If a player $s$ of priority $k$ is connected to a player $a$ in a layer with higher priority then
one of the following statements is satisfied
\begin{itemize}
\item
$a$ is part of a complete component $K$ in the graph restricted to its priority  and $s$ is then connected to each node of $K$,
\item
$a$ is the center of a star $S$ in the graph restricted to its priority and $s$ cannot be linked to any leaf of $S$ if $\vert V(S) \vert \geq 3$.
\end{itemize}
\end{proposition}

\begin{remark}
Connected component with one or two nodes play special roles.
Indeed a connected component of size one is at the same time a star with no leaf and a complete component.
Similarly,
a connected component of size two is at the same time a star centered on one of the two nodes,
a star centered on the second one and a complete component.
\end{remark}

\begin{proof}
Let $s$ be a player of priority $k$ connected to a player $a$ of priority $k'>k$.
Let $C_a$ be the connected component of $G_{k'}$ containing $a$.
By Proposition~\ref{Levelk-StarOrCompleteSubgraphs}, it is necessarily a star or a complete component.
We can assume that $C_a$ contains at least three nodes otherwise Proposition~\ref{PlayerConnectedToAHigherLevel}
is trivially satisfied.

By contradiction,
let us assume that $a$ is not the center of a star and not part of a complete component.
Then $C_a$ is a star and $a$ is a leaf of this star.
Let $c$ be the center of $C_a$ and let $b \not= a$ be another leaf.
Then $\lbrace s,a,c,b \rbrace$ corresponds to a path in $G$
(but not necessarily a $4$-path as some more edges between nodes $s,a,c,b$ could exist).
As $C_a$ is a star, there is no edge linking $a$ and $b$.
We now distinguish several cases:
\begin{itemize}
\item
Let us assume $s$ non-connected to $b$.
If $s$ is not connected
(resp. connected)
to $c$
then $\lbrace s,a,c,b \rbrace$ corresponds to a $4$-path
(resp. $3$-pan)
as represented in Figure~\ref{fig3-panWithPriorityLevel-4-Path}
(resp. Figure~\ref{fig3-panWithPriorityLevel-3-Pan})
\begin{figure}[!ht]
\begin{subfigure}[c]{0.45\textwidth}
\centering
\begin{pspicture}(-1,-.3)(1,2)
\tiny
\begin{psmatrix}[mnode=circle,colsep=0.4,rowsep=0.8]
{$a$}	 	& {$c$} & {$b$}\\
{$s$}
\psset{arrows=-, shortput=nab,labelsep={0.08}}
\tiny
\ncline{1,1}{2,1}
\ncline{1,1}{1,2}
\ncline{1,2}{1,3}
\end{psmatrix}
\normalsize
\uput[0](-2.8,1.2){\textcolor{orange}{$k'$}}
\pspolygon[linecolor=orange,linearc=.4,linewidth=.02](-2.3,.8)(-2.3,1.6)(.2,1.6)(.2,.8)
\uput[0](-2.8,0){\textcolor{cyan}{$k$}}
\pspolygon[linecolor=cyan,linearc=.4,linewidth=.02](-2.3,-.4)(-2.3,.4)(.2,.4)(.2,-.4)
\end{pspicture}
\caption{$\lbrace s, c \rbrace \notin E$ and $\lbrace s, b \rbrace \notin E$.}
\label{fig3-panWithPriorityLevel-4-Path}
\end{subfigure}
\hfill
\begin{subfigure}[c]{0.45\textwidth}
\centering
\begin{pspicture}(-1,-.3)(1,2)
\tiny
\begin{psmatrix}[mnode=circle,colsep=0.4,rowsep=0.8]
{$a$}	 	& {$c$} & {$b$}\\
{$s$}
\psset{arrows=-, shortput=nab,labelsep={0.08}}
\tiny
\ncline{1,1}{2,1}
\ncline{1,1}{1,2}
\ncline{2,1}{1,2}
\ncline{1,2}{1,3}
\end{psmatrix}
\normalsize
\uput[0](-2.8,1.2){\textcolor{orange}{$k'$}}
\pspolygon[linecolor=orange,linearc=.4,linewidth=.02](-2.3,.8)(-2.3,1.6)(.2,1.6)(.2,.8)
\uput[0](-2.8,0){\textcolor{cyan}{$k$}}
\pspolygon[linecolor=cyan,linearc=.4,linewidth=.02](-2.3,-.4)(-2.3,.4)(.2,.4)(.2,-.4)
\end{pspicture}
\caption{$\lbrace s, b \rbrace \notin E$.}
\label{fig3-panWithPriorityLevel-3-Pan}
\end{subfigure}
\caption{$k<k'$.}
\label{fig3-panWithPriorityLevel}
\end{figure}
and one can use Example~\ref{Example3-pan-NonValidFor4-path} to get a contradiction ($s$ playing the role of $1$).

\item
Let us now assume $s$ connected to $b$.
If $s$ is not connected
(resp. connected)
to $c$,
then we have a non-complete cycle
as represented in Figure~\ref{figCycleWithPriorityLevel-1}
(resp. Figure~\ref{figCycleWithPriorityLevel-2})
\begin{figure}[!ht]
\begin{subfigure}[c]{0.45\textwidth}
\centering
\begin{pspicture}(-1,-.3)(1,2)
\tiny
\begin{psmatrix}[mnode=circle,colsep=0.4,rowsep=0.8]
{$a$}	 	& {$c$} & {$b$}\\
{$s$}
\psset{arrows=-, shortput=nab,labelsep={0.08}}
\tiny
\ncline{1,1}{2,1}
\ncline{1,1}{1,2}
\ncline{2,1}{1,3}
\ncline{1,2}{1,3}
\end{psmatrix}
\normalsize
\uput[0](-2.8,1.2){\textcolor{orange}{$k'$}}
\pspolygon[linecolor=orange,linearc=.4,linewidth=.02](-2.3,.8)(-2.3,1.6)(.2,1.6)(.2,.8)
\uput[0](-2.8,0){\textcolor{cyan}{$k$}}
\pspolygon[linecolor=cyan,linearc=.4,linewidth=.02](-2.3,-.4)(-2.3,.4)(.2,.4)(.2,-.4)
\end{pspicture}
\caption{$\lbrace s, c \rbrace \notin E$ and $\lbrace s, b \rbrace \in E$.}
\label{figCycleWithPriorityLevel-1}
\end{subfigure}
\hfill
\begin{subfigure}[c]{0.45\textwidth}
\centering
\begin{pspicture}(-1,-.3)(1,2)
\tiny
\begin{psmatrix}[mnode=circle,colsep=0.4,rowsep=0.8]
{$a$}	 	& {$c$} & {$b$}\\
{$s$}
\psset{arrows=-, shortput=nab,labelsep={0.08}}
\tiny
\ncline{1,1}{2,1}
\ncline{1,1}{1,2}
\ncline{2,1}{1,2}
\ncline{2,1}{1,3}
\ncline{1,2}{1,3}
\end{psmatrix}
\normalsize
\uput[0](-2.8,1.2){\textcolor{orange}{$k'$}}
\pspolygon[linecolor=orange,linearc=.4,linewidth=.02](-2.3,.8)(-2.3,1.6)(.2,1.6)(.2,.8)
\uput[0](-2.8,0){\textcolor{cyan}{$k$}}
\pspolygon[linecolor=cyan,linearc=.4,linewidth=.02](-2.3,-.4)(-2.3,.4)(.2,.4)(.2,-.4)
\end{pspicture}
\caption{$\lbrace s, c \rbrace \in E$ and $\lbrace s, b \rbrace \in E$.}
\label{figCycleWithPriorityLevel-2}
\end{subfigure}
\caption{$k<k'$.}
\label{figCycleWithPriorityLevel}
\end{figure}
and Example~\ref{ExampleNon-CompleteCycle-priorities} gives a contradiction considering $c$ as the key node with high priority.
\end{itemize}

Let us now turn to the case where $C_a$ is a complete component.
\begin{itemize}
\item
Assume that there is exactly one node in $C_a$ that is not linked to $s$ denoted by $t$. Let $b$ be a node of $C_a$ different from $a$ and $t$ as represented in Figure~\ref{figCompleteComponentWithPriorityLevel-1}.
Then $\lbrace s,a,t,b,s \rbrace$ corresponds to a non-complete cycle
(as $\lbrace s,t \rbrace \notin E$)
and Example~\ref{ExampleNon-CompleteCycle-priorities}
gives a contradiction with $a$ as the key node with high priority.
\begin{figure}[!ht]
\begin{subfigure}[c]{0.45\textwidth}
\centering
\begin{pspicture}(-1,-.3)(1,2.6)
\tiny
\begin{psmatrix}[mnode=circle,colsep=0.5,rowsep=0.5]
{$a$}	& {$t$}\\
{} & {$b$}\\
{$s$}
\psset{arrows=-, shortput=nab,labelsep={0.08}}
\tiny
\ncline{1,1}{2,1}
\ncline{1,1}{2,2}
\ncline{1,1}{1,2}
\ncline{2,1}{1,2}
\ncline{2,1}{2,2}
\ncline{1,2}{2,2}
\ncarc[arcangle=-35]{1,1}{3,1}
\ncline{3,1}{2,1}
\ncline{3,1}{2,2}
\end{psmatrix}
\normalsize
\uput[0](-2.4,1.2){\textcolor{orange}{$k'$}}
\pspolygon[linecolor=orange,linearc=.4,linewidth=.02](-1.8,.6)(-1.8,2.3)(.2,2.3)(.2,.6)
\uput[0](-2.4,0){\textcolor{cyan}{$k$}}
\pspolygon[linecolor=cyan,linearc=.4,linewidth=.02](-1.8,-.4)(-1.8,.4)(.2,.4)(.2,-.4)
\end{pspicture}
\caption{$\lbrace s, t \rbrace \notin E$.}
\label{figCompleteComponentWithPriorityLevel-1}
\end{subfigure}
\hfill
\begin{subfigure}[c]{0.45\textwidth}
\centering
\begin{pspicture}(-1,-.3)(1,2.6)
\tiny
\begin{psmatrix}[mnode=circle,colsep=0.5,rowsep=0.5]
{$a$}	& {$c$}\\
{} & {$b$}\\
{$s$}
\psset{arrows=-, shortput=nab,labelsep={0.08}}
\tiny
\ncline{1,1}{2,1}
\ncline{1,1}{2,2}
\ncline{1,1}{1,2}
\ncline{2,1}{1,2}
\ncline{2,1}{2,2}
\ncline{1,2}{2,2}
\ncarc[arcangle=-35]{1,1}{3,1}
\ncline[linestyle=dotted]{3,1}{2,1}
\end{psmatrix}
\normalsize
\uput[0](-2.4,1.2){\textcolor{orange}{$k'$}}
\pspolygon[linecolor=orange,linearc=.4,linewidth=.02](-1.8,.6)(-1.8,2.3)(.2,2.3)(.2,.6)
\uput[0](-2.4,0){\textcolor{cyan}{$k$}}
\pspolygon[linecolor=cyan,linearc=.4,linewidth=.02](-1.8,-.4)(-1.8,.4)(.2,.4)(.2,-.4)
\end{pspicture}
\caption{$\lbrace s, c \rbrace \notin E$ and $\lbrace s, b \rbrace \notin E$.}
\label{figCompleteComponentWithPriorityLevel-2}
\end{subfigure}
\caption{$k<k'$.}
\label{figCompleteComponentWithPriorityLevel}
\end{figure}

\item
Assume that there exist at least two nodes $b$ and $c$ in $C_a$ that are not linked to $s$
as represented in Figure~\ref{figCompleteComponentWithPriorityLevel-2}.
Then $\lbrace s,a,b,c \rbrace$ corresponds to a $3$-pan where $s$ has degree~$1$,
and Example~\ref{Example3-pan-NonValidFor4-path} gives a contradiction where $s$ plays the role of $3$ with a strictly lower priority.
\qedhere
\end{itemize}
\end{proof}

We now investigate how the presence of other layers impact what happens on a layer of priority $k$. We establish that if a graph preserves $(\omega,\Sigma)$-convexity, then there are important restrictions on connected components simultaneously linked to a connected component lying in a higher layer.

\begin{proposition}
\label{propC1andC2LinkedToCOfHigherLevelThenSingletons}
Let $(\omega,\Sigma)$ be a weight system with at least two priority levels and let $G=(N,E)$ be a graph preserving $(\omega,\Sigma)$-convexity.
Let $k, k', k''$ be priority levels with $k \leq k'< k''$.
Let $C_1$
(resp. $C_2$)
be a connected component of $G_k$
(resp. $G_{k'}$)
linked to a connected component $C$ of $G_{k''}$.
Then the following statements are satisfied:
\begin{enumerate}
\item
\label{itemPropC1AndC2LinkedToACommonnodeInC}
There exists a node $i$ in $C$ such that $C_1$ and $C_2$ are both linked to $i$.
\item
\label{itemPropNoLinkBetweenC1AndC2ThenC2Singleton}
If there is no link between $C_1$ and $C_2$ then
\begin{enumerate}
\item
\label{itemPropC2NecessarilySingleton}
$C_2$ is a singleton.
\item
\label{itemPropCIsAStar}
$C$ is a star and $C_2$ is linked to $C$ only at its center $c$.
\item
\label{itemPropC1LinkedOnlyToCenterOfC}
$C_1$ is linked to $C$ only at its center $c$.
\item
\label{itemPropC2CannotBeLinkedToConnectedComponentsOfALowerLayer}
$C_2$ cannot be linked to connected components of a lower layer.
\item
\label{itemPropNoPathLinkingC1ToC2NotPassingThroughi}
There is no path linking $C_1$ to $C_2$ that does not pass through $c$.
\end{enumerate}
\item
\label{itemPropIfk<k'}
If there exists a link between $C_1$ and $C_2$,
then we have $k < k'$ and there exists a node $j_1$
(resp. $j_2$)
in $C_1$
(resp. $C_2$)
such that $\lbrace i, j_1, j_2, i \rbrace$ is a triangle.
Moreover,
any component $C_3$ linked to $C$ belonging to a layer with priority $l \leq k'$ is linked to $j_1$ and $j_2$
and we necessarily have $l \notin \lbrace k,k'\rbrace$.
\end{enumerate}
\end{proposition}

\begin{proof}
Let us consider $C_1$ and $C_2$ linked to $C$.

\textbf{\ref{itemPropC1AndC2LinkedToACommonnodeInC}}
Let us assume that $C_1$ and $C_2$ are not linked to a common node in $C$.
Then, by Proposition~\ref{PlayerConnectedToAHigherLevel},
$C$ is necessarily a star of size $2$.
Let $i_1, i_2$ be the two distinct nodes of $C$ and let $j_1$
(resp. $j_2$)
be a node in $C_1$
(resp. $C_2$)
such that $\lbrace i_1, j_1 \rbrace \in E$
(resp. $\lbrace i_2, j_2 \rbrace \in E$)
as represented in Figure~\ref{figTwoConnectedComponentsLinkedToTwoNodesOfHigherPriority}.
\begin{figure}[!ht]
\centering
\begin{pspicture}(-1,-.3)(1,2)
\tiny
\begin{psmatrix}[mnode=circle,colsep=0.4,rowsep=0.5]
	& {$i_1$} & {$i_2$}\\
{$j_1$} & & & {$j_2$}
\psset{arrows=-, shortput=nab,labelsep={0.08}}
\tiny
\ncline{2,1}{1,2}
\ncline{1,3}{2,4}
\ncline{1,2}{1,3}
\ncline{2,3}{2,4}
\ncline[linestyle=dotted]{2,1}{2,4}
\end{psmatrix}
\normalsize
\uput[0](-3.4,1.2){\textcolor{red}{$k''$}}
\pspolygon[linecolor=red,linearc=.4,linewidth=.02](-2.8,.7)(-2.8,1.5)(-.8,1.5)(-.8,.7)
\uput[0](.6,.2){\textcolor{orange}{$k'$}}
\pspolygon[linecolor=orange,linearc=.4,linewidth=.02](-1.6,-.4)(-1.6,.5)(.6,.5)(.6,-.4)
\uput[0](-4.8,-.2){\textcolor{cyan}{$k$}}
\pspolygon[linecolor=cyan,linearc=.4,linewidth=.02](-4.2,-.5)(-4.2,.4)(-1.8,.4)(-1.8,-.5)
\end{pspicture}
\caption{$k \leq k' < k''$.}
\label{figTwoConnectedComponentsLinkedToTwoNodesOfHigherPriority}
\end{figure}
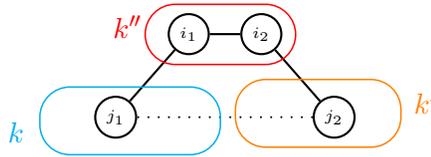
$\lbrace j_1, i_2 \rbrace$ and $\lbrace j_2, i_1\rbrace$ cannot exist otherwise
$C_1$ and $C_2$ would be linked to a common node in $C$.
If $\lbrace j_1, j_2 \rbrace \notin E$,
then $\lbrace j_1, i_1, i_2, j_2 \rbrace$ corresponds to a $4$-path where $i_1$ and $i_2$ have highest priority
and Example~\ref{Example4-path-strict} gives a contradiction to inheritance of $(\omega, \Sigma)$-convexity.
If $\lbrace j_1, j_2 \rbrace \in E$,
then $\lbrace j_1, i_1, i_2, j_2 \rbrace$ corresponds to a $4$-cycle where $i_1$ and $i_2$ have highest priority
and Example~\ref{ExampleNon-CompleteCycle-priorities}
implies $\lbrace j_1, i_2 \rbrace \in E$ and $\lbrace j_2, i_1\rbrace \in E$,
a contradiction.

\textbf{\ref{itemPropNoLinkBetweenC1AndC2ThenC2Singleton}}
We assume that there is no link between $C_1$ and $C_2$.

\textbf{\ref{itemPropC2NecessarilySingleton}}
By contradiction,
let us assume that $C_2$ is not a singleton.
Let $i$ be a node in $C$ and $j_1$
(resp. $j_2$)
be a node in $C_1$
(resp. $C_2$)
such that $\lbrace i,j_1 \rbrace \in E$
(resp. $\lbrace i,j_2 \rbrace \in E$).
As $\vert V(C_2) \vert \geq 2$,
there exists a node $l_2 \not= j_2$ in $C_2$ such that $\lbrace j_2, l_2 \rbrace \in E$.
If $\lbrace i, l_2 \rbrace \notin E$
(resp. $\lbrace i, l_2 \rbrace \in E$),
then there exists a $4$-path
(resp. $3$-pan)
as represented in Figure~\ref{figTwoConnectedComponentsLinkedToAnodeOfHigherPriority-1}
(resp. Figure~\ref{figTwoConnectedComponentsLinkedToAnodeOfHigherPriority-2})
\begin{figure}[!ht]
\begin{subfigure}[c]{0.45\textwidth}
\centering
\begin{pspicture}(-1,-.3)(1,2)
\tiny
\begin{psmatrix}[mnode=circle,colsep=0.4,rowsep=0.5]
	& {$i$}\\
{$j_1$} & & {$j_2$} & {$l_2$}
\psset{arrows=-, shortput=nab,labelsep={0.08}}
\tiny
\ncline{2,1}{1,2}
\ncline{1,2}{2,3}
\ncline{2,3}{2,4}
\end{psmatrix}
\normalsize
\uput[0](-3.4,1.2){\textcolor{red}{$k''$}}
\pspolygon[linecolor=red,linearc=.4,linewidth=.02](-2.8,.6)(-2.8,1.4)(-.8,1.4)(-.8,.6)
\uput[0](.6,.2){\textcolor{orange}{$k'$}}
\pspolygon[linecolor=orange,linearc=.4,linewidth=.02](-1.8,-.4)(-1.8,.5)(.6,.5)(.6,-.4)
\uput[0](-4.8,-.2){\textcolor{cyan}{$k$}}
\pspolygon[linecolor=cyan,linearc=.4,linewidth=.02](-4.2,-.5)(-4.2,.4)(-2.1,.4)(-2.1,-.5)
\end{pspicture}
\caption{$\lbrace i, l_2 \rbrace \notin E$.}
\label{figTwoConnectedComponentsLinkedToAnodeOfHigherPriority-1}
\end{subfigure}
\hfill
\begin{subfigure}[c]{0.45\textwidth}
\centering
\begin{pspicture}(-1,-.3)(1,2)
\tiny
\begin{psmatrix}[mnode=circle,colsep=0.4,rowsep=0.5]
& {$i$}\\
{$j_1$} & & {$j_2$} & {$l_2$}
\psset{arrows=-, shortput=nab,labelsep={0.08}}
\tiny
\ncline{2,1}{1,2}
\ncline{1,2}{2,3}
\ncline{2,3}{2,4}
\ncline{1,2}{2,4}
\end{psmatrix}
\normalsize
\uput[0](-3.4,1.2){\textcolor{red}{$k''$}}
\pspolygon[linecolor=red,linearc=.4,linewidth=.02](-2.8,.6)(-2.8,1.4)(-.8,1.4)(-.8,.6)
\uput[0](.6,.2){\textcolor{orange}{$k'$}}
\pspolygon[linecolor=orange,linearc=.4,linewidth=.02](-1.8,-.4)(-1.8,.5)(.6,.5)(.6,-.4)
\uput[0](-4.8,-.2){\textcolor{cyan}{$k$}}
\pspolygon[linecolor=cyan,linearc=.4,linewidth=.02](-4.2,-.5)(-4.2,.4)(-2.1,.4)(-2.1,-.5)
\end{pspicture}
\caption{$\lbrace i, l_2 \rbrace \in E$.}
\label{figTwoConnectedComponentsLinkedToAnodeOfHigherPriority-2}
\end{subfigure}
\caption{$k \leq k' < k''$.}
\label{figTwoConnectedComponentsLinkedToAnodeOfHigherPriority}
\end{figure}
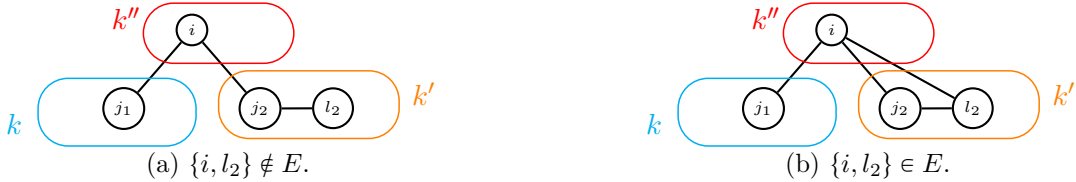
with $p(i) > p(j_2) = p(l_2) \geq p(j_1)$.
Then Example~\ref{Example3-pan-NonValidFor4-path} gives a contradiction to inheritance of $(\omega, \Sigma)$-convexity.
Thus $C_2$ is a singleton.

\textbf{\ref{itemPropCIsAStar}}
Let us prove that $C$ is necessarily a star and that $C_2$ is linked only to the center of $C$.
Let us first consider $\vert V(C) \vert = 2$.
Let us set $V(C) = \lbrace i_1, i_2 \rbrace$.
By Statement~\ref{itemPropC1AndC2LinkedToACommonnodeInC},
we can assume w.l.o.g. that $j_1$ and $j_2$ are both linked to $i_1$.
If $j_1$ is linked
(resp. non linked)
to $i_2$ and $j_2$ non linked
(resp. linked)
to $i_2$
as represented in Figure~\ref{figTwoSingletonsLinkedToAnodeOfHigherPriority-b}
(resp. Figure~\ref{figTwoSingletonsLinkedToAnodeOfHigherPriority-c})
then $\lbrace i_1, j_2, i_2, j_1 \rbrace$ corresponds to a $3$-pan and Example~\ref{Example3-pan-NonValidFor4-path}
gives a contradiction.
If $j_1$ and $j_2$ are linked to $i_2$ as represented in Figure~\ref{figTwoSingletonsLinkedToAnodeOfHigherPriority-a}
\begin{figure}[!ht]
\begin{subfigure}[t]{0.32\textwidth}
\centering
\begin{pspicture}(-1,-.3)(1,1.7)
\tiny
\begin{psmatrix}[mnode=circle,colsep=0.6,rowsep=0.5]
{$i_1$} & {$i_2$}\\
{$j_1$} & [mnode=R]{\framebox{$j_2$}}
\psset{arrows=-, shortput=nab,labelsep={0.08}}
\tiny
\ncline{1,1}{1,2}
\ncline{1,1}{2,1}
\ncline{1,1}{2,2}
\ncline{2,1}{1,2}
\end{psmatrix}
\normalsize
\uput[0](-2.4,1.2){\textcolor{red}{$k''$}}
\pspolygon[linecolor=red,linearc=.4,linewidth=.02](-1.8,.7)(-1.8,1.5)(.2,1.5)(.2,.7)
\uput[0](.7,.2){\textcolor{orange}{$k'$}}
\pspolygon[linecolor=orange,linearc=.4,linewidth=.02](-.8,-.35)(-.8,.45)(.7,.45)(.7,-.35)
\uput[0](-3.2,-.2){\textcolor{cyan}{$k$}}
\pspolygon[linecolor=cyan,linearc=.4,linewidth=.02](-2.6,-.4)(-2.6,.4)(-1,.4)(-1,-.4)
\end{pspicture}
\caption{$\lbrace j_1, i_2 \rbrace \in E$.}
\label{figTwoSingletonsLinkedToAnodeOfHigherPriority-b}
\end{subfigure}
\hfill
\begin{subfigure}[t]{0.32\textwidth}
\centering
\begin{pspicture}(-1,-.3)(1,1.7)
\tiny
\begin{psmatrix}[mnode=circle,colsep=0.6,rowsep=0.5]
{$i_1$} & {$i_2$}\\
{$j_1$} &  [mnode=R]{\framebox{$j_2$}}
\psset{arrows=-, shortput=nab,labelsep={0.08}}
\tiny
\ncline{1,1}{1,2}
\ncline{1,1}{2,1}
\ncline{1,2}{2,2}
\ncline{1,1}{2,2}
\end{psmatrix}
\normalsize
\uput[0](-2.4,1.2){\textcolor{red}{$k''$}}
\pspolygon[linecolor=red,linearc=.4,linewidth=.02](-1.8,.7)(-1.8,1.5)(.2,1.5)(.2,.7)
\uput[0](.7,.2){\textcolor{orange}{$k'$}}
\pspolygon[linecolor=orange,linearc=.4,linewidth=.02](-.8,-.35)(-.8,.45)(.7,.45)(.7,-.35)
\uput[0](-3.2,-.2){\textcolor{cyan}{$k$}}
\pspolygon[linecolor=cyan,linearc=.4,linewidth=.02](-2.6,-.4)(-2.6,.4)(-1,.4)(-1,-.4)
\end{pspicture}
\caption{$\lbrace j_2, i_2 \rbrace \in E$.}
\label{figTwoSingletonsLinkedToAnodeOfHigherPriority-c}
\end{subfigure}
\hfill
\begin{subfigure}[t]{0.32\textwidth}
\centering
\begin{pspicture}(-1,-.3)(1,1.7)
\tiny
\begin{psmatrix}[mnode=circle,colsep=0.6,rowsep=0.5]
{$i_1$} & {$i_2$}\\
{$j_1$} &  [mnode=R]{\framebox{$j_2$}}
\psset{arrows=-, shortput=nab,labelsep={0.08}}
\tiny
\ncline{1,1}{1,2}
\ncline{1,1}{2,1}
\ncline{1,2}{2,2}
\ncline{2,1}{1,2}
\ncline{1,1}{2,2}
\end{psmatrix}
\normalsize
\uput[0](-2.4,1.2){\textcolor{red}{$k''$}}
\pspolygon[linecolor=red,linearc=.4,linewidth=.02](-1.8,.7)(-1.8,1.5)(.2,1.5)(.2,.7)
\uput[0](.7,.2){\textcolor{orange}{$k'$}}
\pspolygon[linecolor=orange,linearc=.4,linewidth=.02](-.8,-.35)(-.8,.45)(.7,.45)(.7,-.35)
\uput[0](-3.2,-.2){\textcolor{cyan}{$k$}}
\pspolygon[linecolor=cyan,linearc=.4,linewidth=.02](-2.6,-.4)(-2.6,.4)(-1,.4)(-1,-.4)
\end{pspicture}
\caption{$\lbrace j_1, i_2 \rbrace \in E$ and $\lbrace j_2, i_2 \rbrace \in E$.}
\label{figTwoSingletonsLinkedToAnodeOfHigherPriority-a}
\end{subfigure}
\caption{$k \leq k' < k''$, $C_2$ singleton, $C_2 = \lbrace j_2 \rbrace$.}
\label{figTwoSingletonsLinkedToAnodeOfHigherPriority}
\end{figure}
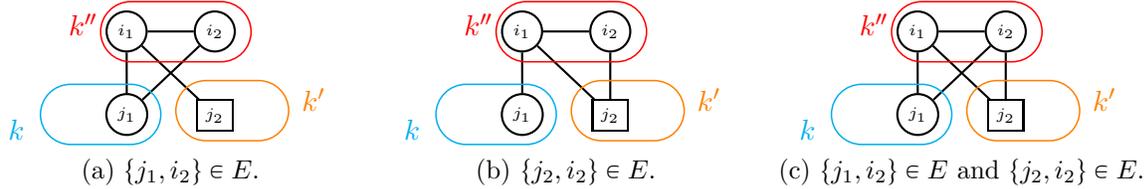
then $\lbrace i_1, j_2, i_2, j_1, i_1 \rbrace$ corresponds to a non-complete cycle and Example~\ref{ExampleNon-CompleteCycle-priorities}
gives a contradiction.
By considering that $C$ is a star with center $c = i_1$, we get the result.
Let us now consider $\vert V(C) \vert \geq 3$.
Let us assume that $C$ is a complete component.
By Proposition~\ref{PlayerConnectedToAHigherLevel},
$j_1$ and $j_2$ are linked to all nodes of $C$.
Let $i_1$ and $i_2$ be two nodes of $C$.
Then $\lbrace i_1, i_2, j_1, j_2 \rbrace$ corresponds to a non-complete cycle
as in Figure~\ref{figTwoSingletonsLinkedToAnodeOfHigherPriority-a}
and Example~\ref{ExampleNon-CompleteCycle-priorities} still gives a contradiction.
Finally,
Proposition~\ref{PlayerConnectedToAHigherLevel} implies that $C$ is a star and that $j_2$ is only linked to its center~$c$.

\textbf{\ref{itemPropC1LinkedOnlyToCenterOfC}}
Let us now prove that $C_1$ is linked to $C$ only at its center $c$.
If $\vert V(C) \vert \geq 3$,
then Proposition~\ref{PlayerConnectedToAHigherLevel} implies the result.
Let us assume $\vert V(C) \vert = 2$
and let $i_2 \not= c$ be the leaf node of~$C$.
If $j_1$ is linked to $i_2$,
then the situation is equivalent to the one represented in Figure~\ref{figTwoSingletonsLinkedToAnodeOfHigherPriority-b}
setting $i_1 = c$ and we get the same contradiction as before.
Let us assume $\vert V(C_1) \vert \geq 2$
and let us consider $j_3 \in V(C_1)$ with $j_3 \not= j_1$.
Let us assume $j_3$ linked to $i_2$
as represented in Figure~\ref{figTwoSingletonsLinkedToAnodeOfHigherPriority-2}.
\begin{figure}[!ht]
\centering
\begin{pspicture}(-1,-.3)(1,1.7)
\tiny
\begin{psmatrix}[mnode=circle,colsep=0.6,rowsep=0.5]
{$i_2$} & {$c$}\\
{$j_3$} & {$j_1$} &  [mnode=R]{\framebox{$j_2$}}
\psset{arrows=-, shortput=nab,labelsep={0.08}}
\tiny
\ncline{1,1}{1,2}
\ncline{1,2}{2,2}
\ncline{1,2}{2,3}
\ncline{2,1}{1,1}
\ncline[linestyle=dotted]{2,1}{2,2}
\end{psmatrix}
\normalsize
\uput[0](-3.7,1.2){\textcolor{red}{$k''$}}
\pspolygon[linecolor=red,linearc=.4,linewidth=.02](-3.1,.7)(-3.1,1.5)(-1.1,1.5)(-1.1,.7)
\uput[0](.7,.2){\textcolor{orange}{$k'$}}
\pspolygon[linecolor=orange,linearc=.4,linewidth=.02](-.8,-.35)(-.8,.45)(.7,.45)(.7,-.35)
\uput[0](-3.6,-.2){\textcolor{cyan}{$k$}}
\pspolygon[linecolor=cyan,linearc=.4,linewidth=.02](-3.1,-.4)(-3.1,.4)(-1,.4)(-1,-.4)
\end{pspicture}
\caption{$k \leq k' < k''$, $C_2$ singleton, $C_2 = \lbrace j_2 \rbrace$.}
\label{figTwoSingletonsLinkedToAnodeOfHigherPriority-2}
\end{figure}
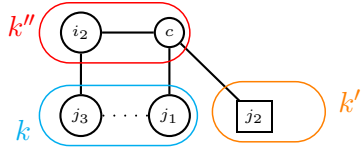
As $C_1$ is connected,
there exists a path $P$ linking $j_1$ to $j_3$.
Then $P \cup \lbrace j_3, i_2, c, j_1 \rbrace$ corresponds to a cycle.
As $p(c) = k''$, Example~\ref{ExampleNon-CompleteCycle-priorities} implies $\lbrace j_1, i_2 \rbrace \in E$, a contradiction.

\textbf{\ref{itemPropC2CannotBeLinkedToConnectedComponentsOfALowerLayer}}
Let us now prove that $C_2$ cannot be linked to components in lower layers.
By contradiction,
let $C_3$ be a connected component in a layer with priority lower than $k'$ linked to $C_2$.
Let $j_3$ be a node in $C_3$ such that $\lbrace j_2, j_3 \rbrace \in E$.
Let us first assume $\lbrace c, j_3 \rbrace \notin E$.
If $\lbrace j_1, j_3 \rbrace \in E$
(resp. $\lbrace j_1, j_3 \rbrace \notin E$),
then $\lbrace j_1, c, j_2, j_3 \rbrace$ corresponds to a non-complete cycle
(resp. a $4$-path)
as represented in Figure~\ref{figThreeConnectedComponentsLinkedToAnodeOfHigherPriority-1}
(resp. Figure~\ref{figThreeConnectedComponentsLinkedToAnodeOfHigherPriority-2})
\begin{figure}[!ht]
\begin{subfigure}[c]{0.45\textwidth}
\centering
\begin{pspicture}(-1,-.3)(1,2.6)
\tiny
\begin{psmatrix}[mnode=circle,colsep=0.4,rowsep=0.5]
	& {$c$}\\
{$j_1$}   & &  [mnode=R]{\framebox{$j_2$}}\\
 & & {$j_3$}
\psset{arrows=-, shortput=nab,labelsep={0.08}}
\tiny
\ncline{2,1}{1,2}
\ncline{1,2}{2,3}
\ncline{2,3}{3,3}
\ncline{2,1}{3,3}
\end{psmatrix}
\normalsize
\uput[0](-3.2,2.2){\textcolor{red}{$k''$}}
\pspolygon[linecolor=red,linearc=.4,linewidth=.02](-2.6,1.7)(-2.6,2.5)(.2,2.5)(.2,1.7)
\uput[0](.2,1.4){\textcolor{orange}{$k'$}}
\pspolygon[linecolor=orange,linearc=.4,linewidth=.02](-1,.7)(-1,1.5)(.2,1.5)(.2,.7)
\uput[0](-3.2,1){\textcolor{cyan}{$k$}}
\pspolygon[linecolor=cyan,linearc=.4,linewidth=.02](-2.6,.65)(-2.6,1.45)(-1.4,1.45)(-1.4,.65)
\uput[0](.2,0){\textcolor{green}{$l$}}
\pspolygon[linecolor=green,linearc=.4,linewidth=.02](-1,-.4)(-1,.4)(.2,.4)(.2,-.4)
\end{pspicture}
\caption{$\lbrace j_1, j_3 \rbrace \in E$.}
\label{figThreeConnectedComponentsLinkedToAnodeOfHigherPriority-1}
\end{subfigure}
\hfill
\begin{subfigure}[c]{0.45\textwidth}
\centering
\begin{pspicture}(-1,-.3)(1,2.6)
\tiny
\begin{psmatrix}[mnode=circle,colsep=0.4,rowsep=0.5]
	& {$c$}\\
{$j_1$} & &  [mnode=R]{\framebox{$j_2$}}\\
& & {$j_3$}
\psset{arrows=-, shortput=nab,labelsep={0.08}}
\tiny
\ncline{2,1}{1,2}
\ncline{1,2}{2,3}
\ncline{2,3}{3,3}
\end{psmatrix}
\normalsize
\uput[0](-3.2,2.2){\textcolor{red}{$k''$}}
\pspolygon[linecolor=red,linearc=.4,linewidth=.02](-2.6,1.7)(-2.6,2.5)(.2,2.5)(.2,1.7)
\uput[0](.2,1.4){\textcolor{orange}{$k'$}}
\pspolygon[linecolor=orange,linearc=.4,linewidth=.02](-1,.7)(-1,1.5)(.2,1.5)(.2,.7)
\uput[0](-3.2,1){\textcolor{cyan}{$k$}}
\pspolygon[linecolor=cyan,linearc=.4,linewidth=.02](-2.6,.65)(-2.6,1.45)(-1.4,1.45)(-1.4,.65)
\uput[0](.2,0){\textcolor{green}{$l$}}
\pspolygon[linecolor=green,linearc=.4,linewidth=.02](-1,-.4)(-1,.4)(.2,.4)(.2,-.4)
\end{pspicture}
\caption{$\lbrace j_1, j_3 \rbrace \notin E$.}
\label{figThreeConnectedComponentsLinkedToAnodeOfHigherPriority-2}
\end{subfigure}
\caption{$k \leq k' < k''$, $l<k$, and $\lbrace c, j_3 \rbrace \notin E$.}
\label{figThreeConnectedComponentsLinkedToAnodeOfHigherPriority}
\end{figure}
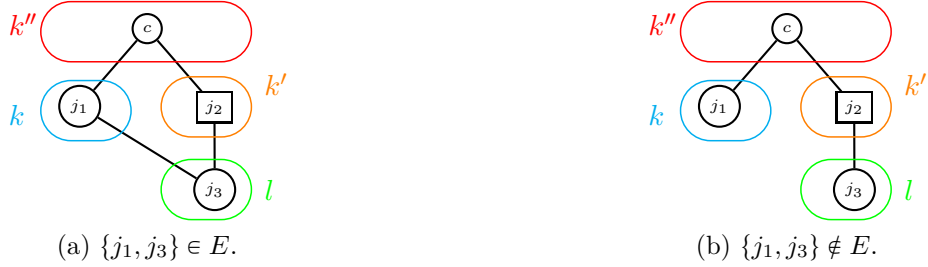
where $c$ has highest priority and Example~\ref{ExampleNon-CompleteCycle-priorities}
(resp. Example~\ref{Example3-pan-NonValidFor4-path})
gives a contradiction.
Let us now assume $\lbrace c, j_3 \rbrace \in E$.
Then we have either a non-complete cycle or a $3$-pan and we get contradictions as in the previous case.

\textbf{\ref{itemPropNoPathLinkingC1ToC2NotPassingThroughi}}
By contradiction let us assume the existence of a path linking $C_1$ to $C_2$
and that does not pass through $c$.
Then we can select a shortest path $P$ linking $j_1$ to $j_2$ and not containing~$c$.
Then $T = P \cup \lbrace j_1, c, j_2 \rbrace$ corresponds to a cycle.
If $p(c) = p(T)$,
then Example~\ref{ExampleNon-CompleteCycle-priorities} implies $\lbrace j_1, j_2 \rbrace \in E$,
a contradiction.
Otherwise,
let $i'$ be a node of $P$ satisfying $p(i')= p(T)$.
Let us note that we necessarily have $i' \notin \lbrace j_1, j_2 \rbrace $
as $p(j_1) = k \leq k' = p(j_2) < k'' = p(c)$.
Then Example~\ref{ExampleNon-CompleteCycle-priorities} implies the existence of an edge linking the neighbors of $i'$ in the path $P$.
This contradicts the minimality of $P$.

\textbf{\ref{itemPropIfk<k'}}
Let us now assume the existence of a link between $C_1$ and $C_2$.
We necessarily have $k < k'$.
Let us consider $i \in C$ and $j_1 \in C_1$
(resp. $j_2 \in C_2$) with $ \lbrace i, j_1 \rbrace \in E$
(resp. $\lbrace i, j_2 \rbrace \in E$).
Let us assume $\lbrace j_1,j_2 \rbrace \notin E$.
By the existence of a link between $C_1$ and $C_2$,
there exists a path $P$ linking $j_1$ to $j_2$
and passing only through $C_1$ and $C_2$
as represented in Figure~\ref{figTwoConnectedComponentsLinkedToAnodeOfHigherPriority-3-1}.
\begin{figure}[!ht]
\begin{subfigure}[c]{0.3\textwidth}
\centering
\begin{pspicture}(-1,-.3)(1,2.6)
\tiny
\begin{psmatrix}[mnode=circle,colsep=0.4,rowsep=0.5]
	& {$i$}\\
  & & {$j_2$}\\
{$j_1$} & {} & {}
\psset{arrows=-, shortput=nab,labelsep={0.08}}
\tiny
\ncline{3,1}{1,2}
\ncline{1,2}{2,3}
\ncline{2,3}{3,3}
\ncline{3,1}{3,2}
\ncline[linestyle=dashed]{3,2}{3,3}
\end{psmatrix}
\normalsize
\uput[0](-3.2,2.2){\textcolor{red}{$k''$}}
\pspolygon[linecolor=red,linearc=.4,linewidth=.02](-2.6,1.7)(-2.6,2.5)(.2,2.5)(.2,1.7)
\uput[0](-3.2,1.2){\textcolor{cyan}{$k'$}}
\pspolygon[linecolor=cyan,linearc=.4,linewidth=.02](-2.6,.7)(-2.6,1.5)(.2,1.5)(.2,.7)
\uput[0](-3.2,0){\textcolor{orange}{$k$}}
\pspolygon[linecolor=orange,linearc=.4,linewidth=.02](-2.6,-.4)(-2.6,.4)(.2,.4)(.2,-.4)
\end{pspicture}
\caption{$\lbrace j_1, j_2 \rbrace \notin E$.}
\label{figTwoConnectedComponentsLinkedToAnodeOfHigherPriority-3-1}
\end{subfigure}
\hfill
\begin{subfigure}[c]{0.3\textwidth}
\centering
\begin{pspicture}(-1,-.3)(1,2.6)
\tiny
\begin{psmatrix}[mnode=circle,colsep=0.4,rowsep=0.5]
	& {$i$}\\
  & & {$j_2$}\\
{$j_1$} &  & & {$j_3$}
\psset{arrows=-, shortput=nab,labelsep={0.08}}
\tiny
\ncline{3,1}{1,2}
\ncline{1,2}{2,3}
\ncline{3,1}{2,3}
\ncarc[arcangle=55]{1,2}{3,4}
\end{psmatrix}
\normalsize
\uput[0](-4.1,2.2){\textcolor{red}{$k''$}}
\pspolygon[linecolor=red,linearc=.4,linewidth=.02](-3.5,1.7)(-3.5,2.5)(-.7,2.5)(-.7,1.7)
\uput[0](-4.1,1.2){\textcolor{cyan}{$k'$}}
\pspolygon[linecolor=cyan,linearc=.4,linewidth=.02](-3.5,.7)(-3.5,1.5)(-.7,1.5)(-.7,.7)
\uput[0](-4.1,0){\textcolor{orange}{$k$}}
\pspolygon[linecolor=orange,linearc=.4,linewidth=.02](-3.5,-.4)(-3.5,.4)(-1.8,.4)(-1.8,-.4)
\uput[0](-1.6,0){\textcolor{green}{$l$}}
\pspolygon[linecolor=green,linearc=.4,linewidth=.02](-1.1,-.4)(-1.1,.4)(.5,.4)(.5,-.4)
\end{pspicture}
\caption{$\lbrace j_1, j_2 \rbrace \in E$, $l \leq k'$.}
\label{figTwoConnectedComponentsLinkedToAnodeOfHigherPriority-3-2}
\end{subfigure}
\hfill
\begin{subfigure}[c]{0.3\textwidth}
\centering
\begin{pspicture}(-1,-.3)(1,2.6)
\tiny
\begin{psmatrix}[mnode=circle,colsep=0.4,rowsep=0.5]
	& {$i$}\\
  & & {$j_2$}\\
{$j_1$} & & & {$j_3$}
\psset{arrows=-, shortput=nab,labelsep={0.08}}
\tiny
\ncline{3,1}{1,2}
\ncline{1,2}{2,3}
\ncline{3,1}{2,3}
\ncline{2,3}{3,4}
\ncarc[arcangle=55]{1,2}{3,4}
\end{psmatrix}
\normalsize
\uput[0](-4.1,2.2){\textcolor{red}{$k''$}}
\pspolygon[linecolor=red,linearc=.4,linewidth=.02](-3.5,1.7)(-3.5,2.5)(-.7,2.5)(-.7,1.7)
\uput[0](-4.1,1.2){\textcolor{cyan}{$k'$}}
\pspolygon[linecolor=cyan,linearc=.4,linewidth=.02](-3.5,.7)(-3.5,1.5)(-.7,1.5)(-.7,.7)
\uput[0](-4.1,0){\textcolor{orange}{$k$}}
\pspolygon[linecolor=orange,linearc=.4,linewidth=.02](-3.5,-.4)(-3.5,.4)(-1.8,.4)(-1.8,-.4)
\uput[0](-1.6,0){\textcolor{green}{$l$}}
\pspolygon[linecolor=green,linearc=.4,linewidth=.02](-1.1,-.4)(-1.1,.4)(.5,.4)(.5,-.4)
\end{pspicture}
\caption{$\lbrace j_1, j_2 \rbrace \in E$, $l \leq k'$.}
\label{figTwoConnectedComponentsLinkedToAnodeOfHigherPriority-3-3}
\end{subfigure}
\caption{$ k < k' < k''$.}
\label{figTwoConnectedComponentsLinkedToAnodeOfHigherPriority-3}
\end{figure}
Example~\ref{ExampleNon-CompleteCycle-priorities} applied with the cycle formed by $\lbrace i, j_1 \rbrace \cup P \cup \lbrace j_2, i \rbrace$
implies $\lbrace j_1,j_2 \rbrace \in E$, a contradiction.
Therefore,
we necessarily have $\lbrace j_1,j_2 \rbrace \in E$.

Let us assume $C$ linked to a component $C_3$ different from $C_1$ and $C_2$.
Let us assume that $C_3$ belongs to a layer with priority $l \leq k'$.
Let $j_3$  be a node of $C_3$ linked to $i$ as represented in Figure~\ref{figTwoConnectedComponentsLinkedToAnodeOfHigherPriority-3-2}
with $l=k$.
If $\lbrace j_2,j_3 \rbrace \notin E$,
then Example~\ref{Example3-pan-NonValidFor4-path} gives a contradiction.
Thus we necessarily have $\lbrace j_2,j_3 \rbrace \in E$ as represented in Figure~\ref{figTwoConnectedComponentsLinkedToAnodeOfHigherPriority-3-3}.
Then $\lbrace i, j_1,j_2,j_3,i  \rbrace$ is a cycle where $i$ has highest priority
and Example~\ref{ExampleNon-CompleteCycle-priorities} implies $\lbrace j_1, j_3 \rbrace \in E$.
\end{proof}

As a direct consequence of Proposition~\ref{propC1andC2LinkedToCOfHigherLevelThenSingletons},
we get the following result.

\begin{lemma}\label{Samelevel}
Let $(\omega,\Sigma)$ be a weight system with at least two priority levels and let $G=(N,E)$ be a graph preserving $(\omega,\Sigma)$-convexity.
Let $k, k'$ be priority levels with $k < k'$.
Let $C_1$ and $C_2$ be connected components of $G_k$
linked to a connected component $C$ of $G_{k'}$.
Then $C$ is a star and $C_1$ and $C_2$ are necessarily singletons linked only to the center of the star.
Moreover,
$C_1$ and $C_2$ cannot be linked to connected components of a lower layer.
\end{lemma}

\subsection{Necessary and sufficient conditions for Priority-decreasing trees.}
\label{Subsection-Necessary-Sufficient-Conditions-Priority-Decreasing-Trees}
We now investigate a special type of graphs.

\begin{definition}
Let $(\omega,\Sigma)$ be a weight system.
We say that a graph $G = (N,E)$ is an \emph{$(\omega, \Sigma)$-priority-decreasing tree} if 
\begin{itemize}
\item
$G$ is cycle-free,
\item
There exists a root $r$ such that the priorities are decreasing for the tree-order generated from $r$.
\end{itemize}
\end{definition}
We give in Figure~\ref{fig(omega,Sigma)-Priority-Decreasing-Trees} examples of $(\omega, \Sigma)$-priority-decreasing trees.
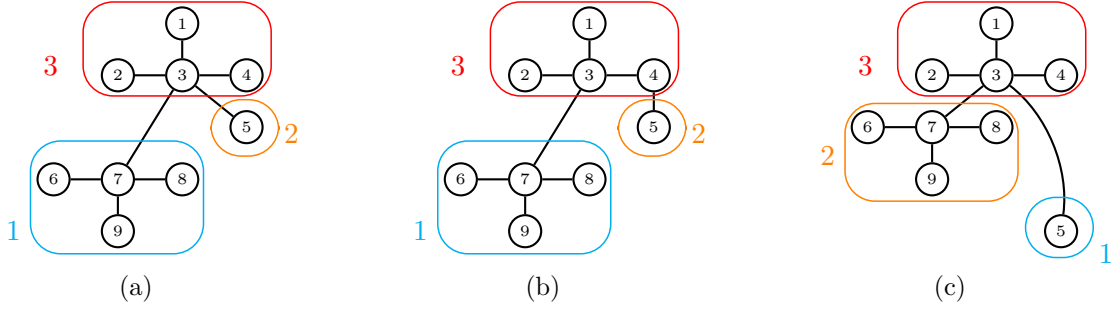
\begin{figure}[!ht]
\begin{subfigure}[t]{0.3\textwidth}
\centering
\begin{pspicture}(-1,-.3)(1,2.8)
\tiny
\begin{psmatrix}[mnode=circle,colsep=0.4,rowsep=0.2]
& & {$1$}\\
& {$2$} & {$3$} & {$4$}\\
&  & &  {$5$}\\
 {$6$} & {$7$} & {$8$}\\
  & {$9$}
\psset{arrows=-, shortput=nab,labelsep={0.08}}
\tiny
\ncline{1,3}{2,3}
\ncline{2,2}{2,3}
\ncline{2,3}{2,4}
\ncline{2,3}{3,4}
\ncline{2,3}{4,2}
\ncline{4,1}{4,2}
\ncline{4,2}{4,3}
\ncline{4,2}{5,2}
\end{psmatrix}
\normalsize
\uput[0](-3.1,2.2){\textcolor{red}{$3$}}
\pspolygon[linecolor=red,linearc=.4,linewidth=.02](-2.4,1.8)(-2.4,3.05)(.1,3.05)(.1,1.8)
\uput[0](.1,1.3){\textcolor{orange}{$2$}}
\pspolygon[linecolor=orange,linearc=.4,linewidth=.02](-.7,1)(-.7,1.75)(.2,1.75)(.2,1)
\uput[0](-3.6,0){\textcolor{cyan}{$1$}}
\pspolygon[linecolor=cyan,linearc=.4,linewidth=.02](-3.1,-.3)(-3.1,1.2)(-.8,1.2)(-.8,-.3)
\end{pspicture}
\caption{}
\label{fig(omega,Sigma)-Priority-Decreasing-Trees-1}
\end{subfigure}
\hfill
\begin{subfigure}[t]{0.3\textwidth}
\centering
\begin{pspicture}(-1,-.3)(1,2.8)
\tiny
\begin{psmatrix}[mnode=circle,colsep=0.4,rowsep=0.2]
& & {$1$}\\
& {$2$} & {$3$} & {$4$}\\
&  & &  {$5$}\\
 {$6$} & {$7$} & {$8$}\\
  & {$9$}
\psset{arrows=-, shortput=nab,labelsep={0.08}}
\tiny
\ncline{1,3}{2,3}
\ncline{2,2}{2,3}
\ncline{2,3}{2,4}
\ncline{2,4}{3,4}
\ncline{2,3}{4,2}
\ncline{4,1}{4,2}
\ncline{4,2}{4,3}
\ncline{4,2}{5,2}
\end{psmatrix}
\normalsize
\uput[0](-3.1,2.2){\textcolor{red}{$3$}}
\pspolygon[linecolor=red,linearc=.4,linewidth=.02](-2.4,1.8)(-2.4,3.05)(.1,3.05)(.1,1.8)
\uput[0](.1,1.3){\textcolor{orange}{$2$}}
\pspolygon[linecolor=orange,linearc=.4,linewidth=.02](-.7,1)(-.7,1.75)(.2,1.75)(.2,1)
\uput[0](-3.6,0){\textcolor{cyan}{$1$}}
\pspolygon[linecolor=cyan,linearc=.4,linewidth=.02](-3.1,-.3)(-3.1,1.2)(-.8,1.2)(-.8,-.3)
\end{pspicture}
\caption{}
\label{fig(omega,Sigma)-Priority-Decreasing-Trees-2}
\end{subfigure}
\hfill
\begin{subfigure}[t]{0.3\textwidth}
\centering
\begin{pspicture}(-1,-.3)(1,2.8)
\tiny
\begin{psmatrix}[mnode=circle,colsep=0.4,rowsep=0.2]
& & {$1$}\\
& {$2$} & {$3$} & {$4$}\\
 {$6$} & {$7$} & {$8$}\\
 & {$9$}\\
&  & &  {$5$}
\psset{arrows=-, shortput=nab,labelsep={0.08}}
\tiny
\ncline{1,3}{2,3}
\ncline{2,2}{2,3}
\ncline{2,3}{2,4}
\ncarc[arcangle=30]{2,3}{5,4}
\ncline{2,3}{3,2}
\ncline{3,1}{3,2}
\ncline{3,2}{3,3}
\ncline{3,2}{4,2}
\end{psmatrix}
\normalsize
\uput[0](-3.1,2.2){\textcolor{red}{$3$}}
\pspolygon[linecolor=red,linearc=.4,linewidth=.02](-2.4,1.8)(-2.4,3.05)(.1,3.05)(.1,1.8)
\uput[0](.1,-.3){\textcolor{cyan}{$1$}}
\pspolygon[linecolor=cyan,linearc=.4,linewidth=.02](-.7,-.35)(-.7,.45)(.2,.45)(.2,-.35)
\uput[0](-3.6,1){\textcolor{orange}{$2$}}
\pspolygon[linecolor=orange,linearc=.4,linewidth=.02](-3.1,.4)(-3.1,1.7)(-.8,1.7)(-.8,.4)
\end{pspicture}
\caption{}
\label{fig(omega,Sigma)-Priority-Decreasing-Trees-3}
\end{subfigure}
\caption{$(\omega, \Sigma)$-priority-decreasing trees with root $3$.}
\label{fig(omega,Sigma)-Priority-Decreasing-Trees}
\end{figure}
The following theorem provides a characterization of priority-decreasing trees preserving $(\omega,\Sigma)$-convexity.

\begin{theorem}
\label{characterization_hierarchy}
Let $(\omega,\Sigma)$ be a weight system and $G=(N,E)$ be an $(\omega,\Sigma)$-priority-decreasing tree.
We denote the priorities by $1$ to $\overline{k}$ from low to high priority.
The following statements are equivalent:
\begin{enumerate}
\item
$G$ preserves $(\omega,\Sigma)$-convexity.
\item
\label{ExistenceOfStrictlyDecreasingSubsequence}
There exist $m\geq 1 $ and a strictly decreasing subsequence of priorities $(k_l)_{1\leq l \leq m}$
such that:
\begin{enumerate}
\item
$k_1=\overline{k}$,
\item
For all $l$ with $1 \leq l \leq m$,
the graph restricted to the set of nodes with priority between $k_l$ and $k_{l+1}+1$ is a star denoted $(N_l,E_l)$
(setting $k_{m+1}=0$),
\item
The center of $(N_l,E_l)$ denoted $c_l$ has priority $k_l$,
\item
If $m \geq 2$, then for all $l$ with $1 \leq l \leq m-1$
there is exactly one edge $e_l$ linking $(N_l,E_l)$ to $(N_{l+1},E_{l+1})$,
and $e_l$ connects a node of priority $k_{l+1}$ in $N_{l+1}$ to~$c_l$.
\end{enumerate}
\end{enumerate}
\end{theorem}

Let us observe that the graph given in Figure~\ref{fig(omega,Sigma)-Priority-Decreasing-Trees-1}
preserves $(\omega,\Sigma)$-convexity as it satisfies Condition~\ref{ExistenceOfStrictlyDecreasingSubsequence} with $m =2$.
Then $(N_1,E_1)$
(resp. $(N_2,E_2)$)
corresponds to the star induced by $\lbrace 1,2,3,4,5 \rbrace$
(resp. $\lbrace 6,7,8,9 \rbrace$).
It can be easily seen that the graphs given in Figures~\ref{fig(omega,Sigma)-Priority-Decreasing-Trees-2}
and \ref{fig(omega,Sigma)-Priority-Decreasing-Trees-3}
do not satisfy Condition~\ref{ExistenceOfStrictlyDecreasingSubsequence}
and therefore dot not preserve $(\omega,\Sigma)$-convexity. 

\begin{proof}
We start by establishing that (\ref{ExistenceOfStrictlyDecreasingSubsequence}) is a necessary condition to preserve $(\omega,\Sigma)$-convexity.
In a second step, we will prove that it is sufficient.

\noindent
\underline{Necessary condition:}
We first define the subsequence of priorities by induction.
\begin{enumerate}
\item
Let $k_1=\overline{k}$ be the maximal priority in $G=(N,E)$.
\item
Let $\tilde{N}_1$ be the set of nodes with priority $k_1$ and $\tilde{N}'_1$ the set of neighbors of all elements in $\tilde{N}_1$.
If $\tilde{N}'_1 \not= \emptyset$,
then define $k_2$ to be the minimum priority in $\tilde{N}'_1$.
\item
For every $l$ with $l \geq 2$, let $\tilde{N}_l$ be the set of nodes with priority $k_l$
and $\tilde{N}'_l$ the set of neighbors of all elements in $\tilde{N}_l$ with priority less than $k_l$.
If $\tilde{N}'_l\not= \emptyset$,
then define $k_{l+1}$ to be the minimum priority in $\tilde{N}'_l$.
\end{enumerate}

Let us now show that $G$ has to satisfy the precise structure described in Condition~\ref{ExistenceOfStrictlyDecreasingSubsequence}.
Set $m=1$.
We first consider the nodes at level $k_1$.
Under the assumption of decreasing order, they form a connected component.
For each node $n$ in $\tilde{N}'_1$, we define its sublevel component $C(n)$
as the connected component containing $n$
in the restricted graph $G_{p(n)}$ (containing only nodes having the same priority as $n$).
By Statements~\ref{itemPropC1AndC2LinkedToACommonnodeInC} and \ref{itemPropNoLinkBetweenC1AndC2ThenC2Singleton}
in Proposition~\ref{propC1andC2LinkedToCOfHigherLevelThenSingletons}
and as $G$ is cycle-free,
$\tilde{N}_1$ forms a star $S_1$ and each sublevel component $C(n)$ is linked to $S_1$ only at its center $c_1$
by the unique edge $\{ n, c_1 \}$.
Moreover, each $C(n)$ with priority between $k_2+1$ and $k_1$ is a leaf node in $G$.
We distinguish depending on the number of sublevel components having priority $k_2$:
\begin{itemize}
\item
If there are at least $2$ such components, then they are necessarily leaf nodes by Lemma~\ref{Samelevel}.
Therefore $G$ is a star network as represented in Figure~\ref{figTwoCompsLinkedToAHigherOne-Cycle-free-1-2-a}
and it ends the proof as Condition~\ref{ExistenceOfStrictlyDecreasingSubsequence} is satisfied taking $m =1$.
\begin{figure}[!ht]
\begin{subfigure}[t]{0.45\textwidth}
\centering
\begin{pspicture}(-1,-.4)(1,2.8)
\tiny
\begin{psmatrix}[mnode=circle,colsep=0.5,rowsep=0.4]
 & {}\\
{} & {$c_1$} & {}\\
&  & & [mnode=R]{\psframebox{$n$}}\\
[mnode=R]{\psframebox{$n_1$}}  & & [mnode=R]{\psframebox{$n_2$}}
\psset{arrows=-, shortput=nab,labelsep={0.08}}
\tiny
\ncline{1,2}{2,2}
\ncline{2,1}{2,2}
\ncline{2,2}{2,3}
\ncline{4,1}{2,2}
\ncline{4,3}{2,2}
\ncline{3,4}{2,2}
\end{psmatrix}
\normalsize
\uput[0](-4.2,1.8){\textcolor{red}{$k_1$}}
\pspolygon[linecolor=red,linearc=.4,linewidth=.02](-3.5,1.3)(-3.5,2.7)(-1,2.7)(-1,1.3)
\uput[0](-.8,-.2){\textcolor{cyan}{$k_2$}}
\pspolygon[linecolor=cyan,linearc=.4,linewidth=.02](-1.7,-.35)(-1.7,.45)(-.8,.45)(-.8,-.35)
\uput[0](-4.4,-.2){\textcolor{cyan}{$k_2$}}
\pspolygon[linecolor=cyan,linearc=.4,linewidth=.02](-3.8,-.4)(-3.8,.4)(-2.9,.4)(-2.9,-.4)
\pspolygon[linecolor=orange,linearc=.4,linewidth=.02](-.7,.4)(-.7,1.15)(.2,1.15)(.2,.4)
\end{pspicture}
\caption{If $p(n_1) = p(n_2) = k_2$ then $C(n_1)$ and $C(n_2)$ are singletons.}
\label{figTwoCompsLinkedToAHigherOne-Cycle-free-1-2-a}
\end{subfigure}
\hfill
\begin{subfigure}[t]{0.45\textwidth}
\centering
\begin{pspicture}(-1,-.3)(1,2.8)
\tiny
\begin{psmatrix}[mnode=circle,colsep=0.4,rowsep=0.2]
& & {}\\
& {} & {$c_1$} & {}\\
&  & &  [mnode=R]{\psframebox{$n$}}\\
 {} & {$n_1$} & {}\\
  & {}
\psset{arrows=-, shortput=nab,labelsep={0.08}}
\tiny
\ncline{1,3}{2,3}
\ncline{2,2}{2,3}
\ncline{2,3}{2,4}
\ncline{2,3}{3,4}
\ncline{2,3}{4,2}
\ncline{4,1}{4,2}
\ncline{4,2}{4,3}
\ncline{4,2}{5,2}
\end{psmatrix}
\normalsize
\uput[0](-3.1,2.2){\textcolor{red}{$k_1$}}
\pspolygon[linecolor=red,linearc=.4,linewidth=.02](-2.4,1.8)(-2.4,3.05)(.1,3.05)(.1,1.8)
\pspolygon[linecolor=orange,linearc=.4,linewidth=.02](-.7,1)(-.7,1.75)(.2,1.75)(.2,1)
\uput[0](-3.8,0){\textcolor{cyan}{$k_2$}}
\pspolygon[linecolor=cyan,linearc=.4,linewidth=.02](-3.1,-.3)(-3.1,1.2)(-.8,1.2)(-.8,-.3)
\end{pspicture}
\caption{$C(n_1)$ unique component with priority $k_2$.}
\label{figTwoCompsLinkedToAHigherOne-Cycle-free-1-2-b}
\end{subfigure}
\caption{$k_2 < p(n) < k_1$.}
\label{figTwoCompsLinkedToAHigherOne-Cycle-free-1-2}
\end{figure}
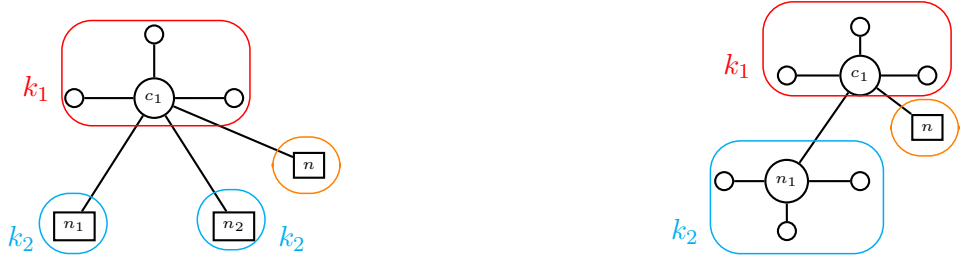
\item
Otherwise,
we increase $m$ by $1$.
Under the assumption of decreasing order,
we know that the priority of nodes in $N \setminus \lbrace c_1 \rbrace$
connected to the unique sublevel component of priority $k_2$ have a lower priority than $k_2$.
It follows that we had all nodes with higher priority.
The graph restricted to the set of nodes with priority between $k_1$ and $k_{2}+1$ is a star denoted $(N_1,E_1)$
with center $c_1$.
By construction, $c_1$ has priority $k_1$ and there exists a unique link between $c_1$ and a node with priority $k_2$
as represented in Figure~\ref{figTwoCompsLinkedToAHigherOne-Cycle-free-1-2-b}.
\end{itemize}
In the second case,
one can repeat the argument in the restricted graph of nodes with priority lower than $k_2$.\\

\noindent
\underline{Sufficient condition:}
We now prove that such a graph preserves the $(\omega,\Sigma)$-convexity.
The proof follows the lines of \cite{Slikker1998}.
We need to check that if the original game is $(\omega,\Sigma)$-convex
then Equation~(\ref{eqDefinition(omega,Sigma)-convexity}) remains true for the communication game
for any pair of subsets $S$ and $T$ such that $S\subset T$.\\ 

If $p(S)<p(T)$, then for every $i\in S$, $\overline{\omega}^T_i=0$ and Equation~(\ref{eqDefinition(omega,Sigma)-convexity}) is trivially true.
We can now turn to the case where $p(S)=p(T)=k$.
Equation~(\ref{eqDefinition(omega,Sigma)-convexity}) can be rewritten as
\begin{equation}
\label{eqDefinition(omega,Sigma)-convexity-sufficiency}
\sum_{i\in S,\ p(i)=k}\overline{\omega}^T_i \left(v^G(S)-v^G(S\setminus \{i\}) \right) \leq \sum_{i\in S,\ p(i)=k} \overline{\omega}_i^T \left(v^G(T)-v^G(T\setminus \{i\}) \right).
\end{equation}
Let $l$ be the highest index in $\lbrace 1, \ldots, m \rbrace$ such that $T$ intersects the star $(N_l,E_l)$
with center~$c_l$.
We have $k_{l+1} < k \leq k_l$.
We distinguish several cases depending if $c_l$ is in $S$ and/or $T$.
\begin{itemize}
\item
If $c_l \notin S$ and $c_l \notin T$
as represented in Figure~\ref{figTwoStarsClCl+1-a},
\begin{figure}[!ht]
\begin{subfigure}[t]{0.45\textwidth}
\centering
\begin{pspicture}(-2,-.2)(1,3.8)
\tiny
\begin{psmatrix}[mnode=circle,colsep=0.35,rowsep=0.25]
& & & {}\\
& {} & {$c_l$} & {}\\
&  & [mnode=R]{\psframebox{}} &  [mnode=R]{\psframebox{}}\\
 {} & {$c_{l+1}$} & {}\\
 [mnode=R]{\psframebox{}} & {} &  [mnode=R]{\psframebox{}}\\
	& &  {}
\psset{arrows=-, shortput=nab,labelsep={0.08}}
\tiny
\ncline[linestyle=dotted]{1,4}{2,3}
\ncline[linestyle=dotted]{4,2}{6,3}
\ncline{2,2}{2,3}
\ncline{2,3}{3,3}
\ncline{2,3}{2,4}
\ncline{2,3}{3,4}
\ncline{2,3}{4,2}
\ncline{4,1}{4,2}
\ncline{4,2}{4,3}
\ncline{4,2}{5,2}
\ncline{4,2}{5,1}
\ncline{4,2}{5,3}
\end{psmatrix}
\normalsize
\uput[0](-3.9,2.4){\textcolor{red}{$(N_l,E_l)$}}
\pspolygon[linecolor=red,linearc=.4,linewidth=.02](-2.3,2)(-2.3,3.3)(.25,3.3)(.25,2)
\uput[0](-5.5,1){\textcolor{cyan}{$(N_{l+1},E_{l+1})$}}
\pspolygon[linecolor=cyan,linearc=.4,linewidth=.02](-3.1,.3)(-3.1,1.85)(-.65,1.85)(-.65,.3)
\uput[0](.8,1){$T$}
\psline[linewidth=.4pt,linearc=0.5]
{-}(.9,0)(.7,1.6)(.1,2.6)(-1,2.6)(-2.6,1.7)(-2.1,0)
\uput[0](-.1,1){$S$}
\psline[linewidth=.4pt,linearc=0.5]
{-}(0,0)(-.2,1.5)(-.45,2.4)(-1,2.5)(-1.4,2.2)(-1.4,0)
\end{pspicture}
\caption{$c_l \notin S$ and $c_l \notin T$.}
\label{figTwoStarsClCl+1-a}
\end{subfigure}
\hfill
\begin{subfigure}[t]{0.45\textwidth}
\centering
\begin{pspicture}(-2,-.2)(1,3.8)
\tiny
\begin{psmatrix}[mnode=circle,colsep=0.35,rowsep=0.25]
& & & {}\\
& {} & {$c_l$} & {}\\
&  & [mnode=R]{\psframebox{}} &  [mnode=R]{\psframebox{}}\\
 {} & {$c_{l+1}$} & {}\\
 [mnode=R]{\psframebox{}} & {} &  [mnode=R]{\psframebox{}}\\
	& &  {}
\psset{arrows=-, shortput=nab,labelsep={0.08}}
\tiny
\ncline[linestyle=dotted]{1,4}{2,3}
\ncline[linestyle=dotted]{4,2}{6,3}
\ncline{2,2}{2,3}
\ncline{2,3}{3,3}
\ncline{2,3}{2,4}
\ncline{2,3}{3,4}
\ncline{2,3}{4,2}
\ncline{4,1}{4,2}
\ncline{4,2}{4,3}
\ncline{4,2}{5,2}
\ncline{4,2}{5,1}
\ncline{4,2}{5,3}
\end{psmatrix}
\normalsize
\uput[0](-3.9,2.4){\textcolor{red}{$(N_l,E_l)$}}
\pspolygon[linecolor=red,linearc=.4,linewidth=.02](-2.3,2)(-2.3,3.3)(.25,3.3)(.25,2)
\uput[0](-5.5,1){\textcolor{cyan}{$(N_{l+1},E_{l+1})$}}
\pspolygon[linecolor=cyan,linearc=.4,linewidth=.02](-3.1,.3)(-3.1,1.85)(-.65,1.85)(-.65,.3)
\uput[0](.8,1){$T$}
\psline[linewidth=.4pt,linearc=0.5]
{-}(.9,0)(.7,1.6)(.1,3.1)(-1,3.6)(-2.6,1.7)(-2.1,0)
\uput[0](-.1,1){$S$}
\psline[linewidth=.4pt,linearc=0.5]
{-}(0,0)(-.2,1.5)(-.45,2.4)(-1,2.5)(-1.4,2.2)(-1.4,0)
\end{pspicture}
\caption{$c_l \notin S$ and $c_l \in T$.}
\label{figTwoStarsClCl+1-b}
\end{subfigure}
\caption{$p(c_l) = k_l$ and $p(c_{l+1}) = k_{l+1}$ .}
\label{figTwoStarsClCl+1}
\end{figure}
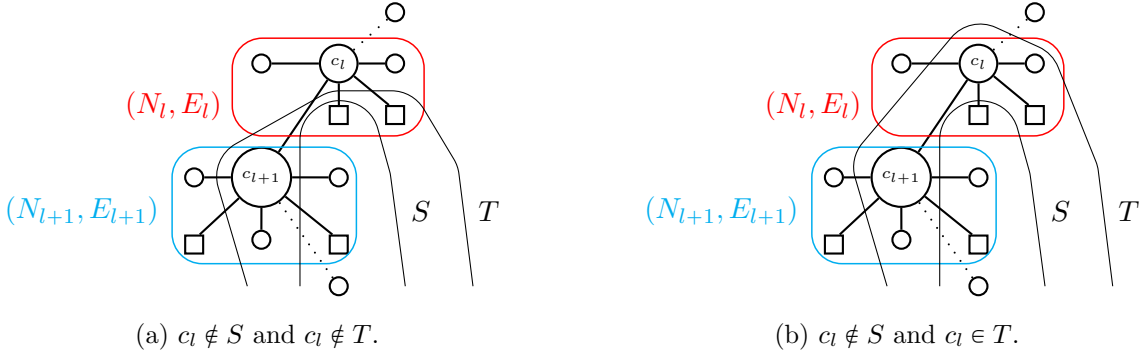
then
\[
v^G(S)=\sum_{i \in S,\, k_{l+1}<p(i) \leq k_{l} } v(\{i\})+v^G(S_{\leq k_{l+1}}).
\]
where $S_{\leq k_{l+1}}$ is the set of elements in $S$ with priority lower than $k_{l+1}$.
Notice that we do not explicit the remaining value since it depends on the connectivity and it is unnecessary.
It follows that
\begin{equation}
\label{vG(S)-vG(S-i)=v(i)-sufficiency}
v^G(S)-v^{G}(S\setminus\{i\})=v(\{i\}), \; \forall i \in S \textrm{ with } p(i)=k.
\end{equation}
At $T$, one obtains the same computation and therefore for every $i\in S$ with $p(i)=k$
\[
v^G(S)-v^{G}(S\setminus\{i\})=v(\{i\})=v^G(T)-v^{G}(T\setminus\{i\}),
\]
and (\ref{eqDefinition(omega,Sigma)-convexity-sufficiency}) is satisfied by considering the weighted sum.
\item
If $c_l \notin S$ but $c_l \in T$ as represented in Figure~\ref{figTwoStarsClCl+1-b},
then (\ref{vG(S)-vG(S-i)=v(i)-sufficiency}) is still satisfied
and we necessarily have $k = k_l$.
Let $\{ T_1, \ldots, T_m \}$  with $m \geq 1$ be the connected components of~$T$.
As $c_l \in T$, we can assume w.l.o.g. $c_l \in T_1$.
Then, we have $p(T_1) = k_l$ and $p(T_j) < k_l$ for all $j$ with $2 \leq j \leq m$.
We get
\begin{equation}
\label{eqSuminInSp(i)=klwTi(vG(T)-vG(T-i))-sufficiency}
\sum_{i \in S,\; p(i) = k_l} \overline{w}^T_i (v^G(T)-v^G(T\setminus\{i\})) =
\sum_{i \in S \cap T_1,\; p(i) = k_l} \overline{w}^{T_1}_i (v^G(T)-v^G(T\setminus\{i\})).
\end{equation}
Let us note that any $i$ in $T_1 \setminus \lbrace c_l \rbrace$ with $p(i) = k_l$ is necessarily a leaf in $(N_l, E_l)$. 
As $v^G(T)=\sum_{j=1}^{n} v(T_j)$
and as $c_l \notin S$,
(\ref{eqSuminInSp(i)=klwTi(vG(T)-vG(T-i))-sufficiency}) implies 
\begin{equation}
\label{eqSuminInSp(i)=klwTi(vG(T)-vG(T-i))-sufficiency-2}
\sum_{i \in S,\; p(i) = k_l} \overline{w}^T_i (v^G(T)-v^G(T\setminus\{i\})) =
\sum_{i \in S \cap T_1,\; p(i) = k_l} \overline{w}^{T_1}_i (v(T_1)-v(T_1\setminus\{i\})).
\end{equation}
Moreover, by~(\ref{vG(S)-vG(S-i)=v(i)-sufficiency}) and as $\overline{\omega}_i^T = \overline{\omega}_i^{T_1}$ for all $i \in S$, we have
\begin{equation}
\label{eqSumiInSCapT1p(i)=klwTAiv(i)-sufficiency}
\sum_{i\in S,\ p(i)=k_l}\overline{\omega}^T_i \left(v^G(S)-v^G(S\setminus \{i\}) \right)=
\sum_{i\in S\cap T_1,\ p(i)=k_l}\overline{\omega}^{T_1}_i v(\{ i \}).
\end{equation}
$(\omega,\Sigma)$-convexity of $v$ applied with $\lbrace i \rbrace \subseteq T_1 $ and $p(i) = k_l $ gives 
\[
\overline{w}^{T_1}_i \left(v(T_1)-v(T_1\setminus\{i\})\right) \geq \overline{w}^{T_1}_i \left( v(\{i\}) - v(\emptyset) \right)
= \overline{w}^{T_1}_i  v(\{i\}).
\]
Then,
using (\ref{eqSuminInSp(i)=klwTi(vG(T)-vG(T-i))-sufficiency-2}) and (\ref{eqSumiInSCapT1p(i)=klwTAiv(i)-sufficiency}),
we get that (\ref{eqDefinition(omega,Sigma)-convexity-sufficiency}) is satisfied.
\item
If $c_l \in S$ and $c_l \in T$, the computation is more complex.
As $c_l \in T$, we necessarily have $k = k_l$.
We first focus on the case where $S$ is composed of a single component.
Any $i \in S$ with $p(i)=k_l$ belongs to the star $(N_l,E_l)$ and therefore satisfies
\[
v^G(S)-v^G(S\setminus \{i\}) =
\begin{cases}
v(S)-v(S\setminus\{i\}) & \textrm{if } i \in S \setminus \lbrace c_l \rbrace,\\
v(S)-\sum_{j \in N_l \cap S} v(\{j\})-v(S \cap N_l^c) & \textrm{if } i = c_l.
\end{cases}
\]
It follows that one can replace all the increments by the first formula except when $i$ is the center.
One obtains
\begin{align}
\label{SumiInSp(i)=kwTi(v(S)-v(S-i))+wTcl(v(S)-v(S-cl))-SumjinNlSv(j)-v(S')}
\sum_{i\in S,\; p(i)=k_l} \overline{\omega}^T_i (v^G(S)-v^G(S\setminus \{i\})) =
&  \sum_{i\in S,\; p(i)=k_l} \overline{\omega}^T_i (v(S)-v(S\setminus \{i\}))\nonumber \\
& + \overline{\omega}^T_{c_l} \left(v(S\setminus{\{c_l\}})-\sum_{j \in N_l\cap S} v(\{j\})-v(S')\right),
\end{align}
where $S':= S \cap N_l^c$.
Similarly, one obtains for $T$
\begin{align}
\label{SumiInSp(i)=kwTi(v(T)-v(T-i))+wTcl(v(T)-v(T-cl))-SumjinNlTv(j)-v(T')}
\sum_{i\in S,\; p(i)=k_l} \overline{\omega}^T_i ( v^G(T)-v^G(T\setminus \{i\})) =
& \sum_{i\in S,\; p(i)=k_l} \overline{\omega}^T_i (v(T)-v(T\setminus \{i\})) \nonumber\\
& +\overline{\omega}^T_{c_l} \left(v(T\setminus{\{c_l\}})-\sum_{j \in N_l \cap T} v(\{j\})-v(T')\right),
\end{align}
where $T':= T \cap N_l^c$.
Let us note that all elements in $T'$ have a priority less than $k_{l+1}$
and all elements in $N_l$ have a priority strictly greater than $k_{l+1}$.
Let $S_1$ be the set of elements in $(S \cap N_l) \setminus \{c_l\}$ with lowest priority.
Let $T_1$ be the set of elements $j$ in $((T \setminus S) \cap N_l) \setminus \{c_l\}$
satisfying $p(j) \leq p(S_1)$ ($T_1$ may be empty).
Let $t_1, \ldots, t_m$ be the elements in $T_1$ ordered by increasing priorities.
Then,
$(\omega,\Sigma)$-convexity of $v$ applied successively to $t_k \subset T' \cup \bigcup_{j=1}^{k}t_j$
with $k \in \lbrace 1, \ldots, m\rbrace$ 
gives
\begin{equation}
\label{eqv(T1UT')-v(T')>=SumjInT1v(j)}
v(T_1 \cup T') - v(T') \geq \sum_{j \in T_1} v(\{j\}).
\end{equation}
Proposition~\ref{weak_superadditivity} applied to $S_1$ and $S' \subset T' \cup T_1$ gives
\begin{equation}
\label{eqv(S1UT1UT')-v(T1UT')>=v(S1US')-v(S')}
v(S_1 \cup T_1 \cup T') - v(T_1 \cup T') \geq v(S_1 \cup S') - v(S').
\end{equation}
(\ref{eqv(T1UT')-v(T')>=SumjInT1v(j)}) and (\ref{eqv(S1UT1UT')-v(T1UT')>=v(S1US')-v(S')}) imply
\begin{equation}
\label{eqv(S1UT1UT')-v(T')-SumjInT1v(j)>=v(S1US')-v(S')}
v(S_1 \cup T_1 \cup T') - v(T') - \sum_{j \in T_1} v(\{j\}) \geq v(S_1 \cup S') - v(S').
\end{equation}
Let now $S_2$ be the set of elements in $(S \cap N_l) \setminus (S_1 \cup \{c_l\})$ with lowest priority.
Let $T_2$ be the set of elements $j$ in $((T\setminus S) \cap N_l) \setminus (T_1 \cup \{c_l\})$
satisfying $p(j) \leq p(S_2)$.
By definition any element $j$ in $T_2$ also satisfies $p(j) > p(S_1)$.
By the same reasoning as before,
we get using successively the $(\omega,\Sigma)$-convexity of $v$
\begin{equation}
\label{eqv(S1UT1UT2UT')-v(S1UT1UT')>=SumjInT2v(j)}
v(S_1 \cup T_1 \cup T_2 \cup T') - v(S_1 \cup T_1 \cup T') \geq \sum_{j \in T_2} v(\{j\}),
\end{equation}
and by Proposition~\ref{weak_superadditivity} applied to $S_2$ and $S_1 \cup S' \subset S_1 \cup T_1 \cup T_2 \cup T'$
\begin{equation}
\label{eqv(S1US2UT1UT2UT')-v(S1UT1UT2UT')>=v(S1US2US')-v(S1US')}
v(S_1 \cup S_2 \cup T_1 \cup T_2 \cup T') - v(S_1 \cup T_1 \cup T_2\cup T') \geq v(S_1 \cup S_2 \cup S') - v(S_1 \cup S').
\end{equation}
(\ref{eqv(S1UT1UT')-v(T')-SumjInT1v(j)>=v(S1US')-v(S')}),
(\ref{eqv(S1UT1UT2UT')-v(S1UT1UT')>=SumjInT2v(j)}) and
(\ref{eqv(S1US2UT1UT2UT')-v(S1UT1UT2UT')>=v(S1US2US')-v(S1US')})
imply
\begin{equation}
\label{eqv(S1US2UT1UT2UT')-v(T')-SumjInT1UT2v(j)>=v(S1US2US')-v(S')}
v(S_1 \cup S_2 \cup T_1 \cup T_2 \cup T') - v(T') - \sum_{j \in T_1 \cup T_2} v(\{j\}) \geq v(S_1 \cup S_2 \cup S') - v(S').
\end{equation}
Repeating the reasoning, we finally get
\begin{equation}
\label{titi}
v(T \setminus \{c_l\}) - \sum_{j \in T\cap N_l} v(\{j\}) - v(T') \geq v(S \setminus \{c_l\}) -\sum_{j \in S\cap N_l} v(\{j\}) - v(S').
\end{equation}


$(\omega,\Sigma)$-convexity of $v$ applied with $S \subset T$ gives
\begin{equation}
\label{eqOmegaSigmaConvexityOfvAppliedWithSAndT-sufficiency}
\sum_{i\in S,\; p(i)=k_l} \overline{\omega}^T_i \left( v(S)-v(S\setminus \{i\}) \right)
\leq
\sum_{i\in S,\; p(i)=k_l} \overline{\omega}^T_i \left( v(T)-v(T\setminus \{i\}) \right).
\end{equation}
Finally, (\ref{SumiInSp(i)=kwTi(v(S)-v(S-i))+wTcl(v(S)-v(S-cl))-SumjinNlSv(j)-v(S')}),
(\ref{SumiInSp(i)=kwTi(v(T)-v(T-i))+wTcl(v(T)-v(T-cl))-SumjinNlTv(j)-v(T')}),
(\ref{titi}), and (\ref{eqOmegaSigmaConvexityOfvAppliedWithSAndT-sufficiency}) imply (\ref{eqDefinition(omega,Sigma)-convexity-sufficiency}).


Let us now show the general case.
Let $\{ S_1,...,S_n \}$
(resp. $\{T_1,...,T_m\}$)
with $n \geq 2$
(resp. $m \geq 1$)
be the partition of $S$
(resp. $T$)
into connected components.
As $c_l \in S$, we can assume w.l.o.g. $c_l \in S_1 \subseteq T_1$.
Then, we have $p(S_1) = p(T_1) = k_l$
and $p(S_j) < k_l$
(resp. $p(T_j) < k_l$)
for all $j$ with $2 \leq j \leq n$
(resp. $2 \leq j \leq m$).
We get
\begin{equation}
\label{eqSuminInSp(i)=klwTi(vG(S)-vG(S-i))-sufficiency}
\sum_{i \in S,\; p(i) = k_l} \overline{w}^T_i (v^G(S)-v^G(S\setminus\{i\})) =
\sum_{i \in S_1,\; p(i) = k_l} \overline{w}^T_i (v^G(S)-v^G(S\setminus\{i\})).
\end{equation}
As $v^G(S)=\sum_{j=1}^{n} v^G(S_j)$
and as $\omega_i^T = \omega_i^{T_1}$ for every $i$ in $S_1$,
(\ref{eqSuminInSp(i)=klwTi(vG(S)-vG(S-i))-sufficiency}) implies
\begin{equation}
\sum_{i \in S, \; p(i)=k_l} \overline{w}^T_i (v^G(S)-v^G(S\setminus\{i\})
=\sum_{i \in S_1, \; p(i)=k_l} \overline{w}^{T_1}_i (v^G(S_1)-v^G(S_1\setminus\{i\}).
\end{equation}
Similarly we have
\begin{equation}
\sum_{i \in S\; p(i)=k_l} \overline{w}^T_i (v^G(T)-v^G(T\setminus\{i\})
=\sum_{i \in S_1\; p(i)=k_l} \overline{w}^{T_1}_i (v^G(T_1)-v^G(T_1\setminus\{i\}).
\end{equation}
We have reduced the problem to the previous situation and therefore established the $(\omega,\Sigma)$-convexity for any pair of sets.
\qedhere
\end{itemize}
\end{proof}

By combining the results of Section~\ref{SectionWeightedShapleyAndweightedConvexity} and 
\ref{Subsection-Necessary-Sufficient-Conditions-Priority-Decreasing-Trees},
we obtain sufficient conditions on the game and on the graph 
to guarantee that the weighted Myerson value belongs to the core.

\appendix

\section{Proof of Proposition \ref{equivalence}}
\label{appendProofOfPropositionEquivalence}

\begin{proof}
Let us check that this definition is indeed coherent with the first definition of the weighted Shapley value.

Since it is defined through the expectation of $v$ under the probability distribution $\mathbb{P}_{\omega,\Sigma}$,
it is clearly a linear mapping of $v$.

Fix the unanimity game on the set $S$, then
\begin{itemize}
\item
A player not in $S$ will never have a positive marginal contribution, hence his value is $0$.
\item
A player $i$ in $S$ only has a positive marginal contribution, if $S \subseteq \{ L \geq i \}$ and $i$ is the pivotal element, hence the first element of $S$ in $L$.
What is the probability of this event?
If $i$ is not in $\overline{S}$, then this probability is $0$, hence he obtains $0$.
If $i$ is in $\overline{S}$, then its probability is exactly $\overline{\omega}_i^S/\overline{\omega}^S$.
The key argument is that the probability of the first element of $S$ to appear to be $i$
conditionally on the fact that an element of $S$ is picked at the current stage is exactly equal to $\frac{\omega_i(S)}{\omega(S)}$.
Let $i\in S$.
For any $t \in \lbrace 1, \ldots, n \rbrace$, we set $L_{\leq t} = \lbrace L(u), u \leq t \rbrace$.
Formally, we have
{
\small
\begin{align*}
& \mathbb{P}_{\omega,\Sigma} \big(i \textrm{ pivot of } S \big)\\
& =\sum_{t=0}^{n-1} \sum_{\substack{T \subset N-S,\\ |T|=t}} \mathbb{P}_{\omega,\Sigma}\Big(L(t+1)=i, L_{\leq t}=T \Big),\\
&=\sum_{t=0}^{n-1} \sum_{\substack{T \subset N-S,\\ |T|=t}} \mathbb{P}_{\omega,\Sigma}\Big(L(t+1) \in S,\ L_{\leq t}=T \Big)
\mathbb{P}_{\omega,\Sigma}\Big(L(t+1)=i|L(t+1) \in S,\ L_{\leq t}=T \Big),\\
&=\sum_{t=0}^{n-1} \sum_{\substack{T \subset N-S,\\ |T|=t}} \mathbb{P}_{\omega,\Sigma}\Big(L(t+1) \in S,\ L_{\leq t}=T \Big)
\frac{\mathbb{P}_{\omega,\Sigma}\Big(\{L(t+1)=i\} \cap \{L(t+1) \in S\} \cap \{L_{\leq t}=T\} \Big)}{\mathbb{P}_{\omega,\Sigma}\Big(\{\{L(t+1) \in S\} \cap \{L_{\leq t}=T\} \Big)},\\
&=\sum_{t=0}^{n-1} \sum_{\substack{T \subset N-S,\\ |T|=t}} \mathbb{P}_{\omega,\Sigma}\Big(L(t+1) \in S,\ L_{\leq t}=T \Big)
\frac{\mathbb{P}_{\omega,\Sigma}\Big(L(t+1)=i| L_{\leq t}=T \Big)}{\mathbb{P}_{\omega,\Sigma}\Big(L(t+1) \in S  |L_{\leq t}=T \Big)},\\
&=\sum_{t=0}^{n-1} \sum_{\substack{T \subset N-S,\\ |T|=t}} \mathbb{P}_{\omega,\Sigma}\Big(L(t+1) \in S,\ L_{\leq t}=T \Big)\left(\frac{\frac{\overline{\omega}^{N-T}_i}{\overline{\omega}^{N-T}}}{\sum_{j\in S}\frac{\overline{\omega}^{N-T}_j}{\overline{\omega}^{N-T}}}\right).
\end{align*}
}
By assumption $T \subset N-S$, hence $N-T$ is a superset of $S$ and therefore $p(N-T)$ is greater or equal to $p(S)$.
We distinguish two different cases, 
\begin{itemize}
\item if $p(N-T)>p(S)$, then $P\Big(L(t) \in S,\ \forall u< t, \  L(u) \notin S \Big)=0,$ since any element of $S$ has still a too low priority.
\item if $p(N-T)=p(S)$, then for every $i\in \overline{S}$, $\overline{\omega}^{N-T}_i=\omega_i$ whereas for $i\in S-\overline{S}$, one has $\overline{\omega}^{N-T}_i=0$.
\end{itemize}
Hence,
\begin{align*}
\mathbb{P}_{\omega,\Sigma}\Big(S \subseteq \{L \geq i \} \Big) &=\sum_{t=0}^{n-1} \sum_{T \subset N-S, |T|=t} \mathbb{P}_{\omega,\Sigma}\Big(L(t+1) \in S,\ \{L(u),u\leq t\}=T \Big)\frac{\omega_i}{\sum_{j\in \overline{S}} \omega_i},\\
&=\frac{\overline{\omega}^S_i}{\overline{\omega}^S} \left(\sum_{t=0}^{n-1}  \sum_{T \subset N-S, |T|=t} \mathbb{P}_{\omega,\Sigma}\Big(L(t+1) \in S,\ \{L(u),u\leq t\}=T \Big)\right),\\
&=\frac{\overline{\omega}^S_i}{\overline{\omega}^S} \left(\sum_{t=0}^{n-1}  \mathbb{P}_{\omega,\Sigma}\Big(L(t+1) \in S,\ \ \forall u\leq t, \ L(u) \notin S\} \Big)\right),\\
&=\frac{\overline{\omega}^S_i}{\overline{\omega}^S}.\qedhere
\end{align*}
\end{itemize}
\end{proof}

\section{Proof of Proposition~\ref{propPsiOmegaiTiInT=PhiOmegavT}}
\label{appendProofOfPropositionRec}

We first establish a formula for the $(\omega, \Sigma)$-weighted Shapley value in terms of the marginal contributions.
\begin{lemma}
The $(\omega, \Sigma)$-weighted Shapley value can be defined
as follows
\[
\Phi_i^\omega (v) = \sum_{\substack{S \subseteq N\\ i \in S}} \gamma_{S,i}^{N,w} (v(S) - v(S \setminus \lbrace i \rbrace))
\]
where
\[
\gamma_{S,i}^{N,w} =
\left \lbrace
\begin{array}{cl}
\sum\limits_{\substack{T \subseteq N :\\ T \supseteq S,\, p(T) = p(S)}} (-1)^{t-s} \frac{\omega_i}{\overline{\omega}^{T}} & \textrm{ if } i \in \overline{S},\\
0 & \textrm{ if } i \in S \setminus \overline{S},
\end{array}
\right.
\]
for all $S \subseteq N$
with $S \not= \emptyset$.
\end{lemma}

\begin{proof}
By definition,
we have
\[
\Phi^{\omega}_{i}(v) =
\sum_{S \subseteq N :\, i \in S } \frac{\overline{\omega}^S_i}{\overline{\omega}^S} \lambda_S(v)=
\sum_{S \subseteq N :\, i \in S }  \sum_{T \subseteq S} (-1)^{s-t} v(T) \frac{\overline{\omega}^S_i}{\overline{\omega}^S}.
\]
We get
\begin{eqnarray}
\Phi_i^\omega (v) & = &
\sum_{S \subseteq N :\, i \in S } \left\lbrack \sum_{T \subseteq S :\, i \in T} (-1)^{s-t} (v(T) - v(T \setminus \lbrace i \rbrace)) \right \rbrack
\frac{\overline{\omega}^S_i}{\overline{\omega}^S} \nonumber\\
& = & \sum_{T \subseteq N :\, i \in T } \left\lbrack \sum_{S \subseteq N :\, S \supseteq T} (-1)^{s-t}
\frac{\overline{\omega}^S_i}{\overline{\omega}^S} \right \rbrack
(v(T) - v(T \setminus \lbrace i \rbrace))
\label{eqSum_TSum_S(-1)^(s-t)(v(T)-v(T-i)}
\end{eqnarray}
For any pair $S,T$ occurring in (\ref{eqSum_TSum_S(-1)^(s-t)(v(T)-v(T-i)})
we have $i \in T \subseteq S$
and thus $p(i) \leq p(T) \leq p(S)$.
As $\overline{\omega}^S_i = \omega_i$
(resp.  $\overline{\omega}^S_i = 0$)
for all $i \in S$ with
$p(i)=p(S)$
(resp. $p(i) \not =p(S)$),
we get the result.
\end{proof}

Let $\Phi^\omega(v^T)$ be the $(\omega, \Sigma)$-weighted Shapley value of the subgame $v^T$.
We have
\begin{equation}
\Phi_i^\omega (v^T) = \sum_{\substack{S \subseteq T\\ i \in S}} \gamma_{S,i}^{T,\omega} (v(S) - v(S \setminus \lbrace i \rbrace)).
\end{equation}

Define the vector $\Psi^{\omega} = (\Psi^{\omega}_{iT})_{i \in T,T\in \mathcal{P}(N)}$ recursively by
\[
\Psi^{\omega}_{iT}=
\frac{\overline{\omega}_i^T}{\overline{\omega}^T} (v(T)-v(T \setminus \{i\}))
+\sum_{j \in T \setminus \lbrace i \rbrace} \frac{\overline{\omega}_j^T}{\overline{\omega}^T} \Psi^{\omega}_{iT\setminus \{j\}},
\]
for all $i \in T$, $T \in \mathcal{P}(N)$
and setting $\Psi^{\omega}_{i\emptyset} = 0$ for all $i \in N$.

We recall Proposition~\ref{propPsiOmegaiTiInT=PhiOmegavT}.
\VectorAndShapleyValue*

In order to prove Proposition~\ref{propPsiOmegaiTiInT=PhiOmegavT},
we first establish the following lemma.

\begin{lemma}
\label{lemSumR(OmegaT-OmegarR)GammaSR=OmegaTGammaST}
For all $S, T \in \mathcal{P}(N)$ with $\emptyset \not= S \subset T$,
and for all $i \in \overline{S}$
\[
\sum_{j \in T \setminus S} \frac{\overline{\omega}_j^T}{\overline{\omega}^T} \gamma_{S,i}^{T \setminus \lbrace j \rbrace, \omega} =
\gamma_{S,i}^{T,\omega}.
\]
\end{lemma}

\begin{proof}
Let us consider $S \subset T \subseteq N$ and $i \in \overline{S}$.
We have 
\begin{eqnarray}
\label{eqSumRr=t-1(OmegaT-OmegaR)GammaSR}
\sum_{j \in T \setminus S} \overline{\omega}_j^T \gamma_{S,i}^{T \setminus \lbrace j \rbrace, \omega}
& = &
\sum_{j \in T \setminus S} \overline{\omega}_j^T
\left \lbrack
\sum\limits_{\substack{R \subseteq T \setminus \lbrace j \rbrace :\\ R \supseteq S,\, p(R) = p(S)}} (-1)^{r-s} \frac{\omega_i}{\overline{\omega}^{R}} 
\right \rbrack
\nonumber\\
& = &
\sum_{\substack{R \subset T :\, R \supseteq S,\\ p(R) = p(S)}}
(-1)^{r-s} \frac{\omega_i}{\overline{\omega}^{R}}
\sum_{j \in T \setminus R}\overline{\omega}_j^T  \nonumber\\
& = &
\sum_{\substack{R \subset T :\, R \supseteq S,\\ p(R) = p(S)}}
(-1)^{r-s} \frac{\omega_i}{\overline{\omega}^{R}}
(\overline{\omega}^T - \overline{\omega}^T_{R}).
\end{eqnarray}
If $p(T) \not= p(S)$,
then any $R \subset T$ with $p(R) = p(S)$ satisfies $\overline{\omega}^T_{R} = 0$
and (\ref{eqSumRr=t-1(OmegaT-OmegaR)GammaSR}) is equivalent to
\[
\sum_{j \in T \setminus S} \overline{\omega}_j^T \gamma_{S,i}^{T \setminus \lbrace j \rbrace, \omega}=
\overline{\omega}^T 
\sum_{\substack{R \subseteq T :\, R \supseteq S,\\ p(R) = p(S)}}
(-1)^{r-s} \frac{\omega_i}{\overline{\omega}^{R}}=
\overline{\omega}^T \gamma_{S,i}^{T,\omega}.
\]
If $p(T) = p(S)$,
then any $R \subseteq T$ with $R \supseteq S$ satisfies $p(R) = p(S)$ and $\overline{\omega}^{R} = \overline{\omega}^T_{R}$
and (\ref{eqSumRr=t-1(OmegaT-OmegaR)GammaSR}) is equivalent to
\begin{eqnarray}
\label{eqSumRr=t-1(OmegaT-OmegaR)GammaSR-2}
\sum_{j \in T \setminus S} \overline{\omega}_j^T \gamma_{S,i}^{T \setminus \lbrace j \rbrace, \omega}
& = &
\overline{\omega}^T 
\sum_{\substack{R \subseteq T :\, R \supseteq S,\\ p(R) = p(S)}}
(-1)^{r-s} \frac{\omega_i}{\overline{\omega}^{R}}
- (-1)^{t-s}\omega_i
- \sum_{R \subset T :\, R \supseteq S} (-1)^{r-s} \omega_i \nonumber\\
& = & 
\overline{\omega}^T \gamma_{S,i}^{T,\omega}
- \omega_i \left((-1)^{t-s}
+ \sum_{k=0}^{t-s-1} C_{t-s}^{k}(-1)^k \right).
\end{eqnarray}
Finally,
as $\sum_{k=0}^{t-s-1} C_{t-s}^{k}(-1)^k = (1-1)^{t-s} - (-1)^{t-s}$,
(\ref{eqSumRr=t-1(OmegaT-OmegaR)GammaSR-2}) implies the result.
\end{proof}

\begin{proof}[Proof of Proposition~\ref{propPsiOmegaiTiInT=PhiOmegavT}]
If $t=1$, then the result is satisfied.
Let us consider $T \in \mathcal{P}(N) $ with $t>1$ and $i \in T$.
We assume $\Psi_S^\omega = \Phi^\omega(v^S)$ for all $S \in \mathcal{P}(N)$ with $s = t-1$.
We get
\begin{eqnarray}
\label{eqPsi_Tw}
\Psi_{iT}^\omega & = & \frac{\overline{\omega}_i^T}{\overline{\omega}^T} (v(T) - v(T \setminus \lbrace i \rbrace))
+ \sum_{j \in T \setminus \lbrace i \rbrace} \frac{\overline{\omega}_j^T}{\overline{\omega}^T} \Phi_i^\omega(v^{T \setminus \lbrace j \rbrace}) \nonumber\\
& = & \frac{\overline{\omega}_i^T}{\overline{\omega}^T} (v(T) - v(T \setminus \lbrace i \rbrace))
+ \sum_{j \in T \setminus \lbrace i \rbrace} \frac{\overline{\omega}_j^T}{\overline{\omega}^T}
\sum_{\substack{S \subseteq T \setminus \lbrace j \rbrace\\ i \in S}}
\gamma_{S,i}^{T \setminus \lbrace j \rbrace, \omega} (v(S) - v(S \setminus \lbrace i \rbrace)) \nonumber \\
& = & \frac{\overline{\omega}_i^T}{\overline{\omega}^T} (v(T) - v(T \setminus \lbrace i \rbrace))
+ \sum_{\substack{S \subset T\\ i \in S}}
\sum_{j \in T \setminus S} \frac{\overline{\omega}_j^T}{\overline{\omega}^T}
\gamma_{S,i}^{T \setminus \lbrace j \rbrace, \omega} (v(S) - v(S \setminus \lbrace i \rbrace)).
\end{eqnarray}
Then,
(\ref{eqPsi_Tw}) and Lemma~\ref{lemSumR(OmegaT-OmegarR)GammaSR=OmegaTGammaST} imply
\begin{eqnarray}
\label{eqPsi_Tw-2}
\Psi_{iT}^\omega & = & \frac{\overline{\omega}_i^T}{\overline{\omega}^T} (v(T) - v(T \setminus \lbrace i \rbrace))
+ \sum_{\substack{S \subset T\\ i \in S}} \gamma_{S,i}^{T,\omega} (v(S) - v(S \setminus \lbrace i \rbrace)).
\end{eqnarray}
As $\gamma_{T,i}^{T,\omega} = \frac{\overline{\omega}_i^T}{\overline{\omega}^T}$,
(\ref{eqPsi_Tw-2}) implies the result.
\end{proof}

\section{Proof of Counter-example \texorpdfstring{$3$}{3}-pan (Example~\ref{Example3-pan-NonValidFor4-path})}
\label{Append-Counter-Example-3-pan-Non-Valid-For-4-path}

We show that Example~\ref{Example3-pan-NonValidFor4-path} described page~\pageref{Example3-pan-NonValidFor4-path}
provides a counter-example if
\[
p(2) \geq p(4),\, p(4) \geq p(1) \text{ and } p(4) \geq p(3).
\]

\noindent
\underline{The game $(N,v)$ is not convex:}

\noindent
The marginal contribution of player $1$ is strictly smaller to $\{3,4\}$ than to $\{4\}$.
\begin{align*}
v(\{1,4\})-v(\{4\})& =X-0=X,\\
v(\{1,3,4\})-v(\{3,4\}) & = (X+Y-1)-Y =X-1.
\end{align*}

\noindent
\underline{The game $(N,v)$ is $(\omega,\Sigma)$-convex:}

\noindent
We need to check the inequalities for any pair of non-empty sets $S$ and $T$ such that $S \subset T$.
\begin{itemize}[leftmargin=*]
\item
Let us first notice that the function $v$ is monotonic since $(Y-1)\geq 0$, $(X-1)\geq 0$ and $\Theta \geq Z\geq Y$.
Hence if $S$ has value $0$, we have
\begin{equation}
\label{eqSumiInSOmegaiT(v(S)-v(S-i))=0-ter}
\sum_{i\in S} \omega_i^T(v(S)-v(S \setminus{i}))=0-0=0.
\end{equation}
Since the function $v$ is monotonic, each marginal contribution has to be positive,
hence the $(\omega, \Sigma)$-convexity inequality is satisfied for any $T$ containing $S$.
(\ref{eqSumiInSOmegaiT(v(S)-v(S-i))=0-ter}) is satisfied if $S$ is a singleton,
and if $S$ is equal to  $\{1,2\}$, $\{1,3\}$, or $\{2,3\}$.
\item
If $S$ is equal to $\{1,4\}$, then
\[
\sum_{i\in S} \omega_i^T(v(S)-v(S \setminus{i})) = \omega_1^T X + \omega_4^T X=(\omega_1^T + \omega_4^T)X.
\]
There are $3$ possible cases for $T$: $\{1,2,4\}$, $\{1,3,4\}$ and $N$.
\begin{itemize}[leftmargin=*]
\item
Assume $T=\{1,2,4\}$,
\begin{align*}
\sum_{i\in S} \omega_i^T(v(T)-v(T \setminus{i})) & = \omega_1^T (X+\alpha_p(Y-1)- \alpha_p Y) + \omega_4^T (X+\alpha_p(Y-1)-0),\\
& = (\omega_1^T + \omega_4^T)X + \alpha_p \omega_4^T Y - \alpha_p (\omega_1^T + \omega_4^T).
\end{align*}

If $\alpha_p = 1$,
then by definition of $Y$ and as $p(4) \geq p(1)$,
we have $\omega_4^T Y - (\omega_1^T + \omega_4^T) \geq 0$ and the $(\omega, \Sigma)$-convexity inequality is satisfied.
If $\alpha_p = 0$, the inequality is satisfied and tight.

\item
Assume $T=\{1,3,4\}$,
\begin{align*}
\sum_{i\in S} \omega_i^T(v(T)-v(T \setminus{i})) & = \omega_1^T (X+Y-1- Y) + \omega_4^T (X+Y-1-0),\\
& = (\omega_1^T + \omega_4^T)X + \omega_4^T Y - (\omega_1^T + \omega_4^T).
\end{align*}
Then we can conclude as in the previous case with $\alpha_p = 1$.

\item
Let us consider $T=N$, then
\begin{align*}
\sum_{i\in S} \omega_i^T (v(T) - v(T \setminus{i})) & =  \omega_1^T (\Theta-Z) + \omega_4^T (\Theta- \alpha_p (X - 1)),\\
& = \omega_1^T (X-1) + \omega_4^T (Z + (1 - \alpha_p)(X-1)).
\end{align*}
If $\alpha_p = 1$,
we get
\begin{align*}
\sum_{i\in S} \omega_i^T (v(T) - v(T \setminus{i}))
& = \omega_1^T (X-1) + \omega_4^T \left(X+2Y \right),\\
& =  (\omega_1^T + \omega_4^T) X + \omega_4^T Y + (\omega_4^T Y-\omega_1^T).
\end{align*}
As $p(4) = p(N)$,
we have $\omega_4^T Y-\omega_1^T =  \omega_4 + \omega_1 - \omega_1^T \geq 0$
and the inequality is satisfied.
If $\alpha_p = 0$, we have
\begin{align*}
\sum_{i\in S} \omega_i^T (v(T) - v(T \setminus{i}))
& \geq  \omega_1^T (X-1) + \omega_4^T Z,\\
& =  \omega_1^T (X-1) + \omega_4^T \left(X+Y+1+\frac{\omega_1}{\omega_2+\omega_3+\omega_4}X \right),\\
& =  (\omega_1^T + \omega_4^T) X + \omega_4^T + (\omega_4^T Y-\omega_1^T) + \frac{\omega_1\omega_4^T}{\omega_2+\omega_3+\omega_4}X.
\end{align*}
If $p(4) < p(N)$ then by assumption we also have $p(1) < p(N)$ and the  $(\omega, \Sigma)$-convexity inequality is trivially satisfied.
Otherwise, we have $\omega_4^T Y-\omega_1^T =  \omega_4 + \omega_1 - \omega_1^T \geq 0$
and the inequality is also satisfied.

\end{itemize}

\item If $S$ is equal to $\{2,4\}$, then
\[
\sum_{i\in S} \omega_i^T (v(S)-v(S \setminus{i})) = \alpha_p (\omega_2^T + \omega_4^T)Y.
\]
There are $3$ possible cases for $T$: $\{1,2,4\}$, $\{2,3,4\}$ and $N$.
\begin{itemize}[leftmargin=*]
\item
Let us assume that $T=\{1,2,4\}$,
\begin{align*}
\sum_{i\in S} \omega_i^T (v(T)-v(T \setminus{i}))
& = \omega_2^T (X+ \alpha_p (Y-1) -X) + \omega_4^T (X+ \alpha_p ( Y-1) -0),\\
& = \alpha_p (\omega_2^T + \omega_4^T)Y + \omega_4^T X - \alpha_p (\omega_2^T + \omega_4^T).
\end{align*}
If $\alpha_p = 0$,
then the inequality is obviously satisfied.
If $\alpha_p = 1$,
then $p(2) = p(4)$
and we have $\omega_4 X - \omega_2 - \omega_4 \geq 0 $ by definition of $X$. 

\item
Let us assume $T=\{2,3,4\}$,
\begin{align*}
\sum_{i\in S} \omega_i^T (v(T)-v(T \setminus{i}))
& = \omega_2^T (Z-Y) + \omega_4^T (Z-0).
\end{align*}
If $\alpha_p = 0$,
then the inequality is obviously satisfied.
If $\alpha_p = 1$,
then $Z = X+2Y$ and the inequality is satisfied.
\item
Finally, let $T=N$, then
\begin{eqnarray}
\sum_{i\in S} \omega_i^T (v(T)-v(T \setminus{i}))
& = & \omega_2^T (\Theta-X-Y+1) + \omega_4^T (\Theta - \alpha_p(X-1)), \nonumber\\
& = & \omega_2^T (Z - Y) + \omega_4^T (Z + (1 - \alpha_p)(X-1)). \nonumber
\end{eqnarray}
If $\alpha_p = 0$,
then the inequality is obviously satisfied.
If $\alpha_p = 1$,
then $Z = X+2Y$ and the inequality is also satisfied.
\end{itemize}
\item If $S$ is equal to $\{3,4\}$, then
\[
\sum_{i\in S} \omega_i^T (v(S)-v(S \setminus{i})) = \omega_3^T Y + \omega_4^T Y = (\omega_3^T + \omega_4^T)Y.
\]
There are $3$ possible cases for $T$: $\{1,3,4\}$, $\{2,3,4\}$ and $N$. 
\begin{itemize}[leftmargin=*]
\item
Let us assume that $T=\{1,3,4\}$,
\begin{align*}
\sum_{i\in S} \omega_i^T (v(T)-v(T \setminus{i}))
& = \omega_3^T (X+Y-1-X) + \omega_4^T (X+Y-1- 0),\\
& = (\omega_3^T + \omega_4^T)Y + \omega_4^T X - (\omega_3^T + \omega_4^T),\\
& \geq (\omega_3^T + \omega_4^T)Y + \omega_3 - \omega_3^T.
\end{align*}
The last inequality is by definition of $X$ and as $p(4) = p(\lbrace 1,3,4 \rbrace)$.
Therefore the $(\omega, \Sigma)$-convexity inequality is satisfied.
\item
Let us now assume that $T=\{2,3,4\}$,
\begin{align*}
\sum_{i\in S} \omega_i^T (v(T)-v(T \setminus{i}))
& = \omega_3^T (Z- \alpha_p Y) + \omega_4^T (Z-0). \nonumber
\end{align*}
If $\alpha_p = 0$,
then the inequality is obviously satisfied.
If $\alpha_p = 1$,
then $Z = X+2Y$ and the inequality is also satisfied.
\item
Finally, let $T=N$, then
\begin{eqnarray*}
\sum_{i\in S} \omega_i^T (v(T)-v(T \setminus{i}))
& = & \omega_3^T (\Theta-X- \alpha_p(Y-1)) + \omega_4^T (\Theta- \alpha_p(X-1)),\\
& = & \omega_3^T (Z-1- \alpha_p(Y-1)) + \omega_4^T (Z + (1- \alpha_p)(X-1)).
\end{eqnarray*}
If $\alpha_p = 0$,
then the inequality is satisfied as $Z \geq Y +1$.
If $\alpha_p = 1$,
then $Z = X+2Y$ and the inequality is also satisfied.
\end{itemize}
\item
Let $S=\{1,2,3\}$ and $T=N$, then
\begin{eqnarray*}
\sum_{i\in S} \omega_i^T (v(S)-v(S \setminus{i}))
& = & \alpha_p (\omega_1^T + \omega_2^T + \omega_3^T) (X-1),
\end{eqnarray*}
and
\begin{eqnarray*}
\sum_{i\in S} \omega_i^T (v(T)-v(T \setminus{i}))
& = & \omega_1^T (\Theta-Z) + \omega_2^T (\Theta-(X+Y-1)) + \omega_3^T (\Theta-(X+ \alpha_p(Y-1))),\\
& = & \omega_1^T (X-1) + \omega_2^T (Z - Y) + \omega_3^T (Z - 1 - \alpha_p(Y-1))).
\end{eqnarray*}
If $\alpha_p = 0$,
then the inequality is satisfied as $Z \geq Y +1$.
If $\alpha_p = 1$,
then $Z = X+2Y$ and the inequality is also satisfied.
\item
Let $S=\{1,2,4\}$ and $T=N$, then
\begin{eqnarray*}
\sum_{i\in S} \omega_i^T (v(S)-v(S \setminus{i}))
& = & \omega_1^T (X- \alpha_p) + \alpha_p \omega_2^T (Y-1) + \omega_4^T (X+ \alpha_p (Y-1)),
\end{eqnarray*}
and
\begin{eqnarray*}
\sum_{i\in S} \omega_i^T (v(T)-v(T \setminus{i}))
& = &\omega_1^T (\Theta-Z) + \omega_2^T (\Theta-(X+Y-1)) + \omega_4^T (\Theta- \alpha_p (X-1)),\\
& = & \omega_1^T (X-1) + \omega_2^T (Z-Y) + \omega_4^T (Z + (1 - \alpha_p) (X-1)).
\end{eqnarray*}
If $\alpha_p = 0$,
then we can conclude as in the case $S = \lbrace  1, 4 \rbrace \subseteq T = N$.
If $\alpha_p = 1$,
then $Z = X+2Y$ and the inequality is also satisfied.
\item
Let $S=\{1,3,4\}$ and $T=N$
\begin{eqnarray*}
\sum_{i\in S} \omega_i^T (v(S)-v(S \setminus{i}))
& = &\omega_1^T (X+Y-1-Y) + \omega_3^T(X+Y-1-X) + \omega_4^T(X+Y-1-0),\\
& = & \omega_1^T (X-1) + \omega_3^T (Y-1) + \omega_4^T (X+Y-1).
\end{eqnarray*}
and
\begin{eqnarray*}
\sum_{i\in S} \omega_i^T (v(T)-v(T \setminus{i}))
& = & \omega_1^T (\Theta-Z) + \omega_3^T(\Theta-X- \alpha_p (Y-1)) + \omega_4^T (\Theta- \alpha_p (X-1)),\\
& = & \omega_1^T (X-1) + \omega_3^T(Z - 1 - \alpha_p (Y-1)) + \omega_4^T (Z + (1 - \alpha_p) (X-1)).
\end{eqnarray*}
If $\alpha_p = 0$,
then the inequality is satisfied as $Z \geq Y$.
If $\alpha_p = 1$,
then $Z = X+2Y$ and the inequality is also satisfied.
\item
Finally, let $S=\{2,3,4\}$ and $T=N$,
\begin{eqnarray*}
\sum_{i\in S} \omega_i^T (v(S)-v(S \setminus{i}))
& = &\omega_2^T (Z-Y) + \omega_3^T (Z - \alpha_p Y) + \omega_4^T (Z-0),
\end{eqnarray*}
and
\begin{eqnarray*}
\sum_{i\in S} \omega_i^T (v(T)-v(T \setminus{i}))
& = &\omega_2^T (\Theta-(X+Y-1)) + \omega_3^T(\Theta-(X+ \alpha_p (Y-1)))\\
& & + \omega_4^T(\Theta- \alpha_p (X-1))\\
& = &\omega_2^T (Z-Y) + \omega_3^T (Z - \alpha_p Y)  + \omega_4^T (Z-0)\\
& & -\omega_3^T (1 - \alpha_p) + \omega_4^T (1 - \alpha_p) (X-1).
\end{eqnarray*}
If $\alpha_p = 1$, then the inequality is tight.
Let us assume $\alpha_p = 0$.
If $p(4) < p(N)$ then by assumption we also have $p(3) < p(N)$
and the inequality is trivially satisfied.
If $p(4) = p(N)$,
then we also have $p(3) = p(N)$ as $\alpha_p = 0$
and by definition of $X$ we get
$-\omega_3^T (1 - \alpha_p) + \omega_4^T (1 - \alpha_p) (X-1) = - \omega_3 + \omega_4 (X - 1) \geq 0$
and the inequality is satisfied.
\end{itemize}

\section{Proof of Counter-example \texorpdfstring{$4$}{4}-path with strict inequality}
\label{proof-Example4-path-strict}
We show that Example~\ref{Example4-path-strict} described page~\pageref{Example4-path-strict} provides a counter-example if
\[
p(2) =p(4) > \max (p(1),p(3)).
\]

\noindent
\underline{The game $(N,v)$ is not convex:}

\noindent
The marginal contribution of player $1$ is strictly smaller to $\{3,4\}$ than to $\{4\}$.
\begin{align*}
v(\{1,4\})-v(\{4\})& =1-0=1,\\
v(\{1,3,4\})-v(\{3,4\}) & = 1-1 = 0.
\end{align*}

\noindent
\underline{The game $(N,v)$ is $(\omega,\Sigma)$-convex:}

\noindent
We need to check the inequalities for any pair of non-empty sets $S$ and $T$ such that $S \subset T$.
\begin{itemize}[leftmargin=*]
\item
Let us first notice that the function $v$ is monotonic.
Hence if $S$ has value $0$, we have
\begin{equation}
\label{eqSumiInSOmegaiT(v(S)-v(S-i))=0-bis}
\sum_{i\in S} \omega_i^T(v(S)-v(S \setminus{i}))=0-0=0.
\end{equation}
Since the function $v$ is monotonic, each marginal contribution is positive
and the $(\omega, \Sigma)$-convexity inequality is satisfied for any $T$ containing $S$.
(\ref{eqSumiInSOmegaiT(v(S)-v(S-i))=0-bis}) is satisfied if $S$ is a singleton,
and if $S$ is equal to  $\{1,2\}$, $\{1,3\}$, $\{2,3\}$, $\{2,4\}$, or $\{1,2,3\}$.
\item
If $S$ is equal to $\{1,4\}$, then
\[
\sum_{i\in S} \omega_i^T(v(S)-v(S \setminus{i})) = \omega_1^T + \omega_4^T.
\]
There are $3$ possible cases for $T$: $\{1,2,4\}$, $\{1,3,4\}$ and $N$.
In any case, we have
\begin{align*}
\sum_{i\in S} \omega_i^T(v(T)-v(T \setminus{i})) = \omega_4^T (1- 0) + \omega_1^T (1-1)= \omega_4^T.
\end{align*}
As $p(4) > p(1)$,
we have $\omega_1^T = 0$ and the $(\omega, \Sigma)$-convexity inequality is satisfied.

\item If $S$ is equal to $\{3,4\}$, then
\[
\sum_{i\in S} \omega_i^T (v(S)-v(S \setminus{i})) = \omega_3^T + \omega_4^T.
\]
There are $3$ possible cases for $T$: $\{1,3,4\}$, $\{2,3,4\}$ and $N$.
In any case we have
\begin{align*}
\sum_{i\in S} \omega_i^T (v(T)-v(T \setminus{i}))  = \omega_4^T (1-0) + \omega_3^T (1-1) = \omega_4^T.
\end{align*}
As $p(4)>p(3)$, we have $\omega_3^T = 0$ and the inequality is satisfied.
\item
Let $S=\{1,2,4\}$ and $T=N$, then
\begin{eqnarray}
\sum_{i\in S} \omega_i^T (v(S)-v(S \setminus{i})) & = &\omega_1^T (1-1) + \omega_2^T (1-1) + \omega_4^T (1-0),\nonumber
\end{eqnarray}
and
\begin{eqnarray}
\sum_{i\in S} \omega_i^T (v(T)-v(T \setminus{i})) & = &\omega_1^T (1-1) + \omega_2^T (1-1) + \omega_4^T (1-0),\nonumber
\end{eqnarray}
and the inequality is satisfied.
\item
Let $S=\{1,3,4\}$ and $T=N$
\begin{eqnarray}
\sum_{i\in S} \omega_i^T (v(S)-v(S \setminus{i})) & = &\omega_1^T (1-1) + \omega_3^T(1-1) + \omega_4^T(1-0),\nonumber
\end{eqnarray}
and
\begin{eqnarray}
\sum_{i\in S} \omega_i^T (v(T)-v(T \setminus{i})) & = & \omega_1^T (1-1) + \omega_3^T(1-1) + \omega_4^T (1-0).\nonumber
\end{eqnarray}
Hence the inequality is satisfied. 
\item
Finally, let $S=\{2,3,4\}$ and $T=N$,
\begin{eqnarray}
\sum_{i\in S} \omega_i^T (v(S)-v(S \setminus{i})) & = &\omega_2^T (1-1) + \omega_3^T (1 - 0) + \omega_4^T (1-0),\nonumber\\
& = & \omega_3^T + \omega_4^T, \nonumber
\end{eqnarray}
and
\begin{eqnarray}
\sum_{i\in S} \omega_i^T (v(T)-v(T \setminus{i})) & = &\omega_2^T (1-1) + \omega_3^T(1-1) + \omega_4^T(1-0)\nonumber\\
&  = & \omega_4^T.\nonumber
\end{eqnarray}
As $p(4)>p(3)$, we have $\omega_3^T = 0$ and the inequality is satisfied.
\end{itemize}


\bibliographystyle{apalike}
\bibliography{biblio}

\end{document}